\documentclass[11pt,a4paper,oneside]{article}%
\usepackage{amsfonts}
\usepackage{amsmath}
\usepackage{makeidx}
\usepackage{graphicx}
\usepackage{amstext}
\usepackage{amssymb}
\usepackage{version}%
\providecommand{\U}[1]{\protect\rule{.1in}{.1in}}
\newtheorem{theorem}{Theorem}
\newtheorem{acknowledgement}[theorem]{Acknowledgement}

\newtheorem{definition}[theorem]{Definition}
\newtheorem{example}[theorem]{Example}

\newtheorem{lemma}[theorem]{Lemma}

\newtheorem{remark}[theorem]{Remark}

\newenvironment{proof}[1][Proof]{\noindent\textbf{#1.} }{\ \rule{0.5em}{0.5em}}
\ifx\pdfoutput\relax\let\pdfoutput=\undefined\fi
\newcount\msipdfoutput
\ifx\pdfoutput\undefined\else
\ifcase\pdfoutput\else
\ifx\paperwidth\undefined\else
\ifdim\paperheight=0pt\relax\else\pdfpageheight\paperheight\fi
\ifdim\paperwidth=0pt\relax\else\pdfpagewidth\paperwidth\fi
\fi\fi\fi
\begin{document}

\title{Angular Spectra for non-Gaussian Isotropic Fields}
\author{Gyorgy Terdik\thanks{The author is very thankfull to an unknown referee for
valuable comments and suggestions.}\\Faculty of Informatics, University of Debrecen, Hungary, \\Terdik.Gyorgy$@$inf.unideb.hu}
\maketitle

\begin{abstract}
Cosmic Microwave Background (CMB) Anisotropies is a subject of intensive
research in several fields of sciences \cite{Hu2002}. In this paper we start a
systematic development of basic notions and theory in statistics according to
the application for CMB. The main result of this paper is the necessary and
sufficient condition for isotropy of a non-Gaussian field in terms of spectra.
Clear formulae for bi-, tri- and polyspectra and bi-, tri-, and higher order
covariances are also given.



\end{abstract}

\section{Introduction}

"Cosmological Principle (first coined by Einstein): the Universe is, in the
large, homogeneous and isotropic" \cite{Bartlett1999}

In the last decade or so there has been some growing interest in studying the
spatial - time data measured on the surface of a sphere. These data includes
cosmic microwave background (CMB) anisotropies \cite{Kogo2006},
\cite{Okamoto2002}, \cite{Adshead2012}, medical imaging \cite{Kakarala1993}
\cite{Kakarala2012}, global and land-based temperature data \cite{Jones1994},
\cite{RaoTerdikRegress06}, gravitational and geomagnetic data, climate model
\cite{North1981}.

One of the problem in focus is the Non-Gaussianity of the observations, which
leads to the investigation of higher order angular spectra called polyspectra,
\cite{Hu2001}, \cite{Bartolo2010}, \cite{Benoit-Levy2012}. Angular
polyspectra, in particular the bispectrum and trispectrum shows to be an
appropriate measure of Gaussianity since for a Gaussian process all higher
(then second order) spectra are zero, \cite{Kamionkowski2011}. An other
important question is the Monte Carlo simulation of non-Gaussian isotropic
maps with a given power spectrum and bispectrum and possible polyspectrum,
\cite{Contaldi2001}, \cite{Rocha2005}.

In this paper some general properties of the angular polyspectra will be given
for isotropic stochastic processes on sphere. Our treatment follows the basic
theory of higher order spectra for non-Gaussian time series,
\cite{Bril_polyspect-65}, \cite{Brill-book-01}, \cite{Rao84a},
\cite{TerdikLN99}. The symmetry relations and the appropriate series expansion
of cumulants might have influence on the estimation of polyspectra as well,
\cite{Smith2012}, \cite{Izumi2012}.

The isotropy assumption implies very particular form for the angular spectra
therefore a delicate question is the construction of isotropic stochastic maps
on the sphere.

\section{Gaussian isotropic fields}

Gaussian isotropic processes on sphere has a long history starting with
\cite{Obukhov1947}, some basic theory and references can be found in
\cite{Yaglom1961}, \cite{Jones1963}, \cite{McLeod1986} \cite{Yadrenko1983}
\cite{Yaglom-book-87}, (\cite{Leonenko1999}). Due to the expansive recent
applications there are new books \cite{Gaetan2010}, \cite{Cressie2011} and
several papers covering a number of problems in general for spatial-time observations.

We consider a stochastic process $X\left(  L\right)  $ on the unit sphere
$\mathbb{S}_{2}$ in $\mathbb{R}^{3}$, where $L=\left(  \vartheta
,\varphi\right)  $, the co-latitude $\vartheta\in\left[  0,\pi\right]  $ and
the longitude $\varphi\in\left[  0,2\pi\right]  $. Let us suppose that
$X\left(  L\right)  $ is continuous (in mean square sense), then it has a
series expansion in terms of spherical harmonics $Y_{\ell}^{m}$%
\begin{equation}
X\left(  L\right)  =\sum_{\ell=0}^{\infty}\sum_{m=-\ell}^{\ell}Z_{\ell}%
^{m}Y_{\ell}^{m}\left(  L\right)  , \label{Series_X_L}%
\end{equation}
where the coefficients are given by
\[
Z_{\ell}^{m}=\int_{\mathbb{S}_{2}}X\left(  L\right)  \overline{Y_{\ell}%
^{m}\left(  L\right)  }\Omega\left(  dL\right)  ,
\]
where $\Omega\left(  dL\right)  =\sin\vartheta d\vartheta d\varphi$ is
Lebesque element of surface area on $\mathbb{S}_{2}$. It is generally
accepted, in time series analysis that $X\left(  L\right)  $ called linear if
$Z_{\ell}^{m}$ are independent and identically distributed ($\mathsf{E}%
X\left(  L\right)  =0$). Now it is straightforward that from the linearity
here follows that the covariance function $\mathcal{C}_{2}\left(  L_{1}%
,L_{2}\right)  =\mathsf{E}X\left(  L_{1}\right)  X\left(  L_{2}\right)  $
depends on the angular distance $\gamma$ between $L_{1}$ and $L_{2}$ only.
Necessarily the covariance function depends on the central angle between
locations and has the form
\begin{equation}
\mathcal{C}_{2}\left(  L_{1},L_{2}\right)  =\mathcal{C}\left(  \cos
\gamma\right)  =\sum_{\ell=0}^{\infty}f_{\ell}\frac{2\ell+1}{4\pi}P_{\ell
}\left(  \cos\gamma\right)  , \label{Cov_X}%
\end{equation}
where $P_{\ell}$ denotes the Legendre polynomial, see \cite{Yadrenko1983},
I.5., and all coefficients $f_{\ell}\geq0$. Note here that the inner product
$L_{1}\cdot L_{2}=\cos\gamma$. In other words $\mathcal{C}_{2}\left(
L_{1},L_{2}\right)  $ is invariant under the group of rotations, i.e.
$X\left(  L\right)  $ is\textit{\ isotropic}. Moreover the assumption of
linearity, i.e. the independence of the triangular array $Z_{\ell}^{m}$ is so
strong that the Gaussianity also follows, \cite{Baldi2007}. In other words the
only linear field on the sphere is the Gaussian one. We shall consider a
linearity which looks weaker although it will be shown to be equivalent to the
whole independence of $Z_{\ell}^{m}$.

\begin{definition}
The field $X\left(  L\right)  $ is linear if the generating array $Z_{\ell
}^{m}$ is uncorrelated and for fixed degree $\ell$, $\left\{  \left.  Z_{\ell
}^{m}\right\vert m=-\ell,-\ell+1\ldots,\ell-1,\ell\right\}  $ are independent.
\end{definition}

This concept of linearity corresponds to the physical approach to the theory
of angular momentum. One may consider linearity in some really weaker sense,
namely the rows $\left\{  Z_{\ell}^{m}\right\}  _{m=-\ell}^{\ell}$ of the
array $Z_{\ell}^{m}$, are independent but the variables $Z_{\ell}^{m}$ inside
the row $\ell$ are uncorrelated.

The convergence of the series $\sum_{\ell=0}^{\infty}f_{\ell}\frac{2\ell
+1}{4\pi}$ is equivalent to the continuity of $\mathcal{C}\left(
\cdot\right)  $ on $\left[  -1,1\right]  $. The superposition (\ref{Cov_X})
corresponds to the superposition of a covariance function on the real line in
terms of its spectrum and the orthogonal system $\left\{  \exp\left(
i2\pi\lambda k\right)  ,\quad k=0,1,2\ldots\right\}  $, hence we treat $\ell$
as the frequency and $\widehat{\mathcal{C}}_{\ell}=f_{\ell}$ is the value of
the spectrum at $\ell$. Since $P_{\ell}\left(  \cos0\right)  =1$, the variance
$\mathsf{E}X\left(  L\right)  ^{2}$ is broken up into the sum of spectra
therefore we have the analysis of variance. The inversion%
\[
\widehat{\mathcal{C}}_{\ell}=\int_{\mathbb{S}_{2}}\mathcal{C}\left(
\cos\gamma\right)  P_{\ell}\left(  \cos\gamma\right)  \Omega\left(  dL\right)
,
\]
is also valid. The stochastic process $X\left(  L\right)  $ itself has
spectral representation (\ref{Series_X_L}) also in terms of spherical
harmonics $Y_{\ell}^{m}$ with complex values, detailed account for spherical
harmonics $Y_{\ell}^{m}$ is found in \cite{Varshalovich1988},
\cite{SteinWeiss}. For a given $X\left(  L\right)  $ we have the inversion
\[
Z_{\ell}^{m}=\int_{\mathbb{S}_{2}}X\left(  L\right)  \overline{Y_{\ell}%
^{m}\left(  L\right)  }\Omega\left(  dL\right)  .
\]
The orthogonal random 'measure' $Z_{\ell}^{m}$ is a triangular array for each
fix $\ell;$ $m=-\ell,-\ell+1\ldots,\ell-1,\ell$, i.e. rows contain $2\ell+1$,
i.i.d Gaussian random variables, $\mathsf{E}Z_{\ell}^{m}=0,\;\mathsf{E}%
Z_{\ell}^{m}Z_{k}^{n\ast}=f_{\ell}\delta_{\ell,k}\delta_{m,n}$. The reason
that $f_{\ell}$ does not depend on $m$, is the Funk-Hecke formula
(\ref{Formula_Funk-Hecke}), see \cite{Yadrenko1983}, I.5., \cite{Jones1963}.

The $\mathcal{C}_{2}\left(  L_{1}\cdot L_{2}\right)  $ is strictly positive
definite if all $f_{\ell}$ is $\geq0$, and only finitely many are zero
(\cite{Schreiner1997}, \cite{Schoenberg1942} ).

\begin{example}
Mat\'{e}rn Class of Covariance Functions \cite{Matern2nded.1986},
\cite{Schabenberger2005} Suppose one has a covariance function on the real
line then its restriction to $\left[  0,\pi\right]  $ gives a covariance
function $\mathcal{C}_{0}\left(  \gamma\right)  ,$ on the sphere
$\mathbb{S}_{2}$, $\mathcal{C}_{0}\left(  \gamma\right)  =\mathcal{C}\left(
\cos\gamma\right)  $, for instance
\begin{align*}
\mathcal{C}_{0}\left(  \gamma\right)   &  =\sigma^{2}\frac{1}{\Gamma\left(
\nu\right)  }\left(  \frac{\vartheta\gamma}{2}\right)  ^{\nu}2K_{\nu}\left(
\vartheta\gamma\right)  ,\;\vartheta>0,\;\nu>0,\\
f_{\ell}  &  =\int_{0}^{\pi}\mathcal{C}_{0}\left(  \gamma\right)  \frac
{2\ell+1}{4\pi}P_{\ell}\left(  \cos\gamma\right)  d\gamma.
\end{align*}
Where $K_{\nu}$ is the modified Bessel function of the second kind. Here the
smoothness parameter $\nu$ controls the continuity, as well as $\vartheta$
controls the regularity \cite{Gaetan2010}, \cite{Sherman2011}. Note $K_{\nu
}\left(  r\right)  \sim\Gamma\left(  \nu\right)  \left(  r/2\right)  ^{-\nu
}/2$ if $r\rightarrow0$, and $\nu>0$, hence $\mathcal{C}_{0}\left(
\gamma\right)  /\sigma^{2}$ is a correlation function.
\end{example}

\begin{example}
The generating function of Legendre polynomial $P_{\ell}$ is
\[
\sum_{\ell=0}^{\infty}P_{\ell}\left(  x\right)  z^{\ell}=\left(
1-2xz+z^{2}\right)  ^{-1/2},\;x\in\left(  -1,1\right)  ,\;\left\vert
z\right\vert <1
\]
see (\ref{Funct_Gen_Pl1}). Let $0<z<1$, hence
\[
\mathcal{C}\left(  \cos\gamma\right)  =\frac{1}{\sqrt{1-2z\cos\gamma+z^{2}}},
\]
is a covariance function on $\cos\gamma\in\left[  -1,1\right]  $, with
$f_{\ell}=\frac{4\pi}{2\ell+1}z^{\ell}$. Similarly for any $n>2$ we have a
covariance function
\[
\mathcal{C}\left(  \cos\gamma\right)  =\frac{1}{\left(  1-2z\cos\gamma
+z^{2}\right)  ^{\left(  n-2\right)  /2}},
\]
if $0<z<1$.
\end{example}

\begin{example}
From the well known series
\[
\sum_{\ell=0}^{\infty}P_{\ell}\left(  x\right)  \frac{a^{\ell}}{2\ell+1}%
=\frac{1-a^{2}}{\left(  1-2xa+a^{2}\right)  ^{3/2}},\;x\in\left(  -1,1\right)
,\;\left\vert a\right\vert <1
\]
see \cite{Yadrenko1983}, for any $0<a<1$, we find the covariance function
\[
\mathcal{C}\left(  \cos\gamma\right)  =\frac{1-a^{2}}{\left(  1-2a\cos
\gamma+a^{2}\right)  ^{3/2}},
\]
on $\left[  -1,1\right]  $ and $f_{\ell}=4\pi a^{\ell}$. Similarly for
$n\geq2$, we have
\[
\mathcal{C}\left(  \cos\gamma\right)  =\frac{1-a^{2}}{\left(  1-2a\cos
\gamma+a^{2}\right)  ^{n/2}}.
\]
Moreover the series
\[
\sum_{\ell=1}^{\infty}P_{\ell}\left(  x\right)  \frac{t^{\ell}}{\ell}=\ln
\frac{2}{1-tx+\sqrt{1-2xt+t^{2}}},\;\;0<t<1,\;\left\vert x\right\vert <1
\]
see \cite{Prudnikov14}, 5.10.1.4., gives a covariance function as well.
\end{example}

\begin{example}
It is known, \cite{Breitenberger1963}, \cite{Weaver2001}, that the probability
density on the sphere
\[
\mathcal{C}\left(  \cos\gamma\right)  =\frac{\kappa}{4\pi\sinh\left(
\kappa\right)  }\exp\left(  \kappa\cos\gamma\right)  ,\;\kappa>0
\]
has series expansion
\[
\mathcal{C}\left(  \cos\gamma\right)  =\frac{\kappa}{\sinh\left(
\kappa\right)  }\sum_{\ell=0}^{\infty}\frac{2\ell+1}{4\pi}\sqrt{\frac{\pi
}{2\kappa}}I_{\ell+1/2}\left(  \kappa\right)  P_{\ell}\left(  \cos
\gamma\right)  ,
\]
where $\sqrt{\frac{\pi}{2\kappa}}I_{\ell+1/2}\left(  \kappa\right)  $ is the
modified spherical Bessel function of the first kind, \cite{Abramowit12},
also
\[
\left(  2\ell+1\right)  \sqrt{\frac{\pi}{2\kappa}}I_{\ell+1/2}\left(
\kappa\right)  \sim\frac{\kappa^{\ell}}{\left(  2\ell-1\right)  !!}.
\]
Hence $\mathcal{C}$ is a covariance function with spectrum
\begin{align*}
f_{\ell}  &  =\frac{\kappa}{\sinh\left(  \kappa\right)  }\sqrt{\frac{\pi
}{2\kappa}}I_{\ell+1/2}\left(  \kappa\right) \\
&  =\frac{I_{\ell+1/2}\left(  \kappa\right)  }{I_{1/2}\left(  \kappa\right)
}.
\end{align*}
Note $\sinh\left(  \kappa\right)  /\kappa=\sqrt{\frac{\pi}{2\kappa}}%
I_{1/2}\left(  \kappa\right)  $, and
\[
\exp\left(  a\cos\vartheta\right)  =\sum_{\ell=0}^{\infty}\left(
2\ell+1\right)  \sqrt{\frac{\pi}{2a}}I_{\ell+1/2}\left(  a\right)  P_{\ell
}\left(  \cos\vartheta\right)
\]
see \cite{Abramowit12}, \cite{Prudnikov14}.
\end{example}

\begin{example}
For $\kappa>0$, $\gamma\in\left[  0,\pi\right]  $ we have
\begin{align*}
\sum_{\ell=0}^{\infty}\frac{\kappa^{\ell}}{\ell!}P_{\ell}\left(  \cos
\gamma\right)   &  =\exp\left(  \kappa\cos\gamma\right)  J_{0}\left(
\kappa\sin\gamma\right)  ,\\
f_{\ell}  &  =\frac{\kappa^{\ell}}{\ell!}\frac{4\pi}{2\ell+1},
\end{align*}
see \cite{Prudnikov14}, hence $\exp\left(  \kappa\cos\gamma\right)
J_{0}\left(  \kappa\sin\gamma\right)  /\exp\left(  \kappa\right)  $ is a
correlation function, where $J_{0}$ denotes Bessel function of the first kind.
\end{example}

\begin{example}
We have a covariance function of the form
\begin{align*}
\mathcal{C}\left(  \cos\gamma\right)   &  =I_{0}\left(  2\kappa\sqrt
{\cos\gamma-1}\right)  I_{0}\left(  2\kappa\sqrt{\cos\gamma+1}\right) \\
&  =\sum_{\ell=0}^{\infty}\frac{\kappa^{\ell}}{\left(  \ell!\right)  ^{2}%
}P_{\ell}\left(  \cos\gamma\right)  ,\\
f_{\ell}  &  =\frac{\kappa^{\ell}}{\left(  \ell!\right)  ^{2}}\frac{4\pi
}{2\ell+1}%
\end{align*}
see \cite{Prudnikov14}. Note that $I_{0}\left(  x\right)  $ is a function of
$x^{2}/4$ see \cite{Abramowit12}{, } hence there is no danger of complex
values. $I_{0}\left(  0\right)  =1$, $I_{0}$ denotes modified Bessel function
of the first kind.
\end{example}

\begin{example}
\textbf{Poisson formula. }For a homogeneous isotropic field on $\mathbb{R}%
^{3}$ we have the spectral representation
\[
\mathcal{C}\left(  r\right)  =\int_{0}^{\infty}j_{0}\left(  \lambda r\right)
\Phi\left(  d\lambda\right)  ,
\]
of a covariance function with spectral measure $\Phi\left(  d\lambda\right)
$, see \cite{Yadrenko1983}, where $j_{0}$ is the Spherical Bessel function of
the first kind, see \cite{Abramowit12}. If we consider two locations $L_{1}$
and $L_{2}$ on the sphere $\mathbb{S}_{2}$ with angle $\gamma\in\left[
0,\pi\right]  $, then the distance $r=\left\Vert L_{1}-L_{2}\right\Vert $
between them in term of the angle is $2\sin\left(  \gamma/2\right)  $, and
$L_{1}\cdot L_{2}=\cos\gamma$. Hence we have the covariance function on sphere
$\mathcal{C}_{0}\left(  \cos\gamma\right)  =\mathcal{C}\left(  2\sin\left(
\gamma/2\right)  \right)  .$ $\mathcal{C}_{0}\left(  \cos\gamma\right)  $
defines an isotropic field on the sphere $\mathbb{S}_{2}$ with spectrum
\begin{equation}
f_{\ell}=2\pi^{2}\int_{0}^{\infty}J_{\ell+1/2}^{2}\left(  \lambda\right)
\frac{1}{\lambda}\Phi\left(  d\lambda\right)  , \label{PoissonForm}%
\end{equation}
where $J_{\ell+1/2}$ denotes the Bessel function of the first kind, see
\cite{Abramowit12}.
\end{example}

\begin{example}
\textbf{Lapalce-Beltrami model on }$\mathbb{S}_{2}$. Consider the homogeneous
isotropic field $X$ on $\mathbb{R}^{3}$ according to the equation
\[
\left(  \bigtriangleup-c^{2}\right)  X=\partial W,
\]
where $\bigtriangleup=\frac{\partial^{2}}{\partial x_{1}^{2}}+\frac
{\partial^{2}}{\partial x_{2}^{2}}+\frac{\partial^{2}}{\partial x_{3}^{2}}$,
denotes the Laplace operator on $\mathbb{R}^{3}$. Its spectrum, see
(\cite{Yadrenko1983}), is
\[
S\left(  \lambda\right)  =\frac{2}{\left(  2\pi\right)  ^{2}}\frac{\lambda
^{2}}{\left(  \lambda^{2}+c^{2}\right)  ^{2}},\quad\lambda^{2}=\left\Vert
\left(  \lambda_{1},\lambda_{2},\lambda_{3}\right)  \right\Vert ^{2},
\]
with covariance of Mat\'{e}rn Class
\[
\mathcal{C}\left(  r\right)  =\frac{1}{\left(  2\pi\right)  ^{3/2}}%
\frac{\left(  cr\right)  ^{1/2}K_{1/2}\left(  cr\right)  }{2c},
\]
where $K_{1/2}$ is the modified Bessel (Hankel) function, see
\cite{Abramowit12}. \newline Now the according to the Lapalce-Beltrami
operator
\[
\bigtriangleup_{B}=\frac{1}{\sin\vartheta}\frac{\partial}{\partial\vartheta
}\left(  \sin\vartheta\frac{\partial}{\partial\vartheta}\right)  +\frac
{1}{\sin^{2}\vartheta}\frac{\partial^{2}}{\partial\varphi^{2}},
\]
we consider the stochastic model
\[
\left(  \bigtriangleup_{B}-c^{2}\right)  X_{B}=\partial W_{B},
\]
on sphere. The covariance function $\mathcal{C}_{0}$ of $X_{B}$ is the
restriction of the covariance function $\mathcal{C}$ of $X$ on sphere and
$\mathcal{C}_{0}\left(  \cos\gamma\right)  =\mathcal{C}\left(  2\sin\left(
\gamma/2\right)  \right)  $, i.e.
\[
\mathcal{C}_{0}\left(  \cos\gamma\right)  =\frac{1}{\left(  2\pi\right)
^{3/2}}\sqrt{\frac{\sin\left(  \gamma/2\right)  }{2c}}K_{1/2}\left(
2c\sin\left(  \gamma/2\right)  \right)  .
\]
We apply the Poisson formula (\ref{PoissonForm}) when $\Phi\left(
d\lambda\right)  =S\left(  \lambda\right)  d\lambda$, and we obtain the
spectrum for $X_{B}$
\begin{align*}
f_{\ell} &  =2\pi^{2}\int_{0}^{\infty}J_{\ell+1/2}^{2}\left(  \lambda\right)
\frac{1}{\lambda}\frac{2}{\left(  2\pi\right)  ^{2}}\frac{\lambda^{2}}{\left(
\lambda^{2}+c^{2}\right)  ^{2}}d\lambda\\
&  =\int_{0}^{\infty}J_{\ell+1/2}^{2}\left(  \lambda\right)  \frac{\lambda
}{\left(  \lambda^{2}+c^{2}\right)  ^{2}}d\lambda,\\
&  =\frac{1}{\left(  \ell\left(  \ell+1\right)  +c^{2}\right)  ^{2}}.
\end{align*}

\end{example}

\section{Non-Gaussian isotropy and the angular spectrum}

In a physical phenomenon the isotropic property is treated as a principle. It
means that there is no reason to make difference between directions. The
corresponding property of a stochastic field is that the finite dimensional
distributions remain unchanged after rotating the space.

\subsection{Isotropy on sphere}

\begin{definition}
A stochastic field $X\left(  L\right)  $ on the unit sphere $\mathbb{S}_{2}$
is isotropic (in strict sense) if all finite dimensional distributions of
$\left\{  X\left(  L\right)  ,L\in\mathbb{S}_{2}\right\}  $ are invariant
under the rotation $g$ for every $g\in SO\left(  3\right)  $.
\end{definition}

From now on we do not assume Gaussianity. But, for the simplicity, we suppose
the existence of moments and that those determine the distribution as well.
The general form of a mean square continuous field is given in terms of
spherical harmonics
\begin{equation}
X\left(  L\right)  =\sum_{\ell=0}^{\infty}\sum_{m=-\ell}^{\ell}Z_{\ell}%
^{m}Y_{\ell}^{m}\left(  L\right)  , \label{X_general}%
\end{equation}
where $\left\{  Z_{\ell}^{m};\ell=0,1,\ldots;\text{ }m=-\ell,1-\ell
,\ldots,-1,0,1,\ldots,\ell-1,\ell\right\}  $ is a triangular array, in the row
$\ell$ we have $2\ell+1$ elements. Notice that $Y_{0}^{0}=1$, put $Z_{0}%
^{0}=\mu,$ otherwise $\mathsf{E}Z_{\ell}^{m}=0,$ therefore $\mathsf{E}X\left(
L\right)  =\mu,$ and the convergence is valid in mean square sense
\cite{Leonenko1999}. This development follows also from the stochastic
Peter-Weyl Theorem, see \cite{Marinucci2011} for details.

In case the $m^{th}$ order cumulants of $X\ $are invariant under the rotation
$g$ for every $g\in SO\left(  3\right)  $ then it will be called
\textit{isotropic in }$m^{th}$ order. Naturally a strictly isotropic field
with $m^{th}$ order moments is isotropic in $m^{th}$ order.

We list here some properties which follows from isotropy, see \cite{Baldi2007}
for details,

\begin{enumerate}
\item The distribution of both $\operatorname*{Re}Z_{\ell}^{m}\overset{d}{=}%
\operatorname*{Im}Z_{\ell}^{m}$ and $\operatorname*{Re}Z_{\ell}^{m}%
/\operatorname*{Im}Z_{\ell}^{m}\ $are \textit{Cauchy.}

\item Marginal distribution of both $\operatorname*{Re}Z_{\ell}^{m}$ and
$\operatorname*{Im}Z_{\ell}^{m}$ are always symmetric

\item For fixed $\ell$, $Z_{\ell}^{m}$, $m=0,1,\ldots,\ell-1,\ell$ are
independent iff they are Gaussian.
\end{enumerate}

Now, let us consider a rotation $g\in SO\left(  3\right)  $, it is known that
the spherical harmonics $Y_{\ell}^{m}$ at the rotated location are given in
terms of the Wigner D-matrix, see \ref{Wigner_D_matrix} Appendix
\ref{App_Spher_Harmonics}, more precisely
\[
\Lambda\left(  g\right)  Y_{\ell}^{m}\left(  L\right)  =\sum_{k=-\ell}^{\ell
}D_{k,m}^{\left(  \ell\right)  }\left(  g\right)  Y_{\ell}^{k}\left(
L\right)
\]
where $\Lambda\left(  g\right)  $ denotes the operator according to the
rotation $g$, $\Lambda\left(  g\right)  Y_{\ell}^{k}\left(  L\right)
=Y_{\ell}^{k}\left(  g^{-1}L\right)  $. Hence the rotated field has the
following form%
\begin{align*}
\Lambda\left(  g\right)  X\left(  L\right)   &  =\sum_{\ell=0}^{\infty}%
\sum_{m=-\ell}^{\ell}Z_{\ell}^{m}\sum_{k=-\ell}^{\ell}D_{k,m}^{\left(
\ell\right)  }\left(  g\right)  Y_{\ell}^{k}\left(  L\right) \\
&  =\sum_{\ell=0}^{\infty}\sum_{k=-\ell}^{\ell}\sum_{m=-\ell}^{\ell}%
D_{k,m}^{\left(  \ell\right)  }\left(  g\right)  Z_{\ell}^{m}Y_{\ell}%
^{k}\left(  L\right)  .
\end{align*}

The isotropy assumption is equivalent to that the distribution of the variable%
\begin{equation}
\mathcal{Z}_{\ell}^{k}\left(  g\right)  =\sum_{m=-\ell}^{\ell}D_{k,m}^{\left(
\ell\right)  }\left(  g\right)  Z_{\ell}^{m}, \label{Z_rotated}%
\end{equation}
is the same as the one of $Z_{\ell}^{k}$. This statement will be used
frequently below.

In matrix form
\[
\mathcal{Z}_{\ell}\left(  g\right)  =D^{\left(  \ell\right)  }\left(
g\right)  Z_{\ell},
\]
where $Z_{\ell}=\left(  Z_{\ell}^{-\ell},Z_{\ell}^{-\ell+1},\ldots,Z_{\ell
}^{\ell}\right)  ^{\top}$, $D^{\left(  \ell\right)  }=\left(  D_{k,m}^{\left(
\ell\right)  }\right)  _{-\ell,-\ell}^{\ell,\ell}$ (the dependence on $g$ will
be omitted unless it is necessary). This equation provides an equation for the
cumulant functions $\Phi_{Z}^{\ell}$ (log of the characteristic function) as
well, in case of isotropy, for each rotation $g$ we have
\[
\Phi_{Z}^{\ell}\left(  \underline{\omega}_{\ell}\right)  =\Phi_{Z}^{\ell
}\left(  \underline{\omega}_{\ell}D^{\left(  \ell\right)  }\left(  g\right)
\right)  ,
\]
where $\underline{\omega}_{\ell}=\left(  \omega_{-\ell},\omega_{-\ell
+1},\ldots,\omega_{\ell-1},\omega_{\ell}\right)  $. Hence the distribution of
$Z_{\ell}$ should be rotationally invariant on $\mathbb{R}^{2\ell+1}$. For
instance the covariance matrix $\mathcal{C}_{Z}\left(  \ell_{1},\ell
_{2}\right)  =\operatorname*{Cov}\left(  Z_{\ell_{1}},Z_{\ell_{2}}\right)  $
commutes with $D^{\left(  \ell_{n}\right)  }$, since $D^{\left(  \ell
_{n}\right)  }$ is unitary and $\mathcal{C}_{Z}\left(  \ell_{1},\ell
_{2}\right)  =D^{\left(  \ell_{1}\right)  \ast}\mathcal{C}_{Z}\left(  \ell
_{1},\ell_{2}\right)  D^{\left(  \ell_{2}\right)  }$. If $\ell_{1}=\ell
_{2}=\ell$, the only matrix which commutes with $D^{\left(  \ell\right)
}\left(  g\right)  $, for any $g\in SO\left(  3\right)  $ is constant times
unit matrix.

Assume for a moment that $\operatorname*{Cum}\nolimits_{2}\left(  Z_{\ell_{1}%
}^{k},Z_{\ell_{2}}^{m}\right)  $ does depend on $\ell_{1},k,\ell_{2},m$, then
we show below that from the isotropy $\operatorname*{Cum}\nolimits_{2}\left(
\mathcal{Z}_{\ell_{1}}^{k},\mathcal{Z}_{\ell_{2}}^{m\ast}\right)
=\delta_{\ell_{1},\ell_{2}}\delta_{k,m}f_{\ell_{1}}$, follows. Hence we should
restrict ourselves on an uncorrelated generating array $Z_{\ell}^{m},$ with
$\;\mathsf{E}Z_{\ell}^{m}Z_{k}^{n\ast}=\delta_{\ell,k}\delta_{m,n}\sigma
_{\ell,m}^{2}$. Since $\sigma_{\ell,m}^{2}$ does not depend on $m$ we denote
it $f_{\ell}$, moreover we shall consider the field $X$ in the following form%
\[
X\left(  L\right)  =\sum_{\ell=0}^{\infty}\sum_{m=-\ell}^{\ell}Z_{\ell}%
^{m}Y_{\ell}^{m}\left(  L\right)  ,
\]
where $\mathsf{E}Z_{\ell}^{m}Z_{k}^{n\ast}=\delta_{\ell,k}\delta_{m,n}f_{\ell
}$. In other words it is seen that the second and higher order structure of
the generating process $Z_{\ell}^{m}$ inside the same degree $\ell$ are
hiding, they are not identifiable.

According to the angular momentum of degree $\ell$ we define the field
\begin{equation}
u_{\ell}\left(  L\right)  =\sum_{m=-\ell}^{\ell}Z_{\ell}^{m}Y_{\ell}%
^{m}\left(  L\right)  . \label{FieldAngMom}%
\end{equation}
We shall be interested in the invariance of the distribution of $u_{\ell
}\left(  L\right)  $ under rotations as well. If the location is fixed then
$u_{\ell}\left(  L\right)  $ is connected to the distribution of
$\mathcal{Z}_{\ell}^{k}$ directly, since the Condon and Shortley phase
convention $\sqrt{\frac{2\ell+1}{4\pi}}D_{m,0}^{\left(  \ell\right)  }\left(
g\right)  =Y_{\ell}^{m\ast}\left(  \vartheta,\varphi\right)  $, where the
rotation $g$ given in Euler coordinates $\left(  \gamma,\vartheta
,\varphi\right)  ,$ provides
\[
u_{\ell}\left(  L\right)  =\sqrt{\frac{2\ell+1}{4\pi}}\sum_{m=-\ell}^{\ell
}D_{0,m}^{\left(  \ell\right)  \ast}\left(  \gamma,\vartheta,\varphi\right)
Z_{\ell}^{m},
\]
where the notation $\ast$ defined as the transpose and conjugate for a matrix
and just conjugate for a scalar. Observe that $\gamma$ here is arbitrary. Our
main interest is the comparison of the finite dimensional distributions of the
field $u_{\ell}\left(  L\right)  $ to the rotated one. In case a rotation
carries the location $L$ to the North pole $N$, $u_{\ell}$ simplifies
\[
u_{\ell}\left(  N\right)  =\sqrt{\frac{2\ell+1}{4\pi}}Z_{\ell}^{0}.
\]
Rewrite $X$ into the form%
\[
X\left(  L\right)  =\sum_{\ell=0}^{\infty}u_{\ell}\left(  L\right)  .
\]
We consider a real valued $X\left(  L\right)  $, therefore $Z_{\ell}^{m\ast
}\overset{d}{=}\left(  -1\right)  ^{m}Z_{\ell}^{-m}$, since $Y_{\ell}%
^{m}\left(  L\right)  ^{\ast}=\left(  -1\right)  ^{m}Y_{\ell}^{-m}\left(
L\right)  $. Moreover if we reflect the location to the center then from
$Y_{\ell}^{m}\left(  -L\right)  =\left(  -1\right)  ^{\ell}Y_{\ell}^{m}\left(
L\right)  $ follows
\begin{equation}
X\left(  -L\right)  =\sum_{\ell=0}^{\infty}\sum_{m=-\ell}^{\ell}\left(
-1\right)  ^{\ell}Z_{\ell}^{m}Y_{\ell}^{m}\left(  L\right)  . \label{Parity}%
\end{equation}
Therefore we have the following

\begin{remark}
\label{Rem_parity}Since the parity $u_{\ell}\left(  -L\right)  =\left(
-1\right)  ^{\ell}u_{\ell}\left(  L\right)  $ is valid, it implies that under
assumption of isotropy
\[
\operatorname*{Cum}\nolimits_{p}\left(  u_{\ell_{1}}\left(  L_{1}\right)
,\ldots,u_{\ell_{p}}\left(  L_{p}\right)  \right)  =\left(  -1\right)
^{\ell_{1}+\ell_{2}+\cdots+\ell_{p}}\operatorname*{Cum}\nolimits_{p}\left(
u_{\ell_{1}}\left(  L_{1}\right)  ,\ldots,u_{\ell_{p}}\left(  L_{p}\right)
\right)  ,
\]
hence for $p>2$ the sum $\mathcal{L}_{p}=\ell_{1}+\ell_{2}+\cdots+\ell_{p}$
must be even.
\end{remark}

\subsubsection{Second order Isotropy}

Consider the second order cumulants of $\mathcal{Z}_{\ell}^{k}$, defined in
(\ref{Z_rotated}),
\begin{align*}
\operatorname*{Cum}\nolimits_{2}\left(  \mathcal{Z}_{\ell_{1}}^{k}%
,\mathcal{Z}_{\ell_{2}}^{m}\right)   &  =\sum_{p,q=-\ell_{1},\ell_{2}}%
^{\ell_{1},\ell_{2}}D_{k,p}^{\left(  \ell_{1}\right)  }D_{m,q}^{\left(
\ell_{2}\right)  }\operatorname*{Cum}\nolimits_{2}\left(  Z_{\ell_{1}}%
^{p},Z_{\ell_{2}}^{q}\right) \\
&  =\sum_{p,q=-\ell_{1},\ell_{2}}^{\ell_{1},\ell_{2}}D_{k,p}^{\left(  \ell
_{1}\right)  }\left(  -1\right)  ^{p}D_{m,q}^{\left(  \ell_{2}\right)
}\operatorname*{Cum}\nolimits_{2}\left(  Z_{\ell_{1}}^{-p\ast},Z_{\ell_{2}%
}^{q}\right)  .
\end{align*}
Now assume isotropy, $\operatorname*{Cum}\nolimits_{2}\left(  \mathcal{Z}%
_{\ell_{1}}^{k},\mathcal{Z}_{\ell_{2}}^{m}\right)  =\operatorname*{Cum}%
\nolimits_{2}\left(  Z_{\ell_{1}}^{k},Z_{\ell_{2}}^{m}\right)  $ and integrate
both sides over the sphere according to the invariant Haar measure, we have
\begin{align*}
\operatorname*{Cum}\nolimits_{2}\left(  Z_{\ell_{1}}^{k},Z_{\ell_{2}}%
^{m}\right)   &  =\sum_{p,q=-\ell_{1},\ell_{2}}^{\ell_{1},\ell_{2}}%
\frac{\delta_{\ell_{1},\ell_{2}}\delta_{-p,q}\delta_{-k,m}\left(  -1\right)
^{k}}{2\ell_{1}+1}\operatorname*{Cum}\nolimits_{2}\left(  Z_{\ell_{1}}%
^{-p\ast},Z_{\ell_{2}}^{q}\right) \\
&  =\frac{\delta_{-k,m}\left(  -1\right)  ^{m}}{2\ell_{1}+1}\sum_{p=-\ell_{1}%
}^{\ell_{1}}\delta_{\ell_{1},\ell_{2}}\operatorname*{Cum}\nolimits_{2}\left(
Z_{\ell_{1}}^{-p\ast},Z_{\ell_{2}}^{p}\right) \\
&  =\delta_{\ell_{1},\ell_{2}}\delta_{-k,m}\left(  -1\right)  ^{m}C_{2}\left(
\ell_{1}\right)
\end{align*}
where%
\[
C_{2}\left(  \ell_{1}\right)  =\frac{1}{2\ell_{1}+1}\sum_{p=-\ell_{1}}%
^{\ell_{1}}\operatorname*{Cum}\nolimits_{2}\left(  Z_{\ell_{1}}^{-p\ast
},Z_{\ell_{1}}^{p}\right)  .
\]
Hence $\operatorname*{Cum}\nolimits_{2}\left(  Z_{\ell_{1}}^{k},Z_{\ell_{2}%
}^{m\ast}\right)  =\delta_{\ell_{1},\ell_{2}}\delta_{k,m}f_{\ell_{1}}$, i.e.
the series $Z_{\ell}^{k}$ is uncorrelated.

If the series $Z_{\ell}^{k}$ is uncorrelated then
\begin{align*}
\operatorname*{Cum}\nolimits_{2}\left(  \mathcal{Z}_{\ell_{1}}^{k}%
,\mathcal{Z}_{\ell_{2}}^{m}\right)   &  =\sum_{p,q=-\ell_{1},\ell_{2}}%
^{\ell_{1},\ell_{2}}D_{k,p}^{\left(  \ell_{1}\right)  }D_{m,q}^{\left(
\ell_{2}\right)  }\operatorname*{Cum}\nolimits_{2}\left(  Z_{\ell_{1}}%
^{p},Z_{\ell_{2}}^{q}\right) \\
&  =\delta_{\ell_{1},\ell_{2}}\sum_{p,q=-\ell_{1}}^{\ell_{1}}D_{k,p}^{\left(
\ell_{1}\right)  }D_{m,q}^{\left(  \ell_{1}\right)  }\left(  -1\right)
^{p}\operatorname*{Cum}\nolimits_{2}\left(  Z_{\ell_{1}}^{-p\ast},Z_{\ell_{1}%
}^{q}\right) \\
&  =\delta_{\ell_{1},\ell_{2}}\sum_{p=-\ell_{1}}^{\ell_{1}}D_{-p,-k}^{\left(
\ell_{1}\right)  \ast}D_{m,q}^{\left(  \ell_{1}\right)  }\left(  -1\right)
^{k}\delta_{-p,q}f_{\ell_{1}}\\
&  =\delta_{\ell_{1},\ell_{2}}\delta_{-k,m}\left(  -1\right)  ^{m}f_{\ell_{1}}%
\end{align*}
since $D_{k,m}^{\left(  \ell\right)  }$ is unitary, hence $\operatorname*{Cum}%
\nolimits_{2}\left(  \mathcal{Z}_{\ell_{1}}^{k},\mathcal{Z}_{\ell_{2}}%
^{-m\ast}\right)  =\delta_{\ell_{1},\ell_{2}}\delta_{-k,m}f_{\ell_{1}%
}=\operatorname*{Cum}\nolimits_{2}\left(  Z_{\ell_{1}}^{k},Z_{\ell_{2}%
}^{-m\ast}\right)  $.

\begin{lemma}
The field $X\left(  L\right)  $ is isotropic in second order iff the
triangular series $Z_{\ell}^{k}$ is uncorrelated with variance $f_{\ell}$.
\end{lemma}

We conclude that a field $X\left(  L\right)  $ with Gaussian i.i.d. $Z_{\ell
}^{m}$ is strictly isotropic. It follows from the isotropy that $u_{\ell_{1}}$
and $u_{\ell_{2}}$with different degrees $\ell_{1}$ and $\ell_{2}$ are
uncorrelated, and
\begin{align*}
\operatorname*{Cum}\nolimits_{2}\left(  u_{\ell}\left(  L_{1}\right)
,u_{\ell}\left(  L_{2}\right)  \right)   &  =\sum_{m_{1},m_{2}=-\ell}^{\ell
}\operatorname*{Cum}\nolimits_{2}\left(  Z_{\ell}^{m_{1}},Z_{\ell}^{m_{2}%
}\right)  Y_{\ell}^{m_{1}}\left(  L_{1}\right)  Y_{\ell}^{m_{2}}\left(
L_{2}\right) \\
&  =\sum_{m_{2}=-\ell}^{\ell}f_{\ell}Y_{\ell}^{-m_{2}}\left(  L_{1}\right)
Y_{\ell}^{m_{2}}\left(  L_{2}\right)  \left(  -1\right)  ^{m_{2}}\\
&  =f_{\ell}\sum_{m_{2}=-\ell}^{\ell}Y_{\ell}^{m_{2}\ast}\left(  L_{1}\right)
Y_{\ell}^{m_{2}}\left(  L_{2}\right) \\
&  =\frac{2\ell+1}{4\pi}f_{\ell}P_{\ell}\left(  L_{1}\cdot L_{2}\right)  ,
\end{align*}
by the addition formula see (\ref{FormulaAddition}). In particular
$\operatorname*{Var}\left(  u_{\ell}\left(  L\right)  \right)  =\frac{2\ell
+1}{4\pi}f_{\ell}$. Instead of the addition formula one arrives to the same
result if applies the isotropy for $\operatorname*{Cum}\nolimits_{2}\left(
u_{\ell}\left(  L_{1}\right)  ,u_{\ell}\left(  L_{2}\right)  \right)  $.
Rotate the locations $L_{1}$ and $L_{2}$ such that $L_{2}$ becomes the North
pole $N$ and at the same time $L_{1}$ belongs to the plane $xOz.$ This
rotation will be denoted by $g_{L_{2}L_{1}}$. We have $\operatorname*{Cum}%
\nolimits_{2}\left(  u_{\ell}\left(  L_{1}\right)  ,u_{\ell}\left(
L_{2}\right)  \right)  =\operatorname*{Cum}\nolimits_{2}\left(  u_{\ell
}\left(  g_{L_{2}L_{1}}L_{1}\right)  ,u_{\ell}\left(  N\right)  \right)  ,$
since $u_{\ell}\left(  N\right)  =\sqrt{\frac{2\ell+1}{4\pi}}Z_{\ell}^{0}$ and
\newline$Y_{\ell}^{0}\left(  \vartheta,\varphi\right)  =\sqrt{\frac{2\ell
+1}{4\pi}}P_{\ell}\left(  \cos\vartheta\right)  $,
\begin{align*}
\operatorname*{Cum}\nolimits_{2}\left(  u_{\ell}\left(  L_{1}\right)
,u_{\ell}\left(  L_{2}\right)  \right)   &  =\operatorname*{Cum}%
\nolimits_{2}\left(  u_{\ell}\left(  g_{L_{2}L_{1}}L_{1}\right)  ,\sqrt
{\frac{2\ell+1}{4\pi}}Z_{\ell}^{0}\right) \\
&  =\operatorname*{Cum}\nolimits_{2}\left(  \sqrt{\frac{2\ell+1}{4\pi}}%
P_{\ell}\left(  g_{L_{2}L_{1}}L_{1}\cdot N\right)  Z_{\ell}^{0},\sqrt
{\frac{2\ell+1}{4\pi}}Z_{\ell}^{0}\right) \\
&  =\frac{2\ell+1}{4\pi}f_{\ell}P_{\ell}\left(  L_{1}\cdot L_{2}\right)  .
\end{align*}

\subsection{Spectrum}

Consider the covariance function $\mathcal{C}_{2}\left(  L_{1},L_{2}\right)
=\mathsf{E}\left(  X\left(  L_{1}\right)  -\mu\right)  \left(  X\left(
L_{2}\right)  -\mu\right)  $ for an isotropic field. Let the rotation
$g_{L_{2}L_{1}}\ $be the one which takes the location $L_{2}$ into the North
pole $N$, and $L_{1}$ into the plane $xOz.$ The Euler coordinates of
$g_{L_{2}L_{1}}L_{1}$ is the co-latitude angle $\vartheta$ and $0$,since the
rotation does not change the angle $\vartheta$ between $L_{1}\ $and $L_{2}$,
such that $\cos\vartheta=L_{1}\cdot L_{2}$. Under isotropy assumption the
joint distribution of $X\left(  L_{1}\right)  $ and $X\left(  L_{2}\right)  $
equals to the joint distribution of $X\left(  N\right)  $ and $X\left(
g_{L_{2}L_{1}}L_{1}\right)  $, i.e. $X\left(  N\right)  $ and $X\left(
\vartheta,0\right)  $ contain all pairwise information. In other words the
covariance of an isotropic field depends on $\vartheta$ only,
$\operatorname*{Cum}\nolimits_{2}\left(  X\left(  L_{1}\right)  ,X\left(
L_{2}\right)  \right)  =\operatorname{Cov}\left(  X\left(  L_{1}\right)
,X\left(  L_{2}\right)  \right)  =\mathcal{C}_{2}\left(  g_{L_{2}L_{1}}%
L_{1},N\right)  $, necessarily $\mathcal{C}_{2}\left(  L_{1},L_{2}\right)
=\mathcal{C}_{2}\left(  L_{1}\cdot L_{2}\right)  =\mathcal{C}\left(
\cos\vartheta\right)  $. Now
\begin{align*}
X\left(  N\right)   &  =\sum_{\ell=0}^{\infty}\sum_{m=-\ell}^{\ell}Y_{\ell
}^{m}\left(  N\right)  Z_{\ell}^{m}\\
&  =\sum_{\ell=0}^{\infty}\sum_{m=-\ell}^{\ell}\delta_{m,0}\sqrt{\frac
{2\ell+1}{4\pi}}Z_{\ell}^{m}\\
&  =\sum_{\ell=0}^{\infty}\sqrt{\frac{2\ell+1}{4\pi}}Z_{\ell}^{0},
\end{align*}
and
\[
X\left(  g_{L_{2}L_{1}}L_{1}\right)  =\sum_{\ell=0}^{\infty}\sum_{m=-\ell
}^{\ell}Y_{\ell}^{m}\left(  g_{L_{2}L_{1}}L_{1}\right)  Z_{\ell}^{m},
\]
We have $\operatorname*{Cum}\nolimits_{2}\left(  Z_{\ell}^{0},Z_{k}%
^{n}\right)  =\delta_{\ell,k}\delta_{0,n}f_{\ell}$, hence%
\begin{align*}
\mathcal{C}_{2}\left(  L_{1},L_{2}\right)   &  =\sum_{\ell=0}^{\infty}f_{\ell
}\sqrt{\frac{2\ell+1}{4\pi}}Y_{\ell}^{0}\left(  g_{L_{2}L_{1}}L_{1}\right) \\
&  =\sum_{\ell=0}^{\infty}f_{\ell}\frac{2\ell+1}{4\pi}P_{\ell}\left(
\cos\vartheta\right)  .
\end{align*}
Similarly to the time series setup the covariance $\mathcal{C}_{2}\left(
L_{1},L_{2}\right)  $ is expanded in terms of an orthonormed system
$\frac{2\ell+1}{4\pi}P_{\ell}\left(  \cos\vartheta\right)  $ with coefficients
$f_{\ell}$. In particular the variance $\operatorname*{Var}\left(  X\left(
L\right)  \right)  $ is split up into the superposition of $f_{\ell}$'s. Hence
$S_{2}\left(  \ell\right)  =f_{\ell}$ is called \textbf{spectrum} of the field
$X\left(  L\right)  $ according to the frequency $\ell$.

\section{Bispectrum}

If the field $X\left(  L\right)  $ is non-Gaussian then the first
characteristic after the second order moments to be considered is the third
order cumulant, referred to bicovariance or 3-point covariance also since it
is the third order central moment. The corresponding quantity in frequency
domain is the bispectrum. Similarly to the spectrum when the covariance
(second order cumulant)%
\[
\operatorname*{Cum}\nolimits_{2}\left(  X\left(  L_{1}\right)  ,X\left(
L_{2}\right)  \right)  =\sum_{\ell=0}^{\infty}f_{\ell}\frac{2\ell+1}{4\pi
}P_{\ell}\left(  \cos\vartheta\right)  ,
\]
is split up to the superposition of some orthogonal system of functions and
the coefficients are called spectrum. We see that the orthogonal and normed
system is the Legendre polynomial system $\left\{  \frac{2\ell+1}{4\pi}%
P_{\ell}\right\}  $ and the spectrum $S_{2}\left(  \ell\right)  =f_{\ell}$
corresponds to the frequency $\ell$. We put a similar question for the higher
order angular spectra as well. Namely, under assumption of isotropy of the
$m^{th}$ order cumulant we consider the series expansion according to an
orthonormed system and the coefficients will be called $m^{th}$ order
spectrum. The Wigner 3j symbols
\[%
\begin{pmatrix}
\ell_{1} & \ell_{2} & \ell_{3}\\
m_{1} & m_{2} & m_{3}%
\end{pmatrix}
\]
will be intensively used from now on, see Appendix \ref{App_Spher_Harmonics},
\ref{Wigner_3j}. It depends on the quantum numbers $\ell_{1},\ell_{2},\ell
_{3}$ called degrees and orders $m_{1},m_{2},m_{3}$.

\subsection{Isotropy in third order}

The rotation applied to he triangular array $Z_{\ell}^{m}$ shows the necessary
and sufficient condition for the third order isotropy.

\begin{lemma}
\label{Lemm_Iso3}The field $X\left(  L\right)  $ is isotropic in third order
iff the bicovariance \newline$\operatorname*{Cum}\nolimits_{3}\left(
Z_{\ell_{1}}^{m_{1}},Z_{\ell_{2}}^{m_{2}},Z_{\ell_{3}}^{m_{3}}\right)  $ of
the triangular series $Z_{\ell}^{m}$ has the form
\begin{equation}
\operatorname*{Cum}\nolimits_{3}\left(  Z_{\ell_{1}}^{m_{1}},Z_{\ell_{2}%
}^{m_{2}},Z_{\ell_{3}}^{m_{3}}\right)  =%
\begin{pmatrix}
\ell_{1} & \ell_{2} & \ell_{3}\\
m_{1} & m_{2} & m_{3}%
\end{pmatrix}
B_{3}\left(  \ell_{1},\ell_{2},\ell_{3}\right)  , \label{BicovZ}%
\end{equation}
with
\[
B_{3}\left(  \ell_{1},\ell_{2},\ell_{3}\right)  =\sum_{k_{1},k_{2},k_{3}}%
\begin{pmatrix}
\ell_{1} & \ell_{2} & \ell_{3}\\
k_{1} & k_{2} & k_{3}%
\end{pmatrix}
\operatorname*{Cum}\nolimits_{3}\left(  Z_{\ell_{1}}^{k_{1}},Z_{\ell_{2}%
}^{k_{2}},Z_{\ell_{3}}^{k_{3}}\right)  .
\]

\end{lemma}

See Appendix \ref{Append_Bisp_Trisp} for the proof. Note here that
$\operatorname*{Cum}\nolimits_{3}\left(  Z_{\ell_{1}}^{m_{1}},Z_{\ell_{2}%
}^{m_{2}},Z_{\ell_{3}}^{m_{3}}\right)  =0$, if $m_{1}+m_{2}+m_{3}\neq0$, that
is not all the possible bicovariances come into the picture, moreover the
cumulants are depending on the orders $m_{1},m_{2},m_{3}$ through the Wigner
3j symbols only. In other words the third order probabilistic property inside
a fixed quantum number $\ell$ does not show up. The function $B_{3}$ of the
frequencies is an average of the cumulants $\operatorname*{Cum}\nolimits_{3}%
\left(  Z_{\ell_{1}}^{k_{1}},Z_{\ell_{2}}^{k_{2}},Z_{\ell_{3}}^{k_{3}}\right)
$ by 'probability' amplitudes, hence it is called as 'angle average
bispectrum', we shall turn back to the notion of the bispectrum later in this section.

Notice first that the quantity
\begin{equation}%
\begin{pmatrix}
\ell_{1} & \ell_{2} & \ell_{3}\\
m_{1} & m_{2} & m_{3}%
\end{pmatrix}
B_{3}\left(  \ell_{1},\ell_{2},\ell_{3}\right)  , \label{Bisp_sym}%
\end{equation}
is symmetric in $\ell_{1},\ell_{2},\ell_{3}$, see Appendix
\ref{Append_Bisp_Trisp}, proof \ref{Bisp_sym_degree}. Moreover $B_{3}\left(
\ell_{1},\ell_{2},\ell_{3}\right)  $ is symmetric in $\ell_{1},\ell_{2}%
,\ell_{3}$ as well, since by parity transformation (\ref{Parity}) $\ell
_{1}+\ell_{2}+\ell_{3}$ must be even, the coefficients $%
\begin{pmatrix}
\ell_{1} & \ell_{2} & \ell_{3}\\
k_{1} & k_{2} & k_{3}%
\end{pmatrix}
$ are symmetric as well.

From now on we \textit{fix the order} of $\ell_{1},\ell_{2},\ell_{3}$ such
that $\ell_{1}\leq\ell_{2}\leq\ell_{3}$.

Now we turn to the angular momentum field $u_{\ell}$. First observe
\begin{multline*}
\operatorname*{Cum}\nolimits_{3}\left(  u_{\ell_{1}}\left(  L_{1}\right)
,u_{\ell_{2}}\left(  L_{2}\right)  ,u_{\ell_{3}}\left(  L_{3}\right)  \right)
=\sum_{m_{1},m_{2},m_{3}}Y_{\ell_{1}}^{m_{1}}\left(  L_{1}\right)  Y_{\ell
_{2}}^{m_{2}}\left(  L_{2}\right)  Y_{\ell_{3}}^{m_{3}}\left(  L_{3}\right)
\operatorname*{Cum}\nolimits_{3}\left(  Z_{\ell_{1:3}}^{m_{1:3}}\right) \\
=\sum_{m_{1},m_{2},m_{3}}Y_{\ell_{1}}^{m_{1}}\left(  L_{1}\right)  Y_{\ell
_{2}}^{m_{2}}\left(  L_{2}\right)  Y_{\ell_{3}}^{m_{3}}\left(  L_{3}\right)
\begin{pmatrix}
\ell_{1} & \ell_{2} & \ell_{3}\\
m_{1} & m_{2} & m_{3}%
\end{pmatrix}
B_{3}\left(  \ell_{1},\ell_{2},\ell_{3}\right) \\
=B_{3}\left(  \ell_{1},\ell_{2},\ell_{3}\right)  \sum_{m_{1},m_{2},m_{3}}%
\begin{pmatrix}
\ell_{1} & \ell_{2} & \ell_{3}\\
m_{1} & m_{2} & m_{3}%
\end{pmatrix}
Y_{\ell_{1}}^{m_{1}}\left(  L_{1}\right)  Y_{\ell_{2}}^{m_{2}}\left(
L_{2}\right)  Y_{\ell_{3}}^{m_{3}}\left(  L_{3}\right)  .
\end{multline*}
This last expression is invariant under both the rotation of $L_{j}$'s and
ordering of $\ell_{1},\ell_{2},\ell_{3}$. The \textit{3-product of spherical
harmonics}\textbf{\ } $Y_{\ell}^{m}$, \textbf{\ }
\[
\widetilde{I}_{\ell_{1},\ell_{2},\ell_{3}}\left(  L_{1},L_{2},L_{3}\right)
=\sum_{m_{1},m_{2},m_{3}}%
\begin{pmatrix}
\ell_{1} & \ell_{2} & \ell_{3}\\
m_{1} & m_{2} & m_{3}%
\end{pmatrix}
Y_{\ell_{1}}^{m_{1}}\left(  L_{1}\right)  Y_{\ell_{2}}^{m_{2}}\left(
L_{2}\right)  Y_{\ell_{3}}^{m_{3}}\left(  L_{3}\right)
\]
is \textit{rotational invariant}, see Appendix (\ref{RotationProd3}),
therefore without any loss of generality we apply the rotation $g_{L_{3}L_{2}%
}$ which takes the location $L_{3}\ $into the North pole $N$ and at the same
time takes $L_{2}$ into the $zOx$- plane
\begin{multline*}
\operatorname*{Cum}\nolimits_{3}\left(  u_{\ell_{1}}\left(  L_{1}\right)
,u_{\ell_{2}}\left(  L_{2}\right)  ,u_{\ell_{3}}\left(  L_{3}\right)  \right)
=\operatorname*{Cum}\nolimits_{3}\left(  u_{\ell_{1}}\left(  g_{L_{3}L_{2}%
}L_{1}\right)  ,u_{\ell_{2}}\left(  g_{L_{3}L_{2}}L_{2}\right)  ,u_{\ell_{3}%
}\left(  N\right)  \right) \\
=B_{3}\left(  \ell_{1},\ell_{2},\ell_{3}\right)  \sum_{m_{1}=-\ell_{1}}%
^{\ell_{1}}\sum_{m_{2}=-\ell_{2}}^{\ell_{2}}Y_{\ell_{1}}^{m_{1}}\left(
g_{L_{3}L_{2}}L_{1}\right)  Y_{\ell_{2}}^{m_{2}}\left(  g_{L_{3}L_{2}}%
L_{2}\right)  Y_{\ell_{3}}^{0}\left(  N\right)
\begin{pmatrix}
\ell_{1} & \ell_{2} & \ell_{3}\\
m_{1} & m_{2} & 0
\end{pmatrix}
\\
=\sqrt{\frac{2\ell_{3}+1}{4\pi}}B_{3}\left(  \ell_{1},\ell_{2},\ell
_{3}\right)  \sum_{m_{1},m_{2}}Y_{\ell_{1}}^{m_{1}}\left(  g_{L_{3}L_{2}}%
L_{1}\right)  Y_{\ell_{2}}^{m_{2}}\left(  g_{L_{3}L_{2}}L_{2}\right)
\begin{pmatrix}
\ell_{1} & \ell_{2} & \ell_{3}\\
m_{1} & m_{2} & 0
\end{pmatrix}
.
\end{multline*}
The third order cumulants of $u_{\ell}$ contains an orthonormed system of
functions
\begin{align*}
I_{\ell_{1},\ell_{2},\ell_{3}}\left(  L_{1},L_{2},L_{3}\right)   &
=I_{\ell_{1},\ell_{2},\ell_{3}}\left(  g_{L_{3}L_{2}}L_{1},g_{L_{3}L_{2}}%
L_{2},N\right) \\
&  =\sqrt{\frac{2\ell_{3}+1}{4\pi}}\sum_{m_{1:2}}%
\begin{pmatrix}
\ell_{1} & \ell_{2} & \ell_{3}\\
m_{1} & m_{2} & 0
\end{pmatrix}
Y_{\ell_{1}}^{m_{1}}\left(  g_{L_{3}L_{2}}L_{1}\right)  Y_{\ell_{2}}^{m_{2}%
}\left(  g_{L_{3}L_{2}}L_{2}\right)  .
\end{align*}
In this way $I_{\ell_{1},\ell_{2},\ell_{3}}\left(  L_{1},L_{2},L_{3}\right)  $
is connected to the \textit{bipolar spherical harmonics}\textbf{,} i.e. to the
irreducible tensor products of the spherical harmonics with different
arguments, \cite{Varshalovich1988}.
\[
\left\{  Y_{\ell_{1}}\left(  L_{1}\right)  \otimes Y_{\ell_{2}}\left(
L_{2}\right)  \right\}  _{\ell,m}=\left(  -1\right)  ^{\ell_{1}-\ell_{2}%
+m}\sqrt{2\ell+1}\sum_{m_{1},m_{2}}%
\begin{pmatrix}
\ell_{1} & \ell_{2} & \ell\\
m_{1} & m_{2} & m
\end{pmatrix}
Y_{\ell_{1}}^{m_{1}}\left(  L_{1}\right)  Y_{\ell_{2}}^{m_{2}}\left(
L_{2}\right)
\]
Rewrite $I_{\ell_{1},\ell_{2},\ell_{3}}$ in terms of Euler angles
\begin{align*}
\mathcal{I}_{\ell_{1},\ell_{2},\ell_{3}}\left(  \vartheta_{1},\varphi
_{1},\vartheta_{2}\right)   &  =I_{\ell_{1},\ell_{2},\ell_{3}}\left(
g_{L_{3}L_{2}}L_{1},g_{L_{3}L_{2}}L_{2},N\right) \\
&  =\sqrt{\frac{2\ell_{3}+1}{4\pi}}\sum_{m_{1},m_{2}}%
\begin{pmatrix}
\ell_{1} & \ell_{2} & \ell_{3}\\
m_{1} & m_{2} & 0
\end{pmatrix}
Y_{\ell_{1}}^{m_{1}}\left(  \vartheta_{1},\varphi_{1}\right)  Y_{\ell_{2}%
}^{-m_{1}}\left(  \vartheta_{2},0\right) \\
&  =\sqrt{\frac{2\ell_{3}+1}{4\pi}}\sum_{m_{1},m_{2}}%
\begin{pmatrix}
\ell_{1} & \ell_{2} & \ell_{3}\\
m_{1} & m_{2} & 0
\end{pmatrix}
\left(  -1\right)  ^{m_{1}}Y_{\ell_{1}}^{m_{1}}\left(  \vartheta_{1}%
,\varphi_{1}\right)  Y_{\ell_{2}}^{m_{1}\ast}\left(  \vartheta_{2},0\right)  .
\end{align*}
According to the usual measure $\Omega\left(  dL_{1}\right)  \Omega\left(
dL_{2}\right)  =\sin\vartheta_{1}d\vartheta_{1}d\varphi_{1}\sin\vartheta
_{2}d\vartheta_{2}d\varphi_{2}$ on $\vartheta\in\left[  0,\pi\right]  $,
$\varphi\in\left[  0,2\pi\right]  $, we have
\[
\int_{\mathbb{S}_{2}}Y_{\ell}^{m\ast}\left(  \vartheta,0\right)  Y_{j}%
^{k}\left(  \vartheta,0\right)  \Omega\left(  dL\right)  =\delta_{\ell
,j}\delta_{m,k},
\]
hence%
\[
\iint\limits_{\mathbb{S}_{2}}\mathcal{I}_{\ell_{1},\ell_{2},\ell_{3}}^{\ast
}\mathcal{I}_{j_{1},j_{2},j_{3}}\Omega\left(  dL_{1}\right)  \Omega\left(
dL_{2}\right)  =\delta_{\ell_{1},j_{1}}\delta_{\ell_{2},j_{2}}\delta_{\ell
_{3},j_{3}},
\]
since
\[
\left(  2\ell_{3}+1\right)  \sum_{m_{1},m_{2}}%
\begin{pmatrix}
\ell_{1} & \ell_{2} & \ell_{3}\\
m_{1} & m_{2} & 0
\end{pmatrix}
^{2}=1.
\]
The third order cumulants of $u_{\ell}$ simplifies
\[
\operatorname*{Cum}\nolimits_{3}\left(  u_{\ell_{1}}\left(  L_{1}\right)
,u_{\ell_{2}}\left(  L_{2}\right)  ,u_{\ell_{3}}\left(  L_{3}\right)  \right)
=B_{3}\left(  \ell_{1},\ell_{2},\ell_{3}\right)  \mathcal{I}_{\ell_{1}%
,\ell_{2},\ell_{3}}\left(  \vartheta_{1},\varphi_{1},\vartheta_{2}\right)  .
\]
If we denote it by
\[
\operatorname*{Cum}\nolimits_{3}\left(  u_{\ell_{1}}\left(  L_{1}\right)
,u_{\ell_{2}}\left(  L_{2}\right)  ,u_{\ell_{3}}\left(  L_{3}\right)  \right)
=C_{u,\ell_{1},\ell_{2},\ell_{3}}\left(  \vartheta_{1},\varphi_{1}%
,\vartheta_{2}\right)  ,
\]
then we have the inversion formula
\begin{align*}
&  \iint\limits_{\mathbb{S}_{2}}C_{u,j_{1},j_{2},j_{3}}\left(  \vartheta
_{1},\varphi_{1},\vartheta_{2}\right)  \mathcal{I}_{\ell_{1},\ell_{2},\ell
_{3}}\left(  \vartheta_{1},\varphi_{1},\vartheta_{2}\right)  \Omega\left(
dL_{1}\right)  \Omega\left(  dL_{2}\right) \\
&  =\delta_{\ell_{1},j_{1}}\delta_{\ell_{2},j_{2}}\delta_{\ell_{3},j_{3}}%
B_{3}\left(  \ell_{1},\ell_{2},\ell_{3}\right)  .
\end{align*}

\subsection{Bispectrum of isotropic fields on sphere}

The use of the bispectrum for detecting non-Gaussianity and nonlinearity is
well known in time series analysis \cite{Rao84a}, \cite{Hinich_test-82}
\cite{Terdik-Math-JTSA-98}, the similar question has been put and studied for
CMB analysis see \cite{Marinucci2004}, \cite{Adshead2012} and references
therein. The asymptotics of the bispectrum has a long history started by
\cite{Rosenblatt-Ness_est_bisp-65}, the asymptotic distribution of angular
bispectrum when the degrees are increasing is considered by
\cite{Marinucci2006}, \cite{Marinucci2008}.

Start with the general case $\operatorname*{Cum}\nolimits_{3}\left(  X\left(
L_{1}\right)  ,X\left(  L_{2}\right)  ,X\left(  L_{3}\right)  \right)  $. The
value of the bicovariance does not change under the rotation $g_{L_{3}L_{2}}$
\begin{align*}
\operatorname*{Cum}\nolimits_{3}\left(  X\left(  L_{1}\right)  ,X\left(
L_{2}\right)  ,X\left(  L_{3}\right)  \right)   &  =\operatorname*{Cum}%
\nolimits_{3}\left(  X\left(  g_{L_{3}L_{2}}L_{1}\right)  ,X\left(
g_{L_{3}L_{2}}L_{2}\right)  ,X\left(  N\right)  \right) \\
&  =\operatorname*{Cum}\nolimits_{3}\left(  X\left(  \vartheta_{1},\varphi
_{1}\right)  ,X\left(  \vartheta_{2},0\right)  ,X\left(  N\right)  \right)  ,
\end{align*}
the spherical coordinates of the locations $g_{L_{3}L_{2}}L_{1}=\left(
\vartheta_{1},\varphi_{1}\right)  $, $g_{L_{3}L_{2}}L_{2}=\left(
\vartheta_{2},0\right)  $ are defined by the locations $L_{1}$, $L_{2}$,
$L_{3}$ as follows: $g_{L_{3}L_{2}}L_{2}\cdot N=L_{2}\cdot L_{3}=$
$\cos\vartheta_{2}$, $g_{L_{3}L_{2}}L_{1}\cdot N=L_{1}\cdot L_{3}%
=\cos\vartheta_{1}$. These angles $\left(  \vartheta_{1},\varphi_{1}%
,\vartheta_{2}\right)  $ define uniquely, up to rotations, the triangle given
by locations $L_{1}$, $L_{2}$, $L_{3}$. In general the bicovariance is written
$C_{3}\left(  \vartheta_{1},\varphi_{1},\vartheta_{2}\right)
=\operatorname*{Cum}\nolimits_{3}\left(  X\left(  \vartheta_{1},\varphi
_{1}\right)  ,X\left(  \vartheta_{2},0\right)  ,X\left(  N\right)  \right)  $,
i.e. it depends on a location $L\left(  \vartheta_{2},0\right)  $ from the
main circle $\left(  \varphi_{2}=0\right)  $ and a general location
$L_{0}=L_{0}\left(  \vartheta_{1},\varphi_{1}\right)  $.

$\;$%
\begin{align}
&  \operatorname*{Cum}\nolimits_{3}\left(  X\left(  L_{1}\right)  ,X\left(
L_{2}\right)  ,X\left(  L_{3}\right)  \right) \nonumber\\
&  =\sum_{\ell_{1},\ell_{2},\ell_{3}=0}^{\infty}\sqrt{\frac{2\ell_{3}+1}{4\pi
}}\sum_{m_{1}=-\ell_{1}}^{\ell_{1}}\sum_{m_{2}=-\ell_{2}}^{\ell_{2}%
}\operatorname*{Cum}\nolimits_{3}\left(  Z_{\ell_{1}}^{m_{1}},Z_{\ell_{2}%
}^{m_{2}},Z_{\ell_{3}}^{0}\right)  Y_{\ell_{1}}^{m_{1}}\left(  g_{L_{3}L_{2}%
}L_{1}\right)  Y_{\ell_{2}}^{m_{2}}\left(  g_{L_{3}L_{2}}L_{2}\right)
\nonumber\\
&  =\sum_{\ell_{1},\ell_{2},\ell_{3}=0}^{\infty}B_{3}\left(  \ell_{1},\ell
_{2},\ell_{3}\right)  \mathcal{I}_{\ell_{1},\ell_{2},\ell_{3}}\left(
\vartheta_{1},\varphi_{1},\vartheta_{2}\right)  . \label{Expansion_bicov}%
\end{align}
This series expansion of $\operatorname*{Cum}\nolimits_{3}\left(  X\left(
L_{1}\right)  ,X\left(  L_{2}\right)  ,X\left(  L_{3}\right)  \right)  $
implies the following definition

\begin{definition}
The bispectrum of the isotropic field $X\left(  L\right)  $ is
\[
S_{3}\left(  \ell_{1},\ell_{2},\ell_{3}\right)  =B_{3}\left(  \ell_{1}%
,\ell_{2},\ell_{3}\right)  ,
\]
and the bicoherence of $X\left(  L\right)  $ is
\[
\frac{S_{3}\left(  \ell_{1},\ell_{2},\ell_{3}\right)  }{\sqrt{S_{2}\left(
\ell_{1}\right)  S_{2}\left(  \ell_{2}\right)  S_{2}\left(  \ell_{3}\right)
}}=\frac{B_{3}\left(  \ell_{1},\ell_{2},\ell_{3}\right)  }{\sqrt{f_{\ell_{1}%
}f_{\ell_{2}}f_{\ell_{3}}}}.
\]

\end{definition}

\begin{theorem}
The bicovariances $\operatorname*{Cum}\nolimits_{3}\left(  X\left(
L_{1}\right)  ,X\left(  L_{2}\right)  ,X\left(  L_{3}\right)  \right)  $ of
the isotropic field $X\left(  L\right)  $ have the series expansion
(\ref{Expansion_bicov}) in terms of the bispectrum $B_{3}\left(  \ell_{1}%
,\ell_{2},\ell_{3}\right)  $ and orthonormed system $\mathcal{I}_{\ell
_{1},\ell_{2},\ell_{3}}$, hence
\[
B_{3}\left(  \ell_{1},\ell_{2},\ell_{3}\right)  =\iint\limits_{\mathbb{S}_{2}%
}\operatorname*{Cum}\nolimits_{3}\left(  X\left(  \vartheta_{1},\varphi
_{1}\right)  ,X\left(  \vartheta_{2},0\right)  ,X\left(  N\right)  \right)
\mathcal{I}_{\ell_{1},\ell_{2},\ell_{3}}\left(  \vartheta_{1},\varphi
_{1},\vartheta_{2}\right)  \Omega\left(  dL_{1}\right)  \Omega\left(
dL_{2}\right)  .
\]

\end{theorem}

\subsubsection{Particular cases}

For any locations $L_{1}$ and $L$ we have
\[
\operatorname*{Cum}\nolimits_{3}\left(  u_{\ell_{1}}\left(  L_{1}\right)
,u_{\ell_{2}}\left(  L\right)  ,u_{\ell_{3}}\left(  L\right)  \right)  =\sqrt{%
{\displaystyle\prod_{i}}
\frac{2\ell_{i}+1}{4\pi}}%
\begin{pmatrix}
\ell_{1} & \ell_{2} & \ell_{3}\\
0 & 0 & 0
\end{pmatrix}
B_{3}\left(  \ell_{1},\ell_{2},\ell_{3}\right)  P_{\ell_{1}}\left(  \cos
L\cdot L_{1}\right)  ,
\]
from this readily follows%
\begin{align*}
\operatorname*{Cum}\nolimits_{3}\left(  X\left(  L_{1}\right)  ,X\left(
L\right)  ,X\left(  L\right)  \right)   &  =\sum_{\ell_{1},\ell_{2},\ell
_{3}=0}^{\infty}%
{\displaystyle\prod_{i=1}^{3}}
\sqrt{\frac{2\ell_{i}+1}{4\pi}}%
\begin{pmatrix}
\ell_{1} & \ell_{2} & \ell_{3}\\
0 & 0 & 0
\end{pmatrix}
\\
&  \times B_{3}\left(  \ell_{1},\ell_{2},\ell_{3}\right)  P_{\ell_{1}}\left(
\cos L\cdot L_{1}\right)  .
\end{align*}
In particular, according to the Condon and Shortley phase convention, we have
for any location $L$ the third order central moment
\[
\operatorname*{Cum}\nolimits_{3}\left(  u_{\ell_{1}}\left(  L\right)
,u_{\ell_{2}}\left(  L\right)  ,u_{\ell_{3}}\left(  L\right)  \right)  =\sqrt{%
{\displaystyle\prod}
\frac{2\ell_{i}+1}{4\pi}}%
\begin{pmatrix}
\ell_{1} & \ell_{2} & \ell_{3}\\
0 & 0 & 0
\end{pmatrix}
B_{3}\left(  \ell_{1},\ell_{2},\ell_{3}\right)  ,
\]
and the third order cumulant (central moment) $\operatorname*{Cum}%
\nolimits_{3}\left(  X\left(  L\right)  ,X\left(  L\right)  ,X\left(
L\right)  \right)  $ of $X\left(  L\right)  $ is
\begin{align*}
\operatorname*{Cum}\nolimits_{3}\left(  X\left(  L\right)  ,X\left(  L\right)
,X\left(  L\right)  \right)   &  =\operatorname*{Cum}\nolimits_{3}\left(
X\left(  N\right)  ,X\left(  N\right)  ,X\left(  N\right)  \right) \\
&  =\sum_{\ell_{1},\ell_{2},\ell_{3}=0}^{\infty}\operatorname*{Cum}%
\nolimits_{3}\left(  Z_{\ell_{1}}^{0},Z_{\ell_{2}}^{0},Z_{\ell_{3}}%
^{0}\right)
{\displaystyle\prod}
\sqrt{\frac{2\ell_{i}+1}{4\pi}}\\
&  =\sum_{\ell_{1},\ell_{2},\ell_{3}=0}^{\infty}%
{\displaystyle\prod}
\sqrt{\frac{2\ell_{i}+1}{4\pi}}%
\begin{pmatrix}
\ell_{1} & \ell_{2} & \ell_{3}\\
0 & 0 & 0
\end{pmatrix}
B_{3}\left(  \ell_{1},\ell_{2},\ell_{3}\right)
\end{align*}
Note that the summation is taken over those $\ell_{1},\ell_{2},\ell_{3}$ when
$\ell_{1}+\ell_{2}+\ell_{3}=\mathcal{L}$ is even and for this case the
following explicit expression is valid
\[%
\begin{pmatrix}
\ell_{1} & \ell_{2} & \ell_{3}\\
0 & 0 & 0
\end{pmatrix}
=\left(  -1\right)  ^{\mathcal{L}/2}\sqrt{\frac{\prod\left(  \mathcal{L}%
-2\ell_{j}\right)  !}{\left(  \mathcal{L}+1\right)  !}}\frac{\left(
\mathcal{L}/2\right)  !}{\prod\left(  \mathcal{L}/2-\ell_{j}\right)  !}%
\]

\subsection{Linear field}

First we assume that the rows of the triangle array $Z_{\ell}^{m}$ contain
\textbf{\ }independent variables. In other words the angular momentum field
$u_{\ell}$ is linear. Then for a fixed degree $\ell$%
\[
\operatorname*{Cum}\nolimits_{3}\left(  Z_{\ell}^{m_{1}},Z_{\ell}^{m_{2}%
},Z_{\ell}^{m_{3}}\right)  =\delta_{m_{i}=m}%
\begin{pmatrix}
\ell & \ell & \ell\\
m & m & m
\end{pmatrix}
B_{3}\left(  \ell,\ell,\ell\right)  ,
\]
now the selection rules apply, see Appendix \ref{App_Spher_Harmonics},
\ref{Select.Rules}
\[
\operatorname*{Cum}\nolimits_{3}\left(  Z_{\ell}^{m_{1}},Z_{\ell}^{m_{2}%
},Z_{\ell}^{m_{3}}\right)  =\delta_{m_{i}=m}\delta_{m=0}%
\begin{pmatrix}
\ell & \ell & \ell\\
0 & 0 & 0
\end{pmatrix}
B_{3}\left(  \ell,\ell,\ell\right)  .
\]
Hence the only nonzero third order cumulant might be $\operatorname*{Cum}%
\nolimits_{3}\left(  Z_{\ell}^{0},Z_{\ell}^{0},Z_{\ell}^{0}\right)  $. Further
the bispect\-rum
\begin{align*}
B_{3}\left(  \ell,\ell,\ell\right)   &  =\sum_{k}%
\begin{pmatrix}
\ell & \ell & \ell\\
k & k & k
\end{pmatrix}
\operatorname*{Cum}\nolimits_{3}\left(  Z_{\ell}^{k},Z_{\ell}^{k},Z_{\ell}%
^{k}\right) \\
&  =%
\begin{pmatrix}
\ell & \ell & \ell\\
0 & 0 & 0
\end{pmatrix}
\operatorname*{Cum}\nolimits_{3}\left(  Z_{\ell}^{0},Z_{\ell}^{0},Z_{\ell}%
^{0}\right)  ,
\end{align*}
therefore
\[
\operatorname*{Cum}\nolimits_{3}\left(  Z_{\ell}^{0},Z_{\ell}^{0},Z_{\ell}%
^{0}\right)  =%
\begin{pmatrix}
\ell & \ell & \ell\\
0 & 0 & 0
\end{pmatrix}
^{2}\operatorname*{Cum}\nolimits_{3}\left(  Z_{\ell}^{0},Z_{\ell}^{0},Z_{\ell
}^{0}\right)  .
\]
We conclude from this that from the isotropy and independence follows
$\operatorname*{Cum}\nolimits_{3}\left(  Z_{\ell}^{m},Z_{\ell}^{m},Z_{\ell
}^{m}\right)  =0$. Moreover if all the members of the triangle array $Z_{\ell
}^{m}$ are independent then \newline$\operatorname*{Cum}\nolimits_{3}\left(
X\left(  L_{1}\right)  ,X\left(  L_{2}\right)  ,X\left(  L_{3}\right)
\right)  =0$. Third order cumulants vanish for instance when the distribution
is Gaussian, or symmetric, etc.

\subsection{Symmetries of the bispectrum}

The cumulants $\operatorname*{Cum}\nolimits_{3}\left(  X\left(  L_{1}\right)
,X\left(  L_{2}\right)  ,X\left(  L_{3}\right)  \right)  $ according to
different locations $L_{1},L_{2}$ and $L_{3}$ are defined by the spherical
triangle with vertices $L_{1},L_{2}$ and $L_{3}$. The efficiency of
computations and avoiding the redundancy in statistical procedures require to
determine the principal domains. \ The\textbf{\ }\textit{principal
domain}\textbf{\ }for the bispectrum according to the \textbf{\ }frequencies
$\ell_{1},\ell_{2},\ell_{3}$ is the following

\begin{enumerate}
\item $\ell_{1},\ell_{2},\ell_{3}$ is monotone: $\ell_{1}\leq\ell_{2}\leq
\ell_{3},$

\item $\ell_{1}+\ell_{2}+\ell_{3}$ is even, see Remark \ref{Rem_parity},

\item $\ell_{1},\ell_{2},\ell_{3}$ fulfils the triangular inequality
$\left\vert \ell_{1}-\ell_{2}\right\vert \leq\ell_{3}\leq\ell_{1}+\ell_{2}$,

\item $\operatorname*{Cum}\nolimits_{3}\left(  Z_{\ell_{1}}^{m_{1}}%
,Z_{\ell_{2}}^{m_{2}},Z_{\ell_{3}}^{m_{3}}\right)  =0$, unless $m_{1}%
+m_{2}+m_{3}=0.$
\end{enumerate}

Similarly the \textbf{\ }\textit{principal domain} for the bicovariance
according to the locations $\left(  L_{1},L_{2},L_{3}\right)  $ is $\left\{
\left(  \vartheta_{1},\varphi_{1},\vartheta_{2}\right)  \left\vert
\vartheta_{1}\in\left[  0,\pi\right]  ,\varphi_{2}\in\left[  0,\pi\right]
,\vartheta_{2}\in\left[  0,\pi\right]  \right.  \right\}  $. We apply the
following notation of these angles $\cos\vartheta_{1}=L_{2}\cdot L_{3}$,
$\cos\vartheta_{2}=L_{1}\cdot L_{3}$, and $\varphi_{2}=\phi_{3}$ is the
surface angle at $L_{3}$. The third central angle is given by $\cos
\vartheta_{3}=L_{1}\cdot L_{2}$. The surface angle $\phi_{3}$ can be
calculated for instance from the cosine formula
\[
\cos\vartheta_{3}=\cos\vartheta_{1}\cos\vartheta_{2}+\sin\vartheta_{1}%
\sin\vartheta_{2}\cos\varphi_{2}.
\]

\section{Trispectrum}

\subsection{Isotropy in fourth order}

In the papers \cite{Hu2001}, \cite{Kogo2006} for instance, the notion of
\textit{trispectrum} is used for the coefficients of the series expansion of
the fourth order moments. Naturally, because of the expression of the moment
via cumulants (\ref{mocum}) it contains terms which correspond to the Gaussian
part of the field. This is one of the reasons that the time series analysis
considers the series expansion of cumulants instead of moments and reserve the
notion of \textit{trispectrum} for this case only, see
\cite{Bril_polyspect-65}, \cite{Molle-Hinich95}. The cumulants up to third
order equal to the central moments, but it is not so for higher order, see
Appendix \ref{Sect_Cummulants}. Cumulants are very useful, for instance they
measure the distance of a distribution from the Gaussian one.

\begin{lemma}
\label{Lemm_Iso4}The field $X\left(  L\right)  $ is isotropic in fourth order
iff the cumulant \newline$\operatorname*{Cum}\nolimits_{4}\left(  Z_{\ell_{1}%
}^{m_{1}},Z_{\ell_{2}}^{m_{2}},Z_{\ell_{3}}^{m_{3}},Z_{\ell_{4}}^{m_{4}%
}\right)  $ of the triangular array $Z_{\ell}^{m}$ has the form
\begin{align}
\operatorname*{Cum}\nolimits_{4}\left(  Z_{\ell_{1}}^{m_{1}},Z_{\ell_{2}%
}^{m_{2}},Z_{\ell_{3}}^{m_{3}},Z_{\ell_{4}}^{m_{4}}\right)   &  =\sum
_{\ell^{1},m^{1}}\left(  -1\right)  ^{m^{1}}%
\begin{pmatrix}
\ell_{1} & \ell_{2} & \ell^{1}\\
m_{1} & m_{2} & m^{1}%
\end{pmatrix}%
\begin{pmatrix}
\ell^{1} & \ell_{3} & \ell_{4}\\
-m^{1} & m_{3} & m_{4}%
\end{pmatrix}
\nonumber\\
&  \times\sqrt{2\ell^{1}+1}T_{4}\left(  \left.  \ell_{1},\ell_{2},\ell
_{3},\ell_{4}\right\vert \ell^{1}\right)  \label{Cum4_Z}%
\end{align}
where $m^{1}=-m_{1}-m_{2}$.
\end{lemma}

See the Appendix \ref{Append_Trisp} for the proof. The function $T_{4}\left(
\left.  \ell_{1},\ell_{2},\ell_{3},\ell_{4}\right\vert \ell^{1}\right)  $ has
the form
\begin{align}
T_{4}\left(  \left.  \ell_{1},\ell_{2},\ell_{3},\ell_{4}\right\vert \ell
^{1}\right)   &  =\sqrt{2\ell^{1}+1}\sum_{k_{1},k_{2},k_{3},k_{4},k^{1}%
}\left(  -1\right)  ^{k^{1}}%
\begin{pmatrix}
\ell_{1} & \ell_{2} & \ell^{1}\\
k_{1} & k_{2} & k^{1}%
\end{pmatrix}%
\begin{pmatrix}
\ell^{1} & \ell_{3} & \ell_{4}\\
-k^{1} & k_{3} & k_{4}%
\end{pmatrix}
\nonumber\\
&  \times\operatorname*{Cum}\nolimits_{4}\left(  Z_{\ell_{1}}^{k_{1}}%
,Z_{\ell_{2}}^{k_{2}},Z_{\ell_{3}}^{k_{3}},Z_{\ell_{4}}^{k_{4}}\right)
\label{Spectr_3}%
\end{align}
where $k^{1}=-k_{1}-k_{2}$.

We notice here, the (\ref{Spectr_3}) shows that the cumulants of $Z_{\ell}%
^{m}$ and the function \newline$T_{4}\left(  \left.  \ell_{1},\ell_{2}%
,\ell_{3},\ell_{4}\right\vert \ell^{1}\right)  $ define each other uniquely.

\begin{remark}
\label{Trisp_principal} If we assume that $T_{4}\left(  \left.  \ell_{1}%
,\ell_{2},\ell_{3},\ell_{4}\right\vert \ell^{1}\right)  $ is known for
$\ell_{1}\leq\ell_{2}\leq\ell_{3}\leq\ell_{4}$, then the cumulants
$\operatorname*{Cum}\nolimits_{4}\left(  Z_{\ell_{1}}^{m_{1}},Z_{\ell_{2}%
}^{m_{2}},Z_{\ell_{3}}^{m_{3}},Z_{\ell_{4}}^{m_{4}}\right)  $ are given for
any $\left(  \ell_{1},\ell_{2},\ell_{3},\ell_{4}\right)  $. \newline Vice
versa $T_{4}\left(  \left.  \ell_{1},\ell_{2},\ell_{3},\ell_{4}\right\vert
\ell^{1}\right)  $ can be calculated for any $\left(  \ell_{1},\ell_{2}%
,\ell_{3},\ell_{4}\right)  $ by (\ref{Spectr_3}).
\end{remark}

The sum in (\ref{Spectr_3}) works for $k_{1},k_{2},k_{3},k_{4}$ and
$k^{1}=k_{1}+k_{2}=-k_{3}-k_{4}$. Hence the summation contains those $k_{m}$'s
when equation $k_{1}+k_{2}+k_{3}+k_{4}=0$, fulfils. Similarly \newline%
$\operatorname*{Cum}\nolimits_{4}\left(  Z_{\ell_{1}}^{m_{1}},Z_{\ell_{2}%
}^{m_{2}},Z_{\ell_{3}}^{m_{3}},Z_{\ell_{4}}^{m_{4}}\right)  =0$, if
$m_{1}+m_{2}+m_{3}+m_{4}\neq0$. The triangular inequality for $\ell_{1}%
,\ell_{2},\ell^{1}$ and $\ell_{3},\ell_{4},\ell^{1}$ suggests that the
'quadrilateral' with edges $\left(  \ell_{1},\ell_{2},\ell_{3},\ell
_{4}\right)  $ consists of two triangles $\ell_{1},\ell_{2},\ell^{1}$ and
$\ell_{3},\ell_{4},\ell^{1}$. The cumulants are invariant of the order of
their variables therefore this should be true for the other diagonal as well
$\ell_{4},\ell_{1},\ell$ and $\ell_{2},\ell_{3},\ell$. Actually since the left
hand side of (\ref{Cum4_Z}) is symmetric in $\ell_{1},\ell_{2},\ell_{3}%
,\ell_{4}$, the right hand side does not depend of the choice of the diagonal.

Note that parity transformation (\ref{Parity}) implies that $\ell_{1}+\ell
_{2}+\ell_{3}+\ell_{4}$ must be even and we turn to the field $u_{\ell}$, see
\ref{FieldAngMom},
\begin{multline*}
\operatorname*{Cum}\nolimits_{4}\left(  u_{\ell_{1}}\left(  L_{1}\right)
,u_{\ell_{2}}\left(  L_{2}\right)  ,u_{\ell_{3}}\left(  L_{3}\right)
,u_{\ell_{4}}\left(  L_{4}\right)  \right)  \\
=\sum_{m_{1:4}}\prod\limits_{j=1}^{4}Y_{\ell_{j}}^{m_{j}}\left(  L_{j}\right)
\operatorname*{Cum}\nolimits_{4}\left(  Z_{\ell_{1}}^{m_{1}},Z_{\ell_{2}%
}^{m_{2}},Z_{\ell_{3}}^{m_{3}},Z_{\ell_{4}}^{m_{4}}\right)  \\
=\sum_{\ell^{1},m^{1}}T_{4}\left(  \left.  \ell_{1},\ell_{2},\ell_{3},\ell
_{4}\right\vert \ell^{1}\right)  \sqrt{2\ell^{1}+1}\sum_{m_{1:4}}%
\prod\limits_{j=1}^{4}Y_{\ell_{j}}^{m_{j}}\left(  L_{j}\right)
\begin{pmatrix}
\ell_{1} & \ell_{2} & \ell^{1}\\
m_{1} & m_{2} & -m^{1}%
\end{pmatrix}%
\begin{pmatrix}
\ell^{1} & \ell_{3} & \ell_{4}\\
m^{1} & m_{3} & m_{4}%
\end{pmatrix}
.
\end{multline*}
The cumulant in the left hand side is symmetric in $\ell_{1},\ell_{2},\ell
_{3},\ell_{4}$, and invariant under rotation. The function
\begin{align*}
\widetilde{I}_{\ell_{1},\ell_{2},\ell_{3},\ell_{4},\ell^{1}}\left(
L_{1},L_{2},L_{3},L_{4}\right)   &  =\sqrt{2\ell^{1}+1}\sum_{m_{1},\ldots
m_{4},m^{1}}\prod\limits_{j=1}^{4}Y_{\ell_{j}}^{m_{j}}\left(  L_{j}\right)  \\
&  \times%
\begin{pmatrix}
\ell_{1} & \ell_{2} & \ell^{1}\\
m_{1} & m_{2} & -m^{1}%
\end{pmatrix}%
\begin{pmatrix}
\ell^{1} & \ell_{3} & \ell_{4}\\
m^{1} & m_{3} & m_{4}%
\end{pmatrix}
,
\end{align*}
is invariant under the rotation, since if we apply a rotation $g$ on each
$L_{j}$, then we can use Lemma \ref{Lemma_4D} Appendix to show that the right
hand side does not change. Therefore we apply the rotation $g_{L_{4}L_{3}}$,
which takes the location $L_{4}$ to the North pole and at the same time takes
$L_{3}$ into the plane $xOz$,
\begin{multline*}
\operatorname*{Cum}\nolimits_{4}\left(  u_{\ell_{1}}\left(  L_{1}\right)
,u_{\ell_{2}}\left(  L_{2}\right)  ,u_{\ell_{3}}\left(  L_{3}\right)
,u_{\ell_{4}}\left(  L_{4}\right)  \right)  \\
=\operatorname*{Cum}\nolimits_{4}\left(  u_{\ell_{1}}\left(  g_{L_{4}L_{3}%
}L_{1}\right)  ,u_{\ell_{2}}\left(  g_{L_{4}L_{3}}L_{2}\right)  ,u_{\ell_{3}%
}\left(  g_{L_{4}L_{3}}L_{3}\right)  ,u_{\ell_{4}}\left(  N\right)  \right)
\\
=\sqrt{\frac{2\ell_{4}+1}{4\pi}}\sum_{\ell^{1}}T_{4}\left(  \left.  \ell
_{1},\ell_{2},\ell_{3},\ell_{4}\right\vert \ell^{1}\right)  \sqrt{2\ell^{1}%
+1}\\
\times\sum_{m_{1},m_{2},m_{3},m^{1}}\prod\limits_{j=1}^{3}Y_{\ell_{j}}^{m_{j}%
}\left(  g_{L_{4}L_{3}}L_{j}\right)
\begin{pmatrix}
\ell_{1} & \ell_{2} & \ell^{1}\\
m_{1} & m_{2} & -m^{1}%
\end{pmatrix}%
\begin{pmatrix}
\ell^{1} & \ell_{3} & \ell_{4}\\
m^{1} & m_{3} & 0
\end{pmatrix}
,
\end{multline*}
we have%
\begin{multline*}
\operatorname*{Cum}\nolimits_{4}\left(  u_{\ell_{1}}\left(  L_{1}\right)
,u_{\ell_{2}}\left(  L_{2}\right)  ,u_{\ell_{3}}\left(  L_{3}\right)
,u_{\ell_{4}}\left(  L_{4}\right)  \right)  \\
=\sum_{\ell^{1}}T_{4}\left(  \left.  \ell_{1},\ell_{2},\ell_{3},\ell
_{4}\right\vert \ell^{1}\right)  \mathcal{I}_{\ell_{1},\ell_{2},\ell_{3}%
,\ell_{4},\ell^{1}}\left(  \vartheta_{1},\vartheta_{2},\vartheta_{3}%
,\varphi_{1},\varphi_{2}\right)  ,
\end{multline*}
where the function $\mathcal{I}_{\ell_{1},\ell_{2},\ell_{3},\ell_{4},\ell^{1}%
}$ is given by
\begin{align*}
\mathcal{I}_{\ell_{1},\ell_{2},\ell_{3},\ell_{4},\ell^{1}}\left(
\vartheta_{1},\vartheta_{2},\vartheta_{3},\varphi_{1},\varphi_{2}\right)   &
=\sqrt{\frac{2\ell_{4}+1}{4\pi}\left(  2\ell^{1}+1\right)  }\sum
_{m_{1:3},m^{1}}\prod\limits_{j=1}^{3}Y_{\ell_{j}}^{m_{j}}\left(
g_{L_{4}L_{3}}L_{j}\right)  \\
&  \times%
\begin{pmatrix}
\ell_{1} & \ell_{2} & \ell^{1}\\
m_{1} & m_{2} & -m^{1}%
\end{pmatrix}%
\begin{pmatrix}
\ell^{1} & \ell_{3} & \ell_{4}\\
m^{1} & m_{3} & 0
\end{pmatrix}
.
\end{align*}
Actually $\mathcal{I}_{\ell_{1},\ell_{2},\ell_{3},\ell_{4},\ell^{1}}$ is the
rotated version of $\widetilde{I}_{\ell_{1},\ell_{2},\ell_{3},\ell_{4}%
,\ell^{1}}$, it corresponds to the tripolar spherical harmonics\textbf{\ }%
defined as \textbf{\ }irreducible tensor products of the spherical harmonics
with different arguments \cite{Varshalovich1988}, p.160. We define the
spherical coordinates $\left(  \vartheta_{1},\vartheta_{2},\vartheta
_{3},\varphi_{1},\varphi_{2}\right)  $ as follows $g_{L_{4}L_{3}}L_{3}=\left(
\vartheta_{3},0\right)  $, $g_{L_{4}L_{3}}L_{j}=\left(  \vartheta_{j}%
,\varphi_{j}\right)  $, $j=1,2$. Note the orthonormality of the system
according to the measure $\prod_{k=1}^{3}\Omega\left(  dL_{k}\right)
=\prod_{k=1}^{3}\sin\vartheta_{k}d\vartheta_{k}d\varphi_{k}$, $\vartheta
_{k}\in\left[  0,\pi\right]  $, $\varphi_{k}\in\left[  0,2\pi\right]  $, i.e.
\[
\iiint\limits_{\mathbb{S}_{2}}\mathcal{I}_{\ell_{1},\ell_{2},\ell_{3},\ell
_{4},\ell^{1}}^{\ast}\mathcal{I}_{j_{1},j_{2},j_{3},j_{4},\ell^{2}}\prod
_{k=1}^{3}\Omega\left(  dL_{k}\right)  =\delta_{\ell^{1},\ell^{2}}%
\prod_{k=1:4}\delta_{\ell_{k},j_{k}},
\]
since collecting the coefficients after integration, we have
\[
\left(  2\ell_{4}+1\right)  \sum_{m^{1},m_{3}}%
\begin{pmatrix}
\ell^{1} & \ell_{3} & \ell_{4}\\
m^{1} & m_{3} & 0
\end{pmatrix}
^{2}\left(  2\ell^{1}+1\right)  \sum_{m_{1},m_{2},m^{2}}%
\begin{pmatrix}
\ell_{1} & \ell_{2} & \ell^{1}\\
m_{1} & m_{2} & -m^{2}%
\end{pmatrix}
^{2}=1.
\]
The expression for the $\operatorname*{Cum}\nolimits_{4}\left(  X\left(
L_{1}\right)  ,X\left(  L_{2}\right)  ,X\left(  L_{3}\right)  ,X\left(
L_{4}\right)  \right)  =\mathcal{C}_{4}\left(  \vartheta_{1},\vartheta
_{2},\vartheta_{3},\varphi_{1},\varphi_{2}\right)  $ is straightforward
\begin{align}
\mathcal{C}_{4}\left(  \vartheta_{1},\vartheta_{2},\vartheta_{3},\varphi
_{1},\varphi_{2}\right)    & =\sum_{\ell_{1},\ell_{2},\ell_{3},\ell_{4}%
=0}^{\infty}\operatorname*{Cum}\nolimits_{4}\left(  u_{\ell_{1}}\left(
L_{1}\right)  ,\ldots,u_{\ell_{4}}\left(  L_{4}\right)  \right)  \nonumber\\
& =\sum_{\ell_{1},\ldots,\ell_{4}=0}^{\infty}\sum_{\ell^{1}}T_{4}\left(
\left.  \ell_{1},\ell_{2},\ell_{3},\ell_{4}\right\vert \ell^{1}\right)
\mathcal{I}_{\ell_{1},\ell_{2},\ell_{3},\ell_{4},\ell^{1}}\left(  L_{1}%
,\ldots,L_{4}\right)  .\label{Expansion_tri}%
\end{align}

\begin{definition}
The trispectrum of the isotropic field $X\left(  L\right)  $ is given by%
\[
S_{4}\left(  \left.  \ell_{1},\ell_{2},\ell_{3},\ell_{4}\right\vert \ell
^{1}\right)  =T_{4}\left(  \left.  \ell_{1},\ell_{2},\ell_{3},\ell
_{4}\right\vert \ell^{1}\right)  .
\]

\end{definition}

\begin{theorem}
The fourth order cumulant $\operatorname*{Cum}\nolimits_{4}\left(  X\left(
L_{1}\right)  ,X\left(  L_{2}\right)  ,X\left(  L_{3}\right)  ,X\left(
L_{4}\right)  \right)  $ of the isotropic field $X\left(  L\right)  $ have the
series expansion (\ref{Expansion_tri}) in terms of the trispectrum
$S_{4}\left(  \left.  \ell_{1},\ell_{2},\ell_{3},\ell_{4}\right\vert \ell
^{1}\right)  $ and orthonormed system $\mathcal{I}_{\ell_{1},\ell_{2},\ell
_{3},\ell_{4},\ell^{1}}$, hence
\[
S_{4}\left(  \left.  \ell_{1},\ell_{2},\ell_{3},\ell_{4}\right\vert \ell
^{1}\right)  =\iiint\limits_{\mathbb{S}_{2}}\mathcal{C}_{4}\left(
\vartheta_{1},\vartheta_{2},\vartheta_{3},\varphi_{1},\varphi_{2}\right)
\mathcal{I}_{\ell_{1},\ell_{2},\ell_{3},\ell_{4},\ell^{1}}\left(
\vartheta_{1},\vartheta_{2},\vartheta_{3},\varphi_{1},\varphi_{2}\right)
\prod_{k=1}^{3}\Omega\left(  dL_{k}\right)  .
\]

\end{theorem}

In particular, we have formulae for the marginal fourth order cumulants%
\begin{align*}
\operatorname*{Cum}\nolimits_{4}\left(  X\left(  L_{1}\right)  ,X\left(
L_{2}\right)  ,X\left(  L\right)  ,X\left(  L\right)  \right)   &  =\sum
_{\ell_{1},\ell_{2},\ell_{3},\ell_{4}=0}^{\infty}\sqrt{\frac{2\ell_{3}+1}%
{4\pi}}\sum_{\ell^{1}}%
\begin{pmatrix}
\ell_{3} & \ell_{4} & \ell^{1}\\
0 & 0 & 0
\end{pmatrix}
\\
&  \times T_{4}\left(  \left.  \ell_{1},\ell_{2},\ell_{3},\ell_{4}\right\vert
\ell^{1}\right)  \mathcal{I}_{\ell_{1},\ell_{2},\ell^{1}}\left(  \vartheta
_{1},\varphi_{1},\vartheta_{2}\right)  ,
\end{align*}
and
\begin{multline*}
\operatorname*{Cum}\nolimits_{4}\left(  X\left(  L_{1}\right)  ,X\left(
L\right)  ,X\left(  L\right)  ,X\left(  L\right)  \right) \\
=\sum_{\ell_{1},\ell_{2},\ell_{3},\ell_{4}=0}^{\infty}\sqrt{\prod_{j=1}%
^{4}\frac{2\ell_{j}+1}{4\pi}}\sum_{\ell^{1}}T_{4}\left(  \left.  \ell
_{1},\ldots,\ell_{4}\right\vert \ell^{1}\right) \\
\times\sqrt{2\ell^{1}+1}%
\begin{pmatrix}
\ell_{1} & \ell_{2} & \ell^{1}\\
0 & 0 & 0
\end{pmatrix}%
\begin{pmatrix}
\ell^{1} & \ell_{3} & \ell_{4}\\
0 & 0 & 0
\end{pmatrix}
P_{\ell_{1}}\left(  L\cdot L_{1}\right)
\end{multline*}
also%
\begin{multline*}
\operatorname*{Cum}\nolimits_{4}\left(  X\left(  L\right)  ,X\left(  L\right)
,X\left(  L\right)  ,X\left(  L\right)  \right) \\
=\sum_{\ell_{1},\ell_{2},\ell_{3},\ell_{4}=0}^{\infty}\sqrt{\prod_{j=1}%
^{4}\frac{2\ell_{i}+1}{4\pi}}\sum_{\ell^{1}}T_{4}\left(  \left.  \ell_{1}%
,\ell_{2},\ell_{3},\ell_{4}\right\vert \ell^{1}\right)  \sqrt{2\ell^{1}+1}%
\begin{pmatrix}
\ell_{1} & \ell_{2} & \ell^{1}\\
0 & 0 & 0
\end{pmatrix}%
\begin{pmatrix}
\ell^{1} & \ell_{3} & \ell_{4}\\
0 & 0 & 0
\end{pmatrix}
.
\end{multline*}

\subsection{Linear field}

If the triangular array $Z_{\ell}^{m}$ contains \textbf{\ }independent
variables, i.e. the angular momentum field $u_{\ell}$ is linear. Then for a
fixed degree $\ell$
\begin{align*}
\operatorname*{Cum}\nolimits_{4}\left(  Z_{\ell}^{m_{1}},Z_{\ell}^{m_{2}%
},Z_{\ell}^{m_{3}},Z_{\ell}^{m_{4}}\right)   &  =\delta_{m_{i}=m}\sum
_{\ell^{1},m^{1}}%
\begin{pmatrix}
\ell_{1} & \ell_{2} & \ell^{1}\\
m_{1} & m_{2} & m^{1}%
\end{pmatrix}%
\begin{pmatrix}
\ell^{1} & \ell_{3} & \ell_{4}\\
-m^{1} & m_{3} & m_{4}%
\end{pmatrix}
\\
&  \times\left(  -1\right)  ^{m^{1}}\sqrt{2\ell^{1}+1}T_{4}\left(  \left.
\ell,\ell,\ell,\ell\right\vert \ell^{1}\right)  ,
\end{align*}
from the selection rules follows $m^{1}=-\left(  m_{1}+m_{2}\right)  =-2m$,
$m^{1}=m_{3}+m_{4}=2m$, therefore $m_{i}=0$. Now for similar reason
\[
T_{4}\left(  \left.  \ell,\ell,\ell,\ell\right\vert \ell^{1}\right)
=\sqrt{2\ell^{1}+1}\operatorname*{Cum}\nolimits_{4}\left(  Z_{\ell}%
^{0},Z_{\ell}^{0},Z_{\ell}^{0},Z_{\ell}^{0}\right)
\begin{pmatrix}
\ell & \ell & \ell^{1}\\
0 & 0 & 0
\end{pmatrix}%
\begin{pmatrix}
\ell^{1} & \ell & \ell\\
0 & 0 & 0
\end{pmatrix}
.
\]
The only nonzero cumulants are
\[
\operatorname*{Cum}\nolimits_{4}\left(  Z_{\ell}^{0},Z_{\ell}^{0},Z_{\ell}%
^{0},Z_{\ell}^{0}\right)  =\sum_{\ell^{1}}\sqrt{2\ell^{1}+1}%
\begin{pmatrix}
\ell & \ell & \ell^{1}\\
0 & 0 & 0
\end{pmatrix}
^{2}%
\begin{pmatrix}
\ell^{1} & \ell & \ell\\
0 & 0 & 0
\end{pmatrix}
^{2}\operatorname*{Cum}\nolimits_{4}\left(  Z_{\ell}^{0},Z_{\ell}^{0},Z_{\ell
}^{0},Z_{\ell}^{0}\right)  ,
\]
hence for an isotropic linear field $u_{\ell}$ $\operatorname*{Cum}%
\nolimits_{4}\left(  Z_{\ell}^{m_{1}},Z_{\ell}^{m_{2}},Z_{\ell}^{m_{3}%
},Z_{\ell}^{m_{4}}\right)  =0,$ for all $\ell$ and $m_{i}$. If additionally
the series of $u_{\ell}$ is independent then $\operatorname*{Cum}%
\nolimits_{4}\left(  X\left(  L_{1}\right)  ,X\left(  L_{2}\right)  ,X\left(
L_{3}\right)  ,X\left(  L_{4}\right)  \right)  =0$ follows.

\section{Higher order spectra for isotropic fields}

The generalization of the bispectrum and trispectrum is possible for arbitrary
higher order \cite{Marinucci2010}, \cite{Marinucci2011}. If the bispectrum and
trispectrum are zero then the field can not be Gaussian, in case all
polyspectra, except the second order one, are zero then the isotropic field is
necessary Gaussian. First we show that the characterization of the isotropy in
$p^{th}$ order can be given similarly to the previous cases.

\begin{lemma}
\label{Lemm_Iso_p}The field
\[
X\left(  L\right)  =\sum_{\ell=0}^{\infty}\sum_{m=-\ell}^{\ell}Z_{\ell}%
^{m}Y_{\ell}^{m}\left(  L\right)  ,
\]
is isotropic in $p^{th}$ order $\left(  p>3\right)  $ iff the cumulant
$\operatorname*{Cum}\nolimits_{p}\left(  Z_{\ell_{1}}^{m_{1}},Z_{\ell_{2}%
}^{m_{2}},\ldots,Z_{\ell_{p}}^{m_{p}}\right)  $ of the triangular array
$Z_{\ell}^{m}$ has the form
\begin{align}
\operatorname*{Cum}\nolimits_{p}\left(  Z_{\ell_{1}}^{m_{1}},Z_{\ell_{2}%
}^{m_{2}},\ldots,Z_{\ell_{p}}^{m_{p}}\right)   &  =\sum_{\substack{\ell
^{1},\ldots,\ell^{p-3} \\m^{1},\ldots,m^{p-3}}}\left(  -1\right)
^{\Sigma_{a=1}^{p-3}m^{a}}\prod\limits_{a=0}^{p-3}%
\begin{pmatrix}
\ell^{a} & \ell_{a+2} & \ell^{a+1}\\
-m^{a} & m_{a+2} & m^{a+1}%
\end{pmatrix}
\label{CumP_Z}\\
&  \times\prod\limits_{a=1}^{p-3}\sqrt{2\ell^{a}+1}\widetilde{S}_{p}\left(
\left.  \ell_{1},\ldots,\ell_{p}\right\vert \ell^{1},\ldots,\ell^{p-3}\right)
,\nonumber
\end{align}
where $\ell^{0}=\ell_{1}$, $\ell^{p-2}=\ell_{p}$, $k^{0}=-k_{1}$,
$k^{p-2}=k_{p}$, $m^{0}=-m_{1}$, $m^{p-2}=m_{p}$. The function $\widetilde{S}%
_{p}\left(  \left.  \ell_{1},\ldots,\ell_{p}\right\vert \ell^{1},\ldots
,\ell^{p-3}\right)  $ has the form
\begin{align}
\widetilde{S}_{p}\left(  \left.  \ell_{1},\ldots,\ell_{p}\right\vert \ell
^{1},\ldots,\ell^{p-3}\right)   &  =\sum_{\substack{k_{1},\ldots,k_{p}
\\k^{1},\ldots,k^{p-3}}}\left(  -1\right)  ^{\Sigma_{a=1}^{p-3}k^{a}}%
\prod\limits_{a=0}^{p-3}%
\begin{pmatrix}
\ell^{a} & \ell_{a+2} & \ell^{a+1}\\
-k^{a} & k_{a+2} & k^{a+1}%
\end{pmatrix}
\label{Spectr_p}\\
&  \times\prod\limits_{a=1}^{p-3}\sqrt{2\ell^{a}+1}\operatorname*{Cum}%
\nolimits_{p}\left(  Z_{\ell_{1}}^{k_{1}},Z_{\ell_{2}}^{k_{2}},\ldots
,Z_{\ell_{p}}^{k_{p}}\right)  .\nonumber
\end{align}

\end{lemma}

See the Appendix \ref{Append_PolySp} for the proof. Consider a convex polygon
with vertices $A_{1:p-1}$, and edges $\ell_{1:p}$. The diagonals denoted by
$\ell^{j}$, $j=1,2,\ldots,p-3$, starting from the vertex $A_{1}$ divide this
polygon into $p-2$ triangles. The first triangle has sides (angular mometum)
$\ell_{1}$, $\ell_{2}$ and $\ell^{1}$, the next one has sides $\ell^{1}$,
$\ell_{3}$ and $\ell^{2}$, the general one is $\ell^{a} $, $\ell_{a+2}$ and
$\ell^{a+1}$, finally the last one $\ell^{p-3}$, $\ell_{p-1}$ and $\ell_{p}$.
For each $a$ the sides of the triangle $\left(  \ell^{a},\ell_{a+2},\ell
^{a+1}\right)  $ should fulfil the triangle inequality $\left\vert \ell
^{a}-\ell^{a+1}\right\vert \leq\ell_{a+2}\leq\ell^{a}+\ell^{a+1}$. The
coefficients in (\ref{CumP_Z}) will differ from zero if orders $-m^{a}%
,m_{a+2},m^{a+1}$, fulfil the assumption $-m^{a}+m_{a+2}+m^{a+1}=0$, for all
$a$. This implies $m_{1}+m_{2}=-m^{1}$, $m_{3}+m^{2}=m^{1}$, \ldots,
$m_{p-1}+m_{p}=m^{p-3}$. Let us plug in consecutively $m^{a}$ and we shall
arrive to the result $m_{1}+m_{2}+\cdots+m_{p}=0$. Hence $\operatorname*{Cum}%
\nolimits_{p}\left(  Z_{\ell_{1}}^{m_{1}},Z_{\ell_{2}}^{m_{2}},\ldots
,Z_{\ell_{p}}^{m_{p}}\right)  =0$, unless $m_{1}+m_{2}+\cdots+m_{p}=0$.

\begin{remark}
The cumulant $\operatorname*{Cum}\nolimits_{p}\left(  Z_{\ell_{1}}^{m_{1}%
},Z_{\ell_{2}}^{m_{2}},\ldots,Z_{\ell_{p}}^{m_{p}}\right)  $ is invariant
under the order of the quantum numbers $\ell_{1},\ell_{2},\ldots,\ell_{p}$,
hence the right hand side of (\ref{CumP_Z}) is invariant as well. The result
is that the function $\widetilde{S}_{p}\left(  \left.  \ell_{1},\ldots
,\ell_{p}\right\vert \ell^{1},\ldots,\ell^{p-3}\right)  $ given on values
$\ell_{1}\leq\ell_{2}\leq\cdots\leq\ell_{p}$ will determine all cumulants
$\operatorname*{Cum}\nolimits_{p}\left(  Z_{\ell_{1}}^{m_{1}},Z_{\ell_{2}%
}^{m_{2}},\ldots,Z_{\ell_{p}}^{m_{p}}\right)  $ by (\ref{CumP_Z}).
\end{remark}

Repeat the representation of the field
\[
X\left(  L\right)  =\sum_{\ell=0}^{\infty}u_{\ell}\left(  L\right)  ,
\]
where%
\[
u_{\ell}\left(  L\right)  =\sum_{m=-\ell}^{\ell}Y_{\ell}^{m}\left(  L\right)
Z_{\ell}^{m}.
\]
We are interested in the invariance of the distribution of $u_{\ell}\left(
L\right)  $ under rotations. If the location is fixed then $u_{\ell}\left(
L\right)  $ is connected to the distribution of $\mathcal{Z}_{\ell}^{k}$,
since the Condon and Shortley phase convention
(\ref{Convention_Condon_Shortley}) provides
\[
u_{\ell}\left(  L\right)  =\sqrt{\frac{2\ell+1}{4\pi}}\sum_{m=-\ell}^{\ell
}D_{m,0}^{\left(  \ell\right)  \ast}Z_{\ell}^{m}.
\]
Our main interest is the comparison of the finite dimensional distributions of
the process $u_{\ell}\left(  L\right)  $ to the rotated one in case the
rotation carries one location to the North pole $N$, since $u_{\ell}$
simplifies to
\[
u_{\ell}\left(  N\right)  =\sqrt{\frac{2\ell+1}{4\pi}}Z_{\ell}^{0}.
\]
We have
\begin{multline*}
\operatorname*{Cum}\nolimits_{p}\left(  u_{\ell_{1}}\left(  L_{1}\right)
,u_{\ell_{2}}\left(  L_{2}\right)  ,\ldots,u_{\ell_{p}}\left(  L_{p}\right)
\right)  =\sum_{m_{1},\ldots,m_{p}}\prod\limits_{j=1}^{p}Y_{\ell_{j}}^{m_{j}%
}\left(  L_{j}\right)  \operatorname*{Cum}\nolimits_{p}\left(  Z_{\ell_{1}%
}^{m_{1}},Z_{\ell_{2}}^{m_{2}},\ldots,Z_{\ell_{p}}^{m_{p}}\right) \\
=\sum_{m_{1},\ldots,m_{p}}\prod\limits_{j=1}^{p}Y_{\ell_{j}}^{m_{j}}\left(
L_{j}\right)  \sum_{\substack{\ell^{1},\ldots,\ell^{p-3}\\m^{1},\ldots
,m^{p-3}}}\left(  -1\right)  ^{\Sigma_{a=1}^{p-3}m^{a}}\prod\limits_{a=0}%
^{p-3}%
\begin{pmatrix}
\ell^{a} & \ell_{a+2} & \ell^{a+1}\\
-m^{a} & m_{a+2} & m^{a+1}%
\end{pmatrix}
\\
\times\prod\limits_{a=1}^{p-3}\sqrt{2\ell^{a}+1}\widetilde{S}_{p}\left(
\left.  \ell_{1},\ldots,\ell_{p}\right\vert \ell^{1},\ldots,\ell^{p-3}\right)
\\
=\sum_{\ell^{1},\ldots,\ell^{p-3}}\widetilde{I}_{\ell_{1:p},\ell^{1:p-3}%
}\left(  L_{1},\ldots,L_{p}\right)  \widetilde{S}_{p}\left(  \left.  \ell
_{1},\ldots,\ell_{p}\right\vert \ell^{1},\ldots,\ell^{p-3}\right)  .
\end{multline*}
The following $p-$product\textit{\ of spherical harmonics}\textbf{\ }
$Y_{\ell}^{m}$
\begin{multline*}
\widetilde{I}_{\ell_{1},\ldots,\ell_{p},\ell^{1},\ldots,\ell^{p-3}}\left(
L_{1},\ldots,L_{p}\right)  =\sum_{\substack{m_{1},\ldots,m_{p}\\m^{1}%
,\ldots,m^{p-3}}}\prod\limits_{j=1}^{p}Y_{\ell_{j}}^{m_{j}}\left(
L_{j}\right) \\
\times\left(  -1\right)  ^{\Sigma k^{1:p-3}}\prod\limits_{a=0}^{p-3}%
\begin{pmatrix}
\ell^{a} & \ell_{a+2} & \ell^{a+1}\\
-m^{a} & m_{a+2} & m^{a+1}%
\end{pmatrix}
\prod\limits_{a=1}^{p-3}\sqrt{2\ell^{a}+1},
\end{multline*}
is rotation invariant, see (\ref{RotationProd3}), here $k^{1:p-3}$ denotes
$\left(  k^{1},\ldots,k^{p-3}\right)  $ and $\Sigma k^{1:p-3}=%
{\textstyle\sum\nolimits_{j=1}^{p-3}}
k^{j}$, for short. Hence we can apply the rotation $g_{L_{p}L_{p-1}}$ such
that $g_{L_{p}L_{p-1}}L_{p}=N$, and $g_{L_{p}L_{p-1}}L_{p-1}$ belongs into the
plane $xOz$, we have
\begin{multline*}
I_{\ell_{1},\ldots,\ell_{p},\ell^{1},\ldots,\ell^{p-3}}\left(  L_{1}%
,\ldots,L_{p}\right)  =\sqrt{\frac{2\ell_{p}+1}{4\pi}}\sum_{\substack{m_{1}%
,\ldots,m_{p}\\m^{1},\ldots,m^{p-3}}}\prod\limits_{j=1}^{p-1}Y_{\ell_{j}%
}^{m_{j}}\left(  g_{L_{p}L_{p-1}}L_{j}\right) \\
\times\left(  -1\right)  ^{\Sigma k^{1:p-3}}\prod\limits_{a=0}^{p-3}%
\begin{pmatrix}
\ell^{a} & \ell_{a+2} & \ell^{a+1}\\
-m^{a} & m_{a+2} & m^{a+1}%
\end{pmatrix}
\prod\limits_{a=1}^{p-3}\sqrt{2\ell^{a}+1}%
\begin{pmatrix}
\ell^{p-3} & \ell_{p-1} & \ell_{p}\\
-m^{p-3} & m_{p-1} & 0
\end{pmatrix}
.
\end{multline*}
According to the usual measure $\prod_{k=1}^{p-1}\Omega\left(  dL_{k}\right)
=\prod_{k=1}^{p-1}\sin\vartheta_{k}d\vartheta_{k}d\varphi_{k}$, $\vartheta
_{k}\in\left[  0,\pi\right]  $, $\varphi_{k}\in\left[  0,2\pi\right]  $, the
system of functions $I_{\ell_{1},\ldots,\ell_{p},\ell^{1},\ldots,\ell^{p-3}}$
forms an orthogonal system.

Consider now
\begin{align*}
\widetilde{S}_{p}\left(  \left.  \ell_{1},\ldots,\ell_{p}\right\vert \ell
^{1},\ldots,\ell^{p-3}\right)   &  =\sum_{\substack{m_{1},\ldots,m_{p}
\\m^{1},\ldots,m^{p-3}}}\left(  -1\right)  ^{\Sigma_{a=1}^{p-3}m^{a}}%
\prod\limits_{a=0}^{p-3}%
\begin{pmatrix}
\ell^{a} & \ell_{a+2} & \ell^{a+1}\\
m^{a} & m_{a+2} & m^{a+1}%
\end{pmatrix}
\\
&  \times\prod\limits_{a=1}^{p-3}\sqrt{2\ell^{a}+1}\operatorname*{Cum}%
\nolimits_{p}\left(  Z_{\ell_{1}}^{m_{1}},Z_{\ell_{2}}^{m_{2}},\ldots
,Z_{\ell_{p}}^{m_{p}}\right)
\end{align*}
where remember $\ell^{0}=\ell_{1}$, and $\ell^{p-2}=\ell_{p}$. The left size
is symmetric in $\ell_{1},\ldots,\ell_{p}$ hence the right size must be
symmetric as well. One might fix an order for the entries of $\ell_{1}%
,\ldots,\ell_{p}$ to get a unique representation for the cumulant. We consider
a monotone ordering $\ell_{1}\leq\ell_{2}\leq\ldots\leq\ell_{p}$ and refer to
it as canonical representation.

Note that parity transformation (\ref{Parity}) implies that $\ell_{1}+\ell
_{2}+\ell_{3}+\ldots+\ell_{p}$ must be even.%
\begin{align}
\operatorname*{Cum}\nolimits_{p}\left(  X\left(  L_{1}\right)  ,X\left(
L_{2}\right)  ,X\left(  L_{3}\right)  ,\ldots,X\left(  L_{p}\right)  \right)
& =\sum_{\ell_{1:p}=0}^{\infty}\operatorname*{Cum}\nolimits_{p}\left(
u_{\ell_{1}}\left(  L_{1}\right)  ,u_{\ell_{2}}\left(  L_{2}\right)
,\ldots,u_{\ell_{p}}\left(  L_{p}\right)  \right)  \nonumber\\
& =\sum_{\ell_{1:p}=0}^{\infty}\widetilde{S}_{p}\left(  \left.  \ell
_{1},\ldots,\ell_{p}\right\vert \ell^{1},\ldots,\ell^{p-3}\right)  I_{\ell
_{1},\ldots,\ell_{p},\ell^{1},\ldots,\ell^{p-3}}\left(  L_{1},\ldots
,L_{p}\right)  .\label{Expansion_p}%
\end{align}

\begin{definition}
The p$^{th}$ order polyspectrum of the isotropic field $X\left(  L\right)  $
is \newline$S_{p}\left(  \left.  \ell_{1},\ldots,\ell_{p}\right\vert \ell
^{1},\ldots,\ell^{p-3}\right)  =\widetilde{S}_{p}\left(  \left.  \ell
_{1},\ldots,\ell_{p}\right\vert \ell^{1},\ldots,\ell^{p-3}\right)  $.
\end{definition}

\begin{theorem}
The p$^{th}$ order cumulant $\operatorname*{Cum}\nolimits_{p}\left(  X\left(
L_{1}\right)  ,X\left(  L_{2}\right)  ,\ldots,X\left(  L_{p}\right)  \right)
$ of the isotropic field $X\left(  L\right)  $ have the series expansion
(\ref{Expansion_p}) in terms of the polyspectrum $S_{p}\left(  \left.
\ell_{1},\ldots,\ell_{p}\right\vert \ell^{1},\ldots,\ell^{p-3}\right)  $ and
orthonormed system $\mathcal{I}_{\left.  \ell_{1},\ldots,\ell_{p}\right\vert
\ell^{1},\ldots,\ell^{p-3}}$, hence
\[
S_{p}\left(  \left.  \ell_{1},\ldots,\ell_{p}\right\vert \ell^{1},\ldots
,\ell^{p-3}\right)  =%
{\textstyle\idotsint_{\mathbb{S}_{2}}}
\mathcal{C}_{p}\left(  \vartheta_{1:p-1},\varphi_{1:p-3}\right)
\mathcal{I}_{\left.  \ell_{1},\ldots,\ell_{p}\right\vert \ell^{1},\ldots
,\ell^{p-3}}\left(  \vartheta_{1:p-1},\varphi_{1:p-3}\right)  \prod
_{k=1}^{p-1}\Omega\left(  dL_{k}\right)  .
\]

\end{theorem}

\subsection{Linear field}

Let us consider the particular case when $Z_{\ell}^{m}$ are independent if
$\ell$ is fixed. We have%

\begin{align*}
\operatorname*{Cum}\nolimits_{p}\left(  Z_{\ell}^{m_{1}},Z_{\ell}^{m_{2}%
},\ldots,Z_{\ell}^{m_{p}}\right)   &  =\delta_{m_{i}=m}\sum_{\substack{\ell
^{1},\ldots,\ell^{p-3} \\m^{1},\ldots,m^{p-3}}}\left(  -1\right)
^{\Sigma_{a=1}^{p-3}m^{a}}\prod\limits_{a=0}^{p-3}%
\begin{pmatrix}
\ell^{a} & \ell & \ell^{a+1}\\
-m^{a} & m_{a+2} & m^{a+1}%
\end{pmatrix}
\\
&  \times S_{p}\left(  \left.  \ell\right\vert \ell^{1},\ldots,\ell
^{p-3}\right)  \prod\limits_{a=1}^{p-3}\sqrt{2\ell^{a}+1}.
\end{align*}
We have seen that for $p=3,4$, all $m_{i}=0$ follows. Let us see the case
$p=5$,
\[%
\begin{pmatrix}
\ell & \ell^{1}\\
m_{1:2} & m^{1}%
\end{pmatrix}%
\begin{pmatrix}
\ell^{1} & \ell & \ell^{2}\\
-m^{1} & m_{3} & m^{2}%
\end{pmatrix}%
\begin{pmatrix}
\ell^{2} & \ell & \ell\\
-m^{2} & m_{4} & m_{5}%
\end{pmatrix}
\]
then from the selection rules follows that $m^{1}=-\left(  m_{1}+m_{2}\right)
=-2m$, $m^{2}=m^{1}-m_{3}=-3m$, $m^{2}=m_{4}+m_{5}=2m$, hence $m=0.$ In
general it is easy to see that $m^{k}=-\left(  k+1\right)  m$, for
$k=1,2,\ldots,p-3$, and at the same time $m^{p-3}=2m$, hence $m=0$, and
$m^{j}=0$ as well.

Consider the polyspectrum
\[
\widetilde{S}_{p}\left(  \left.  \ell\right\vert \ell^{1:p-3}\right)
=\sum_{\substack{k_{1},\ldots,k_{p} \\k^{1},\ldots,k^{p-3}}}\left(  -1\right)
^{\Sigma k^{1:p-2}}\prod\limits_{a=0}^{p-3}%
\begin{pmatrix}
\ell^{a} & \ell & \ell^{a+1}\\
-k^{a} & k_{a+2} & k^{a+1}%
\end{pmatrix}
\prod\limits_{a=1}^{p-3}\sqrt{2\ell^{a}+1}\operatorname*{Cum}\nolimits_{p}%
\left(  Z_{\ell}^{k_{1:p}}\right)  ,
\]
from the independence follows that $\ell_{j}=\ell$ and $k_{j}=k$, for
$j=1,2,\ldots,p$. Hence similar argument to the previous one leads to the
result: $k_{j}=0$, and $k^{j}=0$ as well. We have
\[
\widetilde{S}_{p}\left(  \left.  \ell\right\vert \ell^{1:p-3}\right)
=\prod\limits_{a=0}^{p-3}%
\begin{pmatrix}
\ell^{a} & \ell & \ell^{a+1}\\
0 & 0 & 0
\end{pmatrix}
\prod\limits_{a=1}^{p-3}\sqrt{2\ell^{a}+1}\operatorname*{Cum}\nolimits_{p}%
\left(  Z_{\ell}^{0},Z_{\ell}^{0},\ldots,Z_{\ell}^{0}\right)  .
\]
Now, by the Lemma \ref{Lemm_Iso_p} we get
\[
\operatorname*{Cum}\nolimits_{p}\left(  Z_{\ell}^{0},Z_{\ell}^{0}%
,\ldots,Z_{\ell}^{0}\right)  =\sum_{\ell^{1},\ldots,\ell^{p-3}}\prod
\limits_{a=0}^{p-3}%
\begin{pmatrix}
\ell^{a} & \ell & \ell^{a+1}\\
0 & 0 & 0
\end{pmatrix}
^{2}\prod\limits_{a=1}^{p-3}\sqrt{2\ell^{a}+1}\operatorname*{Cum}%
\nolimits_{p}\left(  Z_{\ell}^{0},Z_{\ell}^{0},\ldots,Z_{\ell}^{0}\right)  .
\]
Hence $\operatorname*{Cum}\nolimits_{p}\left(  Z_{\ell_{1}}^{k_{1}}%
,Z_{\ell_{2}}^{k_{2}},\ldots,Z_{\ell_{p}}^{k_{p}}\right)  =0$, for all $k_{j}%
$. Now we have a general conclusion because the only case when all cumulants
vanish except the second order one is the Gaussian. Once the rows of $\left\{
Z_{\ell}^{m}\right\}  $ are Gaussian then all the entries of $\left\{
Z_{\ell}^{m}\right\}  $ are independent. Indeed the isotropy implies that all
the entries of $\left\{  Z_{\ell}^{m}\right\}  $ are uncorraleted now they are
Gaussian hence they are independent.

\begin{lemma}
If the isotropic field $X\left(  L\right)  $ is linear, i.e. inside the rows
of the generating array $Z_{\ell}^{m}$ all random variables are independent,
then the whole array contains independent entries and $X\left(  L\right)  $ is Gaussian.
\end{lemma}

\section{Construction of isotropic field}

The stochastic isotropic field%
\[
X\left(  L\right)  =\sum_{\ell=0}^{\infty}\sum_{m=-\ell}^{\ell}Z_{\ell}%
^{m}Y_{\ell}^{m}\left(  L\right)  ,
\]
on sphere has very special angular spectra, in particular the form of the
cumulants of the triangular array $\left\{  Z_{\ell}^{m}\right\}  $ is
restricted, see (\ref{CumP_Z}). If one starts with a non-Gaussian continuous
in mean square $X\left(  L\right)  $ then the triangular array $\left\{
Z_{\ell}^{m}\right\}  $ is naturally given by the inversion formula
\[
Z_{\ell}^{m}=\int_{\mathbb{S}_{2}}X\left(  L\right)  Y_{\ell}^{m\ast}\left(
L\right)  dL.
\]
Now, we address the reverse question of construction of a triangular array
$\left\{  Z_{\ell}^{m}\right\}  $ with the desired properties of their
cumulants. Let us start with triangular array $\left\{  \widetilde{Z}_{\ell
}^{m}\right\}  $, assume it is uncorrelated and all moments exist. Consider
the vectors $\widetilde{Z}_{\ell}=\left[  \widetilde{Z}_{\ell}^{-\ell
},\widetilde{Z}_{\ell}^{-\ell+1},\ldots,\widetilde{Z}_{\ell}^{\ell}\right]
^{\top}$, .$\ell=0,1,2,\ldots$, according to the rows of $\left\{
\widetilde{Z}_{\ell}^{m}\right\}  $. The finite dimensional distribution is
characterized by its cumulant function (logarithm of the characteristic
function)
\[
\Phi_{\widetilde{Z}}\left(  \left.  \omega_{\ell}\right\vert \ell
=0,1,2,\ldots\right)  =\ln\mathsf{E\exp}\left(  i\sum_{\ell}\omega_{\ell
}^{\top}\widetilde{Z}_{\ell}\right)  ,
\]
such that only finite many coordinates of the variables $\left(  \left.
\omega_{\ell}\right\vert \ell=0,1,2,\ldots\right)  $ are different from zero.
Consider the transformed series $\widehat{Z}_{\ell}=D^{\left(  \ell\right)
}\widetilde{Z}_{\ell}$, where $D^{\left(  \ell\right)  }$ is the Wigner matrix
of rotations $D^{\left(  \ell\right)  }=\left[  D_{k,m}^{\left(  \ell\right)
}\right]  _{k,m=-\ell}^{\ell}$, and define a triangular array $\left\{
Z_{\ell}^{m}\right\}  $ through the cumulant function
\[
\Phi_{Z}\left(  \left.  \omega_{\ell}\right\vert \ell=0,1,2,\ldots\right)
=\int_{SO\left(  3\right)  }\ln\mathsf{E\exp}\left(  i\sum_{\ell}\omega_{\ell
}^{\top}D^{\left(  \ell\right)  }\widetilde{Z}_{\ell}\right)  dg
\]
where again only finite many coordinates of the variables $\left(  \left.
\omega_{\ell}\right\vert \ell=0,1,2,\ldots\right)  $ are different from zero
and $dg=\sin\vartheta d\vartheta d\varphi d\gamma/8\pi^{2}$, is the Haar
measure with unite mass. This new triangular array $\left\{  Z_{\ell}%
^{m}\right\}  $ will be called Wigner D-transform of $\left\{  \widetilde{Z}%
_{\ell}^{m}\right\}  $. The third order cumulants of $Z_{\ell}^{m}$ for
instance
\begin{align*}
\operatorname*{Cum}\nolimits_{3}\left(  Z_{\ell_{1}}^{k_{1}},Z_{\ell_{2}%
}^{k_{2}},Z_{\ell_{3}}^{k_{3}}\right)   &  =\int_{SO\left(  3\right)  }%
\sum_{m_{1},m_{2},m_{3}}D_{k_{1},m_{1}}^{\left(  \ell_{1}\right)  }%
D_{k_{2},m_{2}}^{\left(  \ell_{2}\right)  }D_{k_{3},m_{3}}^{\left(  \ell
_{3}\right)  }dg\operatorname*{Cum}\nolimits_{3}\left(  \widetilde{Z}%
_{\ell_{1}}^{m_{1}},\widetilde{Z}_{\ell_{2}}^{m_{2}},\widetilde{Z}_{\ell_{3}%
}^{m_{3}}\right) \\
&  =%
\begin{pmatrix}
\ell_{1} & \ell_{2} & \ell_{3}\\
k_{1} & k_{2} & k_{3}%
\end{pmatrix}
\sum_{m_{1},m_{2},m_{3}}%
\begin{pmatrix}
\ell_{1} & \ell_{2} & \ell_{3}\\
m_{1} & m_{2} & m_{3}%
\end{pmatrix}
\operatorname*{Cum}\nolimits_{3}\left(  \widetilde{Z}_{\ell_{1}}^{m_{1}%
},\widetilde{Z}_{\ell_{2}}^{m_{2}},\widetilde{Z}_{\ell_{3}}^{m_{3}}\right) \\
&  =%
\begin{pmatrix}
\ell_{1} & \ell_{2} & \ell_{3}\\
k_{1} & k_{2} & k_{3}%
\end{pmatrix}
B_{3}\left(  \ell_{1},\ell_{2},\ell_{3}\right)  ,
\end{align*}
which fulfils (\ref{BicovZ}). The function $B_{3}\left(  \ell_{1},\ell
_{2},\ell_{3}\right)  $ is given by
\[
B_{3}\left(  \ell_{1},\ell_{2},\ell_{3}\right)  =\sum_{m_{1},m_{2},m_{3}}%
\begin{pmatrix}
\ell_{1} & \ell_{2} & \ell_{3}\\
m_{1} & m_{2} & m_{3}%
\end{pmatrix}
\operatorname*{Cum}\nolimits_{3}\left(  \widetilde{Z}_{\ell_{1}}^{m_{1}%
},\widetilde{Z}_{\ell_{2}}^{m_{2}},\widetilde{Z}_{\ell_{3}}^{m_{3}}\right)  ,
\]
and also
\[
B_{3}\left(  \ell_{1},\ell_{2},\ell_{3}\right)  =\sum_{m_{1},m_{2},m_{3}}%
\begin{pmatrix}
\ell_{1} & \ell_{2} & \ell_{3}\\
m_{1} & m_{2} & m_{3}%
\end{pmatrix}
\operatorname*{Cum}\nolimits_{3}\left(  Z_{\ell_{1}}^{m_{1}},Z_{\ell_{2}%
}^{m_{2}},Z_{\ell_{3}}^{m_{3}}\right)  .
\]
The conclusion is that a subset of cumulants $\operatorname*{Cum}%
\nolimits_{3}\left(  \widetilde{Z}_{\ell_{1}}^{m_{1}},\widetilde{Z}_{\ell_{2}%
}^{m_{2}},\widetilde{Z}_{\ell_{3}}^{m_{3}}\right)  $, i.e. $m_{1}+m_{2}%
+m_{3}=0$, is used in the construction and the superposition with
'probability' amplitudes is applied. In general we also have%
\begin{align*}
\operatorname*{Cum}\nolimits_{p}\left(  Z_{\ell_{1}}^{m_{1}},Z_{\ell_{2}%
}^{m_{2}},\ldots,Z_{\ell_{p}}^{m_{p}}\right)   &  =\sum_{\substack{\ell
^{1},\ldots,\ell^{p-3}\\m^{1},\ldots,m^{p-3}}}\left(  -1\right)
^{\Sigma_{a=1}^{p-3}m^{a}}\prod\limits_{a=0}^{p-3}%
\begin{pmatrix}
\ell^{a} & \ell_{a+2} & \ell^{a+1}\\
-m^{a} & m_{a+2} & m^{a+1}%
\end{pmatrix}
\\
&  \times\prod\limits_{a=1}^{p-3}\sqrt{2\ell^{a}+1}S_{p}\left(  \left.
\ell_{1:p}\right\vert \ell^{1:p-3}\right)
\end{align*}
where
\begin{align*}
S_{p}\left(  \left.  \ell_{1:p}\right\vert \ell^{1:p-3}\right)   &
=\sum_{\substack{m_{1},\ldots,m_{p}\\m^{1},\ldots,m^{p-3}}}\left(  -1\right)
^{\Sigma_{a=1}^{p-3}m^{a}}\prod\limits_{a=0}^{p-3}%
\begin{pmatrix}
\ell^{a} & \ell_{a+2} & \ell^{a+1}\\
-m^{a} & m_{a+2} & m^{a+1}%
\end{pmatrix}
\\
&  \times\prod\limits_{a=1}^{p-3}\sqrt{2\ell^{a}+1}\operatorname*{Cum}%
\nolimits_{p}\left(  \widetilde{Z}_{\ell_{1}}^{k_{1}},\widetilde{Z}_{\ell_{2}%
}^{k_{2}},\ldots,\widetilde{Z}_{\ell_{p}}^{k_{p}}\right)  .
\end{align*}

\appendix

\section{Cumulants\label{Sect_Cummulants}}

\subsection{Basic properties%
\index{Cumulant!properties}%
}

Cumulants (also called invariants) are very important quantities for the
characterization of a random series, see \cite{TerdikRaosSankhya06} for some
details. Moments and cumulants are equvivalent, since the moments can be
defined as the coefficients in the series expansion of the characteristic
function similarly cumulants are the coefficients in the series expansion of
the log-characteristic function. The main importance of the cumulants is that
it does not contain the Gaussian part of the moments. Let $\underline{{X}}\in
R^{n}$ denote a random vector.

\subsubsection{ Symmetry}%

\[
\operatorname*{Cum}(\underline{{X}})=\operatorname*{Cum}(\underline{X}%
_{\mathfrak{p}(1:n\mathfrak{)}}),\;\underline{{X}}\in\mathbb{R}^{n},
\]
where $\mathfrak{p}(1:n)=\left(  \mathfrak{p}\left(  1\right)  ,\mathfrak{p}%
\left(  2\right)  ,\ldots,\mathfrak{p}\left(  n\right)  \right)  ,$
$p\in\mathfrak{{P}_{n}}$ and $\mathfrak{{P}_{n}}$ denotes the set of all
permutations of the integers $(1:n)$.

\subsubsection{Multilinearity}

For any constants $c_{1:2}=(c_{1},c_{2})$
\[
\operatorname*{Cum}(c_{1}Y_{1}+c_{2}Y_{2},\underline{{X}})=c_{1}%
\operatorname*{Cum}(Y_{1},\underline{{X}})+c_{2}\operatorname*{Cum}%
(Y_{2},\underline{{X}}).
\]

\subsubsection{Independence}

If $\underline{{X}}\in R^{n}$ is independent of $\underline{Y}\in R^{m}$ where
$n,m>0$ then
\[
\operatorname*{Cum}(\underline{{X}},\underline{Y})=0.
\]
In particular if $n=m$ then for independent $\underline{{X}}$ and
$\underline{Y}$
\[
\operatorname*{Cum}(\underline{{X}}+\underline{Y})=\operatorname*{Cum}%
(\underline{{X}})+\operatorname*{Cum}(\underline{Y}).
\]
This formula is the additive version of the formula of the moment of the
product of independent variables.

\subsubsection{Gaussianity}

The random vector $\underline{{X}}\in R^{n}$ is Gaussian if and only if for
all vector $k_{(1:m)}$ with elements from $(1:n)$
\[
\operatorname*{Cum}(\underline{X}_{k_{(1:m)}})=0,\ m>2.
\]

\subsubsection{Expression of the cumulant via moments%
\index{Cumulant!by moments}%
}

The first order cumulants is the moment
\[
\mathsf{\operatorname*{Cum}}(X)=\mathsf{E}X,
\]
the second and third order cumulants are the second and third order central
moments%
\begin{align}
\mathsf{\operatorname*{Cum}}(X_{1},X_{2})  &  =\operatorname*{Cov}(X_{1}%
,X_{2})\label{cov_cum}\\
\mathsf{\operatorname*{Cum}}(X_{1},X_{2},X_{3})  &  =\mathsf{E}\prod_{i=1}%
^{3}\left(  X_{i}-\mathsf{E}X_{i}\right)  .
\end{align}
The general formula is
\begin{equation}
\operatorname*{Cum}(\underline{{X}})=\sum_{m=1}^{n}(-1)^{m-1}(m-1)!\sum
_{U\in\mathcal{P}_{n}}\prod_{j=1}^{m}\mathsf{E\ }{X}_{1:n}^{u_{j,1:n}%
},\;\underline{{X}}\in\mathbb{R}^{n}, \label{cummo}%
\end{equation}
where the second summation is taken over all possible partitions $U$, where
$\left[  u_{j,1:n}\right]  _{j=1:m}$ is an indicator of a partition
$\mathcal{L} $ with $m$ subsets of the set $\left(  1:n\right)  =\left(
1,2,\ldots,n\right)  $, $u_{j,k}=1$ if $k\in\mathcal{L}$, otherwise $0$. The
double sum in (\ref{cummo}) is over all partitions of $\mathcal{L\in P}_{n},$
where $\mathcal{P}_{n}$ is the set of all partitions of the numbers $\left(
1:n\right)  $.

\begin{example}
\label{ex_partC} First note that if $n=3,$ then the $m^{th}$ order partitions
$\mathcal{L}$ $\in\mathcal{P}_{3}$ with the corresponding indicators and
permutations of $\left(  1:3\right)  $ are

\begin{enumerate}
\item for $m=1;$ $\mathcal{L}=\left\{  \left(  1:3\right)  \right\}  $%
\[
\mathcal{L}\leftrightarrow u_{1,1:3}=\left[  1,1,1\right]  ,\quad
\mathsf{E\ }{X}_{1:3}^{u_{1,1:3}}=\mathsf{E}X_{1}X_{2}X_{3},
\]

\item for $m=2;$ $\mathcal{L}_{1}=\left\{  (1),\ (2,3)\right\}  ,$
$\mathcal{L}_{2}=\left\{  (1,3),\ (2)\right\}  ,$ $\mathcal{L}_{3}=\left\{
(1,2),\ (3)\right\}  ,$%
\begin{align*}
\mathcal{L}_{1}  &  \leftrightarrow%
\begin{bmatrix}
u_{1,1:3}\\
u_{2,1:3}%
\end{bmatrix}
=\left[
\begin{array}
[c]{ccc}%
1 & 0 & 0\\
0 & 1 & 1
\end{array}
\right]  ,\quad\prod_{j=1}^{2}\mathsf{E\ }{X}_{1:3}^{u_{j,1:3}}=\mathsf{E}%
X_{1}\mathsf{E}X_{2}X_{3},\\
\mathcal{L}_{2}  &  \leftrightarrow%
\begin{bmatrix}
u_{1,1:3}\\
u_{2,1:3}%
\end{bmatrix}
=\left[
\begin{array}
[c]{ccc}%
1 & 0 & 1\\
0 & 1 & 0
\end{array}
\right]  ,\quad\prod_{j=1}^{2}\mathsf{E\ }{X}_{1:3}^{u_{j,1:3}}=\mathsf{E}%
X_{2}\mathsf{E}X_{1}X_{3},\\
\mathcal{L}_{3}  &  \leftrightarrow%
\begin{bmatrix}
u_{1,1:3}\\
u_{2,1:3}%
\end{bmatrix}
=\left[
\begin{array}
[c]{ccc}%
1 & 1 & 0\\
0 & 0 & 1
\end{array}
\right]  ,\quad\prod_{j=1}^{2}\mathsf{E\ }{X}_{1:3}^{u_{j,1:3}}=\mathsf{E}%
X_{3}\mathsf{E}X_{1}X_{2},
\end{align*}

\item for $m=3;\mathcal{L}=\left\{  (1),\ \left(  2\right)  ,(3)\right\}  $
\[
\mathcal{L}\leftrightarrow%
\begin{bmatrix}
u_{1,1:3}\\
u_{2,1:3}\\
u_{2,1:3}%
\end{bmatrix}
=\left[
\begin{array}
[c]{ccc}%
1 & 0 & 0\\
0 & 1 & 0\\
0 & 0 & 1
\end{array}
\right]  ,\quad\prod_{j=1}^{3}\mathsf{E\ }{X}_{1:3}^{u_{j,1:3}}=\mathsf{E}%
X_{1}\mathsf{E}X_{2}\mathsf{E}X_{3},
\]
Now we apply (\ref{cummo}):
\begin{align*}
\operatorname*{Cum}(X_{1},X_{2},X_{3})  &  =\mathsf{E}X_{1}X_{2}%
X_{3}-\mathsf{E}X_{1}\mathsf{E}X_{2}X_{3}-\mathsf{E}X_{2}\mathsf{E}X_{1}%
X_{3}\\
&  -\mathsf{E}X_{3}\mathsf{E}X_{1}X_{2}+2\mathsf{E}X_{1}\mathsf{E}%
X_{2}\mathsf{E}X_{3}\\
&  =\mathsf{E}\prod_{i=1}^{3}\left(  X_{i}-\mathsf{E}X_{i}\right)  ,
\end{align*}
The first three cumulants, see also (\ref{cov_cum}), equal the central moments
but this is not true for higher order cumulants. One might easily check this
for the case of cumulants of order four. Indeed, $\underline{X}\in
\mathbb{R}^{4}$
\begin{align}
\operatorname*{Cum}\left(  \underline{X}\right)   &  =\operatorname*{Cum}%
\left(  \underline{X}-\mathsf{E}\underline{X}\right) \nonumber\\
&  =\mathsf{E}\prod_{i=1}^{4}\left(  X_{i}-\mathsf{E}X_{i}\right)
-\mathsf{\operatorname*{Cov}}\left(  X_{1},X_{2}\right)
\mathsf{\operatorname*{Cov}}\left(  X_{3},X_{4}\right)  \mathsf{\ }\nonumber\\
&  \hspace{-0.4in}-\mathsf{\operatorname*{Cov}}\left(  X_{1},X_{3}\right)
\mathsf{\operatorname*{Cov}}\left(  X_{2},X_{4}\right)
-\mathsf{\operatorname*{Cov}}\left(  X_{1},X_{4}\right)
\mathsf{\operatorname*{Cov}}\left(  X_{2},X_{3}\right)  \mathsf{\ }\nonumber\\
&  \neq\mathsf{E}\prod_{i=1}^{4}\left(  X_{i}-\mathsf{E\ }X_{i}\right)  ,
\label{cum-4}%
\end{align}
unless all covariances are zero.
\end{enumerate}
\end{example}

\subsubsection{\textbf{Expression of the moment via cumulants}%
\index{Moment!by cumulants}%
}

We consider the moment $\mathsf{E}\prod_{i=1}^{n}X_{i}$ as the general case
because the moment $\mathsf{E}Y_{1:m}^{k_{1:m}}$ can be put into the form
$\mathsf{E}\prod_{i=1}^{n}X_{i}$, where we define $n=\Sigma k_{(1:m)}$ and
\[
\underline{{X}}=\underset{k_{1}}{(\underbrace{Y_{1},\ldots,Y_{1}}}%
,\ldots,\underset{k_{m}}{\underbrace{Y_{m},\ldots,Y_{m}}}).
\]
Now
\begin{equation}
\mathsf{E}\prod_{i=1}^{n}X_{i}=\sum_{\mathcal{L\in P}_{n}}\prod_{\mathbf{b}%
_{j}\in\mathcal{L}}\operatorname*{Cum}(X_{\mathbf{b}_{j}}), \label{mocum}%
\end{equation}
where the summation is over all partitions $\mathcal{L=}\left\{
\mathbf{b}_{1},\mathbf{b}_{2},\ldots,\mathbf{b}_{k}\right\}  $ of $(1:n)$.

An example of the use of the formula (\ref{mocum}) can be seen when $n=3$.

\begin{example}
Each partition of $\left(  1:3\right)  $ are listed in the previus example
therefore if $\underline{X}\in\mathbb{R}^{3},$
\begin{align*}
\mathsf{E}X_{1}X_{2}X_{3}  &  =\operatorname*{Cum}(X_{1:3}%
)+\operatorname*{Cum}(X_{1})\operatorname*{Cum}(X_{2},X_{3})\\
&  \ +\operatorname*{Cum}(X_{2})\operatorname*{Cum}(X_{1},X_{3}%
)+\operatorname*{Cum}(X_{3})\operatorname*{Cum}(X_{1},X_{2})\\
&  \ +\operatorname*{Cum}(X_{1})\operatorname*{Cum}(X_{2})\operatorname*{Cum}%
(X_{3}).
\end{align*}
Now in particular
\[
\mathsf{E\ }X^{3}=\operatorname*{Cum}(X,X,X)+3\operatorname*{Cum}%
(X)\operatorname*{Cum}(X,X)+\operatorname*{Cum}(X)^{3}.
\]

\end{example}

\bigskip

\section{Spherical Harmonics\label{App_Spher_Harmonics}}

The following notation and results are used in this mauscript.

\begin{enumerate}
\item $\mathbb{R}^{3}$ is Euclidean space of dimension $3$.

\item $\mathbb{S}_{2}$ denotes the \textbf{sphere }with radius $1$ in
$\mathbb{R}^{3}$. The following notations are used; colatitude (angular)
coordinate: $\vartheta\in\left[  0,\pi\right]  $, longitude coordinate
$\varphi\in\left[  0,2\pi\right]  $,$\;$North pole: $N$ with $\vartheta=0$,
and $\varphi$ is indeterminate, (latitude coordinate is expressed by
colatitude coordinate: $\pi/2-\vartheta\in\left[  -\pi/2,\pi/2\right]  $, in
that case the North pole has $\vartheta=\pi/2$). The spherical coordinates on
$\mathbb{S}_{2}$: $\left(  \sin\vartheta\cos\varphi,\sin\vartheta\sin
\varphi,\cos\vartheta\right)  $

\item $SO\left(  3\right)  $ denotes the 3D (special orthogonal) rotation group

\item \textbf{Addition Theorem}\textit{. }\cite{Erdelyi11a} vol. 2,7.15 (30),
pp.116 Let $U_{1}$ and $U_{2}$ be two vectors with angle $\vartheta$ and
lengths $r_{1}$ and $r_{2}$ respectively, and the Euclidean distance is
denoted by $\rho=\left\Vert U_{1}-U_{2}\right\Vert =\sqrt{r_{1}^{2}+r_{2}%
^{2}-2r_{1}r_{2}\cos\gamma}$, then we have \textit{\ }
\begin{equation}
\rho^{-\nu}J_{\nu}\left(  \rho\right)  =\left(  \frac{r_{1}r_{2}}{2}\right)
^{-\nu}\Gamma\left(  \nu\right)  \sum_{\ell=0}^{\infty}\left(  \ell
+\nu\right)  C_{\ell}^{\nu}\left(  \cos\gamma\right)  J_{\ell+\nu}\left(
r_{1}\right)  J_{\ell+\nu}\left(  r_{2}\right)  , \label{AdditionTheorem}%
\end{equation}
where $J_{\nu}$ denotes the Bessel function of the first kind, $C_{\ell}^{\nu
}$ Gegenbauer polynomial $\left(  C_{\ell}^{1/2}=P_{\ell}\right)  $.

\item \textbf{Lapalce-Beltrami operator} $\bigtriangleup_{B}$,
\begin{align*}
\bigtriangleup_{B}  &  =\frac{1}{\sin\vartheta}\frac{\partial}{\partial
\vartheta}\left(  \sin\vartheta\frac{\partial}{\partial\vartheta}\right)
+\frac{1}{\sin^{2}\vartheta}\frac{\partial^{2}}{\partial\varphi^{2}},\\
\bigtriangleup_{B}  &  =\frac{1}{\sin\vartheta}\left[  \frac{\partial
}{\partial\vartheta}\left(  \sin\vartheta\frac{\partial}{\partial\vartheta
}\right)  +\frac{1}{\sin\vartheta}\frac{\partial^{2}}{\partial\varphi^{2}%
}\right]  ,\\
\bigtriangleup_{B}  &  =\frac{\partial^{2}}{\partial\vartheta^{2}}+\frac
{\cos\vartheta}{\sin\vartheta}\frac{\partial}{\partial\vartheta}+\frac{1}%
{\sin^{2}\vartheta}\frac{\partial^{2}}{\partial\varphi^{2}},
\end{align*}
(\cite{Apostol2007}, vol. 2 pp. 293).

\item $\bigtriangledown$ is the central difference operator $\left(
\bigtriangledown_{1}^{2}+\bigtriangledown_{2}^{2}\right)  X_{j,k}%
=X_{j+1,k}+X_{j-1,k}+X_{j,k+1}+X_{j,k-1}-4X_{j,k}$

\item \textbf{Spherical function} $f$ of order $\ell$.%
\begin{align*}
\bigtriangleup_{B}f  &  =\frac{1}{\sin\vartheta}\frac{\partial}{\partial
\vartheta}\left(  \sin\vartheta\frac{\partial f}{\partial\vartheta}\right)
+\frac{1}{\sin^{2}\vartheta}\frac{\partial^{2}f}{\partial\varphi^{2}}\\
&  =-\ell\left(  \ell+1\right)  f,
\end{align*}
\cite{Varshalovich1988}. \textit{Regular }\textbf{spherical harmonics
}$Y_{\ell}$, homogeneous polynomial of degree $\ell$ and
\[
\nabla^{2}Y_{\ell}=0.
\]

\item \textit{Normalized} \textbf{Legendre polynomial} $\widetilde{P}_{\ell} $
(\textit{Rodrigues' formula})%
\begin{align*}
\widetilde{P}_{\ell}\left(  x\right)   &  =\sqrt{\frac{2\ell+1}{2}}\frac
{1}{2^{\ell}\ell!}\frac{d^{\ell}\left(  x^{2}-1\right)  ^{\ell}}{dx^{\ell}%
},\;x\in\left[  -1,1\right]  ,\\
&  \int_{-1}^{1}\left(  \widetilde{P}_{\ell}\left(  x\right)  \right)
^{2}dx=1
\end{align*}
\textit{Standardized }\textbf{Legendre polynomial }$P_{0}\left(  x\right)
=1,$
\[
P_{\ell}\left(  x\right)  =\frac{1}{2^{\ell}\ell!}\frac{d^{\ell}\left(
x^{2}-1\right)  ^{\ell}}{dx^{\ell}},\;x\in\left[  -1,1\right]  ,
\]
$P_{\ell}\left(  1\right)  =1$(\cite{Erdelyi11a}, vol. 2, pp.180) it is
orthogonal and
\begin{align}
\int_{-1}^{1}\left[  P_{\ell}\left(  x\right)  \right]  ^{2}dx  &  =\frac
{2}{2\ell+1}\nonumber\\
\int_{\mathbb{S}_{2}}\left[  P_{\ell}\left(  \cos\vartheta\right)  \right]
^{2}\Omega\left(  dL\right)   &  =\int_{0}^{2\pi}\int_{0}^{\pi}\left[
P_{\ell}\left(  \cos\vartheta\right)  \right]  ^{2}\sin\vartheta d\vartheta
d\varphi\nonumber\\
&  =2\pi\int_{0}^{\pi}\left[  P_{\ell}\left(  \cos\vartheta\right)  \right]
^{2}\sin\vartheta d\vartheta\nonumber\\
&  =2\pi\int_{-1}^{1}\left[  P_{\ell}\left(  x\right)  \right]  ^{2}%
dx\nonumber\\
&  =\frac{4\pi}{2\ell+1} \label{LegengreNormSqu}%
\end{align}
\textbf{recurrence}%
\[
\left(  \ell+1\right)  P_{\ell+1}\left(  x\right)  =\left(  2\ell+1\right)
xP_{\ell}\left(  x\right)  -\ell P_{\ell-1}\left(  x\right)
\]
\textbf{generating functions }(\cite{Erdelyi11a}, vol. 2, pp.183)%
\begin{align}
\sum_{\ell=0}^{\infty}P_{\ell}\left(  x\right)  z^{\ell}  &  =\left(
1-2xz+z^{2}\right)  ^{-1/2},\;x\in\left(  -1,1\right)  ,\;\left\vert
z\right\vert <1\label{Funct_Gen_Pl1}\\
\sum_{\ell=0}^{\infty}\frac{1}{\ell!}P_{\ell}\left(  \cos\gamma\right)
z^{\ell}  &  =e^{z\cos\gamma}J_{0}\left(  z\sin\vartheta\right)
\label{Funct_Gen_Pl2}\\
\sum_{\ell=0}^{\infty}\frac{\left(  -1\right)  ^{\ell}}{\ell+1/2}P_{\ell
}\left(  \cos\gamma\right)  z^{2\ell+1}  &  =F\left(  \sin\vartheta
/2,\varphi\right)  ,\quad z=\tan\frac{\varphi}{2},\quad0<\varphi<\frac{\pi}%
{2},0<\vartheta<\pi.\nonumber
\end{align}

\item $P_{\ell}^{m}$ is the \textit{associated normalized }\textbf{Legendre
function} \textit{of the first kind} (Gegenbauer polynomial at particular
indices) of \textbf{degree} $\ell$ and \textbf{order} $m$
\begin{align*}
P_{\ell}^{m}\left(  x\right)   &  =\left(  -1\right)  ^{m}\left(
1-x^{2}\right)  ^{m/2}\frac{d^{m}P_{\ell}\left(  x\right)  }{dx^{m}},\\
P_{\ell}^{m}\left(  z\right)   &  =\left(  z^{2}-1\right)  ^{m/2}\frac
{d^{m}P_{\ell}\left(  z\right)  }{dz^{m}},
\end{align*}
recurrence formula is%
\[
\left(  \ell-m+1\right)  P_{\ell+1}^{m}\left(  x\right)  =\left(
2\ell+1\right)  xP_{\ell}^{m}\left(  x\right)  -\left(  \ell+m\right)
P_{\ell-1}^{m}\left(  x\right)  ,
\]
in particular $P_{\ell}^{m}\left(  1\right)  =\delta_{m,0}$, $P_{\ell}%
^{0}\left(  x\right)  =P_{\ell}\left(  x\right)  $,
\[
P_{\ell}^{-m}\left(  x\right)  =\left(  -1\right)  ^{m}\frac{\Gamma\left(
\ell-m+1\right)  }{\Gamma\left(  \ell+m+1\right)  }P_{\ell}^{m}\left(
x\right)  .
\]

\item \textbf{Funk-Hecke formula} (\cite{Mueller1966}) Suppose $G$ is
continuous on $\left[  -1,1\right]  $, then for any spherical harmonic
$Y_{\ell}\left(  L\right)  $%
\begin{align*}
\int_{\mathbb{S}_{2}}G\left(  L_{1}\cdot L\right)  Y_{\ell}\left(  L\right)
\Omega\left(  dL\right)   &  =cY_{\ell}\left(  L_{1}\right)  ,\\
c  &  =2\pi\int_{-1}^{1}G\left(  x\right)  P_{\ell}\left(  x\right)  dx,
\end{align*}
where $\Omega\left(  dL\right)  =\sin\vartheta d\vartheta d\varphi$ is
Lebesque element of surface area on $\mathbb{S}_{2}$, $L_{1}\cdot L_{2}%
=\cos\vartheta$. In particular
\begin{align}
\int_{\mathbb{S}_{2}}G\left(  L_{1}\cdot L\right)  Y_{\ell}^{m}\left(
L\right)  \Omega\left(  dL\right)   &  =g_{\ell}Y_{\ell}^{m}\left(
L_{1}\right)  ,\label{Formula_Funk-Hecke}\\
g_{\ell}  &  =2\pi\int_{-1}^{1}G\left(  x\right)  P_{\ell}\left(  x\right)
dx\nonumber
\end{align}
see \cite{Yadrenko1983} pp.72. For a general dimension $d$ instead of
$P_{\ell}$ one has the Gegenbauer polynomial $C_{\ell}^{\left(  d-2\right)
/2}$. \newline Also%
\begin{align*}
\int_{\mathbb{S}_{2}}G\left(  L_{1}\cdot L\right)  P_{\ell}\left(  L_{2}\cdot
L\right)  \Omega\left(  dL\right)   &  =f_{\ell}P_{\ell}\left(  L_{1}\cdot
L_{2}\right) \\
f_{\ell}  &  =2\pi\int_{-1}^{1}G\left(  x\right)  P_{\ell}\left(  x\right)  dx
\end{align*}
here $f_{\ell}=\left(  LT\right)  _{\ell}G$ is the Legendre transform of $G$.
Hence
\[
G\left(  \cos\vartheta\right)  =\sum_{\ell=0}^{\infty}\frac{2\ell+1}{2}%
\int_{-1}^{1}G\left(  x\right)  P_{\ell}\left(  x\right)  dxP_{\ell}\left(
\cos\vartheta\right)
\]
In particular \textit{Funk-Hecke formula when }$N$ denotes the north pole
\begin{align*}
\int_{\mathbb{S}_{2}}\left[  P_{\ell}\left(  N\cdot L\right)  \right]
^{2}\Omega\left(  dL\right)   &  =cP_{\ell}\left(  N\cdot N\right) \\
c  &  =2\pi\int_{-1}^{1}\left[  P_{\ell}\left(  x\right)  \right]
^{2}dx,\quad P_{\ell}\left(  1\right)  =1\\
\int_{\mathbb{S}_{2}}\left[  P_{\ell}\left(  N\cdot L\right)  \right]
^{2}\Omega\left(  dL\right)   &  =\frac{4\pi}{2\ell+1}%
\end{align*}

\item \textbf{Orthonormal} \textbf{spherical harmonics} \textbf{with complex
values }$Y_{\ell}^{m}\left(  \vartheta,\varphi\right)  $, $\ell=0,1,2,\ldots$,
$m=-\ell,-\ell+1,\ldots-1,0,1,\ldots,\ell-1,\ell$ of \textbf{degree} $\ell$
and \textbf{order} $m$ (rank $\ell$ and projection $m$)
\begin{equation}
Y_{\ell}^{m}\left(  \vartheta,\varphi\right)  =\left(  -1\right)  ^{m}%
\sqrt{\frac{2\ell+1}{4\pi}\frac{\left(  \ell-m\right)  !}{\left(
\ell+m\right)  !}}P_{\ell}^{m}\left(  \cos\vartheta\right)  e^{im\varphi
},\;\varphi\in\left[  0,2\pi\right]  ,\;\vartheta\in\left[  0,\pi\right]
\label{SpheriHarmOrthoNorm}%
\end{equation}%
\begin{align*}
Y_{\ell}^{m}\left(  \vartheta,\varphi\right)  ^{\ast}  &  =\left(  -1\right)
^{m}Y_{\ell}^{-m}\left(  \vartheta,\varphi\right) \\
Y_{\ell}^{m}\left(  -\vartheta,-\varphi\right)   &  =\left(  -1\right)
^{\ell}Y_{\ell}^{m}\left(  \vartheta,\varphi\right)  .
\end{align*}
In particular $Y_{\ell}^{0}\left(  \vartheta,\varphi\right)  =\sqrt
{\frac{2\ell+1}{4\pi}}P_{\ell}\left(  \cos\vartheta\right)  $, $Y_{0}%
^{0}\left(  \vartheta,\varphi\right)  =\sqrt{\frac{1}{4\pi}}$,
\[
Y_{\ell}^{m}\left(  N\right)  =\delta_{m,0}\sqrt{\frac{2\ell+1}{4\pi}},
\]
we have%
\[
\int_{0}^{2\pi}\int_{0}^{\pi}\left\vert Y_{\ell}^{m}\left(  \vartheta
,\varphi\right)  \right\vert ^{2}\sin\vartheta d\vartheta d\varphi=1
\]
note that $Y_{\ell}^{m}$ is normalized \textit{fully}, some authors do not
apply $1/\sqrt{4\pi}$ in the definition of $Y_{\ell}^{m}$, also for a sphere
with radius $R$ \textit{spherical harmonics are normalized additionally}%
\textbf{\ }$Y_{\ell}^{m}\left(  \vartheta,\varphi\right)  /R$. It also follows%
\begin{align*}
Y_{\ell}^{m\ast}\left(  \vartheta,\varphi\right)   &  =Y_{\ell}^{m}\left(
\vartheta,-\varphi\right) \\
&  =\left(  -1\right)  ^{m}Y_{\ell}^{-m}\left(  \vartheta,\varphi\right)  ,\\
Y_{\ell}^{-m}\left(  \vartheta,\varphi\right)   &  =\left(  -1\right)
^{m}e^{-i2m\varphi}Y_{\ell}^{m}\left(  \vartheta,\varphi\right)  .
\end{align*}

\item \label{Wigner_3j}\textbf{Wigner }$3j$\textbf{-symbols}, \cite{Louck2006}
notation%
\[%
\begin{pmatrix}
\ell_{1:3}\\
m_{1:3}%
\end{pmatrix}
=%
\begin{pmatrix}
\ell_{1} & \ell_{2} & \ell_{3}\\
m_{1} & m_{2} & m_{3}%
\end{pmatrix}
,
\]
\textit{Clebsch-Gordan coefficients,}%
\begin{align*}
C_{\ell_{1},k_{1};\ell_{2},k_{2}}^{\ell,k}  &  =\left(  -1\right)  ^{\ell
_{1}-\ell_{2}+k}\sqrt{2\ell+1}%
\begin{pmatrix}
\ell_{1:2} & \ell\\
k_{1:2} & -k
\end{pmatrix}
,\\%
\begin{pmatrix}
\ell_{1:3}\\
m_{1:3}%
\end{pmatrix}
&  =\frac{\left(  -1\right)  ^{\ell_{3}+m_{3}+2\ell_{1}}}{\sqrt{2\ell_{3}+1}%
}C_{\ell_{1},-k_{1};\ell_{2},-k_{2}}^{\ell_{3},k_{3}}.
\end{align*}
Symmetries%
\begin{align}%
\begin{pmatrix}
\ell_{1:3}\\
m_{1:3}%
\end{pmatrix}
&  =\left(  -1\right)  ^{\ell_{1}+\ell_{2}+\ell_{3}}%
\begin{pmatrix}
\ell_{2} & \ell_{1} & \ell_{3}\\
m_{2} & m_{1} & m_{3}%
\end{pmatrix}
\label{W3_perm}\\
&  =%
\begin{pmatrix}
\ell_{2} & \ell_{3} & \ell_{1}\\
m_{2} & m_{3} & m_{1}%
\end{pmatrix}
,\nonumber\\%
\begin{pmatrix}
\ell_{1:3}\\
m_{1:3}%
\end{pmatrix}
&  =\left(  -1\right)  ^{\ell_{1}+\ell_{2}+\ell_{3}}%
\begin{pmatrix}
\ell_{1} & \ell_{2} & \ell_{3}\\
-m_{1} & -m_{2} & -m_{3}%
\end{pmatrix}
.\nonumber
\end{align}
Orthogonality relations
\begin{align}
\left(  2\ell+1\right)  \sum_{m_{1:2}}%
\begin{pmatrix}
\ell_{1:2} & \ell\\
m_{1:2} & m
\end{pmatrix}%
\begin{pmatrix}
\ell_{1:2} & j\\
m_{1:2} & k
\end{pmatrix}
&  =\delta_{m,k}\delta_{\ell,j}\label{W3j_Orth1}\\
\sum_{\ell,k}\left(  2\ell+1\right)
\begin{pmatrix}
\ell_{1:2} & \ell\\
m_{1:2} & k
\end{pmatrix}%
\begin{pmatrix}
\ell_{1:2} & \ell\\
k_{1:2} & k
\end{pmatrix}
&  =\delta_{m_{1},k_{1}}\delta_{m_{2},k_{2}} \label{W3j_Orth2}%
\end{align}%
\[%
\begin{pmatrix}
\ell_{1} & \ell_{2} & \ell_{3}\\
0 & 0 & 0
\end{pmatrix}
=0,
\]
if $\mathcal{L=}\ell_{1}+\ell_{2}+\ell_{3}$ is odd, otherwise see
\cite{Edmonds1957}, (3.7.17)\newline%
\[%
\begin{pmatrix}
\ell_{1} & \ell_{2} & \ell_{3}\\
0 & 0 & 0
\end{pmatrix}
=\left(  -1\right)  ^{\mathcal{L}/2}\sqrt{\frac{\prod\left(  \mathcal{L}%
-2\ell_{j}\right)  !}{\left(  \mathcal{L}+1\right)  !}}\frac{\left(
\mathcal{L}/2\right)  !}{\prod\left(  \mathcal{L}/2-\ell_{j}\right)  !},
\]
if $\ell$ is even say $\ell=2k$,then%
\[%
\begin{pmatrix}
\ell & \ell & \ell\\
0 & 0 & 0
\end{pmatrix}
=\left(  -1\right)  ^{3k}\sqrt{\frac{\left(  \ell!\right)  ^{3}}{\left(
6k+1\right)  !}}\frac{\left(  3k\right)  !}{\left(  k!\right)  ^{3}}.
\]

\item \textbf{\label{Select.Rules} Selection rules: a }Wigner 3j symbol
vanishes unless

\begin{itemize}
\item $m_{1}+m_{2}+m_{3}=0,$

\item Integer perimeter rule: $\mathcal{L}=\ell_{1}+\ell_{2}+\ell_{3}$ is an
integer (if $m_{1}=m_{2}=m_{3}=0$, then $\mathcal{L}$ is even).

\item Triangular inequality $\left\vert \ell_{1}-\ell_{2}\right\vert \leq
\ell_{3}\leq\ell_{1}+\ell_{2}$ fulfilles.
\end{itemize}

\item \textbf{Gaunt series } \cite{Edmonds1957}, (4.6.5)%
\begin{align*}
Y_{\ell_{1}}^{m_{1}}\left(  \vartheta,\varphi\right)  Y_{\ell_{2}}^{m_{2}%
}\left(  \vartheta,\varphi\right)   &  =\sum_{\ell,m}\sqrt{\frac{\left(
2\ell_{1}+1\right)  \left(  2\ell_{2}+1\right)  \left(  2\ell+1\right)  }%
{4\pi}}\\
&  \times%
\begin{pmatrix}
\ell_{1} & \ell_{2} & \ell\\
0 & 0 & 0
\end{pmatrix}%
\begin{pmatrix}
\ell_{1:2} & \ell\\
m_{1:2} & m
\end{pmatrix}
Y_{\ell}^{m\ast}\left(  \vartheta,\varphi\right)  ,\\
Y_{\ell_{1}}^{m_{1}}\left(  \vartheta,\varphi\right)  Y_{\ell_{2}}^{m_{2}%
}\left(  \vartheta,\varphi\right)   &  =\sum_{\ell,m}\sqrt{\frac{\left(
2\ell_{1}+1\right)  \left(  2\ell_{2}+1\right)  }{4\pi\left(  2\ell+1\right)
}}C_{\ell_{1},m_{1};\ell_{2},m_{2}}^{\ell,m}C_{\ell_{1},0;\ell_{2},0}%
^{\ell_{3},0}Y_{\ell}^{m}\left(  \vartheta,\varphi\right)
\end{align*}
\cite{Varshalovich1988}, pp. 144, products of three and more

\item \textbf{Rayleigh plane wave expansion in 3D:} $\underline{\widehat{k}%
}=\underline{k}/\left\vert \underline{k}\right\vert $, $\underline{\widehat{x}%
}=\underline{x}/\left\vert \underline{x}\right\vert $, $r=\left\vert
\underline{x}\right\vert $, $k=\left\vert \underline{k}\right\vert $,
\[
e^{i\underline{k}\cdot\underline{x}}=e^{ikr\cos\vartheta}=4\pi\sum_{\ell
=0}^{\infty}\sum_{m=-\ell}^{\ell}i^{\ell}j_{\ell}\left(  kr\right)  Y_{\ell
}^{m}\left(  \underline{\widehat{k}}\right)  ^{\ast}Y_{\ell}^{m}\left(
\underline{\widehat{x}}\right)
\]
where
\[
j_{\ell}\left(  kr\right)  =\sqrt{\frac{\pi}{2kr}}J_{\ell+1/2}\left(
kr\right)  ,
\]
is the Spherical Bessel function of the first kind, also%
\begin{align*}
e^{iz\cos\vartheta}  &  =\sum_{\ell=-\infty}^{\infty}i^{\ell}J_{\ell}\left(
z\right)  e^{i\ell\vartheta},\\
e^{iz\cos\vartheta}  &  =J_{0}\left(  z\right)  +\sum_{\ell=-\infty}^{\infty
}i^{\ell}J_{\ell}\left(  z\right)  \cos\ell\vartheta,
\end{align*}
see \cite{Louck2006}.

\item \textbf{2-D Dirac delta on sphere}%
\[
\delta\left(  L_{1},L_{2}\right)  =\sum_{\ell,m}Y_{\ell}^{m}\left(
L_{1}\right)  Y_{\ell}^{m}\left(  L_{2}\right)  =\frac{1}{4\pi}\sum_{\ell
}\left(  2\ell+1\right)  P_{\ell}\left(  L_{1}\cdot L_{2}\right)  ,
\]
\qquad\textbf{1-D Dirac delta}%
\[
\delta\left(  x-y\right)  =\frac{1}{2}\sum_{\ell}\left(  2\ell+1\right)
P_{\ell}\left(  x\right)  P_{\ell}\left(  y\right)  ,
\]
see \cite{Louck2006}.

\item \textbf{Rotational invariant functions. }The function\textbf{\ }
\begin{align}
I_{\ell}\left(  L_{1},L_{2}\right)   &  =\frac{4\pi}{2\ell+1}\sum_{k=-\ell
}^{\ell}\left(  -1\right)  ^{k}Y_{\ell}^{k}\left(  L_{1}\right)  Y_{\ell}%
^{-k}\left(  L_{2}\right) \nonumber\\
&  =\frac{4\pi}{2\ell+1}\sum_{k=-\ell}^{\ell}Y_{\ell}^{k}\left(  L_{1}\right)
Y_{\ell}^{k\ast}\left(  L_{2}\right)  \label{RotationProd2}%
\end{align}
is rotational invariant. This expression is valid when we apply rotation
$g_{L_{1}}$, where $g_{L_{1}}L_{1}=N$
\begin{align*}
I_{\ell}\left(  L_{1},L_{2}\right)   &  =\frac{4\pi}{2\ell+1}\frac{2\ell
+1}{4\pi}\sum_{k=-\ell}^{\ell}\delta_{k,0}Y_{\ell}^{k}\left(  g_{L_{1}}%
L_{2}\right) \\
&  =Y_{\ell}^{0}\left(  g_{L_{1}}L_{2}\right) \\
&  =P_{\ell}\left(  \cos L_{1}\cdot L_{2}\right)  .
\end{align*}
It follows from the addition formula (\ref{FormulaAddition}) as well. The
function
\begin{equation}
I_{\ell_{1}\ell_{2}\ell_{3}}\left(  L_{1},L_{2},L_{3}\right)  =\frac{\left(
4\pi\right)  ^{3/2}}{\sqrt{\prod\left(  2\ell_{j}+1\right)  }}\sum_{m_{1:3}}%
\begin{pmatrix}
\ell_{1:3}\\
m_{1:3}%
\end{pmatrix}
Y_{\ell_{1}}^{m_{1}}\left(  L_{1}\right)  Y_{\ell_{2}}^{m_{2}}\left(
L_{2}\right)  Y_{\ell_{3}}^{m_{3}}\left(  L_{3}\right)  \label{RotationProd3}%
\end{equation}
of three location is rotational invariant, see \cite{Louck2006}, pp.14. We
repeat the notation $%
\begin{pmatrix}
\ell_{1:3}\\
m_{1:3}%
\end{pmatrix}
=%
\begin{pmatrix}
\ell_{1} & \ell_{2} & \ell_{3}\\
m_{1} & m_{2} & m_{3}%
\end{pmatrix}
$. In particular we have
\[
I_{\ell_{1:3}}\left(  L_{1:3}\right)  =I_{\ell_{1:3}}\left(  N,g_{L_{1}}%
L_{2},g_{L_{1}}L_{3}\right)  .
\]

\item \label{Wigner_D_matrix}\textbf{Wigner D-matrix} Let $\Lambda\left(
g\right)  Y_{\ell}^{m}\left(  L\right)  =Y_{\ell}^{m}\left(  g^{-1}L\right)
$,%
\begin{equation}
\Lambda\left(  g\right)  Y_{\ell}^{m}\left(  L\right)  =\sum_{k=-\ell}^{\ell
}D_{k,m}^{\left(  \ell\right)  }\left(  g\right)  Y_{\ell}^{k}\left(
L\right)  \label{MatrixWignerD}%
\end{equation}
if $\ell$ is fixed $D_{m,k}^{\left(  \ell\right)  }\left(  g\right)  $ is
unitary
\begin{align*}
\sum_{k=-\ell}^{\ell}D_{m_{1},k}^{\left(  \ell\right)  }\left(  g\right)
D_{m_{2},k}^{\left(  \ell\right)  \ast}\left(  g\right)   &  =\delta
_{m_{1},m_{2}},\\
\sum_{k=-\ell}^{\ell}D_{k,m_{1}}^{\left(  \ell\right)  }\left(  g\right)
D_{k,m_{2}}^{\left(  \ell\right)  \ast}\left(  g\right)   &  =\delta
_{m_{1},m_{2}},\\
\sum_{k=-\ell}^{\ell}D_{k,m_{1}}^{\left(  \ell\right)  \ast}\left(  g\right)
D_{k,m_{2}}^{\left(  \ell\right)  }\left(  g\right)   &  =\delta_{m_{1},m_{2}%
}.
\end{align*}
We introduce the notation $D^{\left(  \ell\right)  }=\left[  D_{m,k}^{\left(
\ell\right)  }\right]  $, for fixed rotation $g$. Thus $D^{\left(
\ell\right)  }$ denotes a unitary matrix of order $2\ell+1$, and it follows
$D^{\left(  \ell\right)  }\left[  D^{\left(  \ell\right)  }\right]
^{-1}=D^{\left(  \ell\right)  }D^{\left(  \ell\right)  \ast}$, $\det
D^{\left(  \ell\right)  }=1$ (unimodular). $D_{m,k}^{\left(  \ell\right)
}\left(  g\right)  $ is given by the clear formula
\begin{equation}
D_{m,k}^{\left(  \ell\right)  }\left(  \varphi,\vartheta,\gamma\right)
=\exp\left(  -im\varphi\right)  d_{m,k}^{\left(  \ell\right)  }\left(
\vartheta\right)  \exp\left(  -ik\gamma\right)  \label{Wigner_D_expl_form}%
\end{equation}
where $d_{m,k}^{\left(  \ell\right)  }$ is the Wigner (small) d-matrix it is
real and
\begin{align*}
d_{m,k}^{\left(  \ell\right)  }\left(  -\vartheta\right)   &  =\left(
-1\right)  ^{m-k}d_{m,k}^{\left(  \ell\right)  }\left(  \vartheta\right) \\
&  =d_{k,m}^{\left(  \ell\right)  }\left(  \vartheta\right) \\
&  =\left(  -1\right)  ^{k-m}d_{-k,-m}^{\left(  \ell\right)  }\left(
\vartheta\right) \\
&  =d_{-m,-k}^{\left(  \ell\right)  }\left(  \vartheta\right)
\end{align*}
see \cite{Varshalovich1988}, pp79 for details\newline for instance%
\begin{align*}
\left[  D_{m,k}^{\left(  \ell\right)  }\left(  \varphi,\vartheta
,\gamma\right)  \right]  ^{\ast}  &  =\left[  \overline{D_{k,m}^{\left(
\ell\right)  }}\left(  \varphi,\vartheta,\gamma\right)  \right]  =\left[
\left(  -1\right)  ^{m-k}\overline{D_{m,k}^{\left(  \ell\right)  }}\left(
\gamma,\vartheta,\varphi\right)  \right] \\
&  =\left[  \left(  -1\right)  ^{m-k}D_{-m,-k}^{\left(  \ell\right)  }\left(
\gamma,\vartheta,\varphi\right)  \right] \\
&  =\left[  D_{m,k}^{\left(  \ell\right)  }\left(  -\gamma,-\vartheta
,-\varphi\right)  \right]  ,
\end{align*}%
\begin{align*}
D_{k,m}^{\left(  \ell\right)  }\left(  \varphi,\vartheta,\gamma\right)   &
=\left(  -1\right)  ^{m-k}D_{m,k}^{\left(  \ell\right)  }\left(
\gamma,\vartheta,\varphi\right) \\
&  =D_{m,k}^{\left(  \ell\right)  }\left(  \gamma,-\vartheta,\varphi\right)
\end{align*}
note here that if the Euler angles $\left(  \varphi,\vartheta,\gamma\right)  $
corresponds to $g$ then $\left(  -\gamma,-\vartheta,-\varphi\right)  $
corresponds to $g^{-1}$, hence
\begin{align}
D_{m,k}^{\left(  \ell\right)  }\left(  g\right)   &  =\left(  -1\right)
^{m-k}D_{-m,-k}^{\left(  \ell\right)  \ast}\left(  g\right)
\label{D_conjugate}\\
D_{m,k}^{\left(  \ell\right)  \ast}\left(  g\right)   &  =D_{k,m}^{\left(
\ell\right)  }\left(  g^{-1}\right) \nonumber
\end{align}
also%
\[
D_{m,k}^{\left(  \ell\right)  }\left(  g\right)  =D_{\ell-m+1,\ell
-k+1}^{\left(  \ell\right)  }\left(  g\right)
\]
Let $g$ the successive application of $g_{1}$ and $g_{2}$ then
\[
D_{m,k}^{\left(  \ell\right)  }\left(  g\right)  =\sum_{j=-\ell}^{\ell}%
D_{m,j}^{\left(  \ell\right)  }\left(  g_{1}\right)  D_{j,k}^{\left(
\ell\right)  }\left(  g_{2}\right)  .
\]

\item \textbf{Condon and Shortley phase convention}, \cite{Edmonds1957},
(4.3.3)
\begin{align}
Y_{\ell}^{m}\left(  \vartheta,\varphi\right)   &  =\sqrt{\frac{2\ell+1}{4\pi}%
}D_{0,-m}^{\left(  \ell\right)  }\left(  \gamma,\vartheta,\varphi\right)
\label{Convention_Condon_Shortley}\\
&  =\left(  -1\right)  ^{m}D_{0,m}^{\left(  \ell\right)  \ast}\left(
\gamma,\vartheta,\varphi\right) \nonumber\\
&  =\left(  -1\right)  ^{m}\sqrt{\frac{2\ell+1}{4\pi}}D_{-m,0}^{\left(
\ell\right)  }\left(  \varphi,\vartheta,\gamma\right) \nonumber\\
&  =\sqrt{\frac{2\ell+1}{4\pi}}D_{m,0}^{\left(  \ell\right)  \ast}\left(
\varphi,\vartheta,\gamma\right)  ,\nonumber
\end{align}
where $\gamma$ is arbitrary angle. This form is referred to as passive
convention as well, see \cite{Morrison1987}. It also follows%
\begin{align*}
Y_{\ell}^{m\ast}\left(  \vartheta,\varphi\right)   &  =Y_{\ell}^{m}\left(
\vartheta,-\varphi\right) \\
&  =\left(  -1\right)  ^{m}Y_{\ell}^{-m}\left(  \vartheta,\varphi\right)  ,\\
Y_{\ell}^{-m}\left(  \vartheta,\varphi\right)   &  =\left(  -1\right)
^{m}e^{-i2m\varphi}Y_{\ell}^{m}\left(  \vartheta,\varphi\right)  .
\end{align*}
\newline\cite{Wormer}, (52) in terms of Wigner $3j$-symbols,
\cite{Edmonds1957} (4.3.2)

\item Singly coupled form, Clebsch-Gordan series:%
\begin{align}
D_{m_{1},k_{1}}^{\left(  \ell_{1}\right)  }D_{m_{2},k_{2}}^{\left(  \ell
_{2}\right)  }  &  =\sum_{\ell,m,k}\left(  2\ell+1\right)
\begin{pmatrix}
\ell_{1:2} & \ell\\
m_{1:2} & m
\end{pmatrix}
D_{m,k}^{\left(  \ell\right)  \ast}%
\begin{pmatrix}
\ell_{1:2} & \ell\\
k_{1:2} & k
\end{pmatrix}
\label{CoupledD}\\
&  =\sum_{\ell,m,k}C_{\ell_{1},m_{1};\ell_{2},m_{2}}^{\ell,m}C_{\ell_{1}%
,k_{1};\ell_{2},k_{2}}^{\ell,k}D_{m,k}^{\left(  \ell\right)  },\nonumber
\end{align}%
\[
D_{m_{1},k_{1}}^{\left(  \ell_{1}\right)  }D_{m_{2},k_{2}}^{\left(  \ell
_{2}\right)  }=\sum_{\ell=\left\vert \ell_{1}-\ell_{2}\right\vert }^{\ell
_{1}+\ell_{2}}%
\begin{pmatrix}
\ell_{1:2} & \ell\\
m_{1:2} & m
\end{pmatrix}%
\begin{pmatrix}
\ell_{1:2} & \ell\\
k_{1:2} & k
\end{pmatrix}
\left(  -1\right)  ^{m-k}\left(  2\ell+1\right)  D_{m_{1}+m_{2},k_{1}+k_{2}%
}^{\left(  \ell\right)  }%
\]
$-m=m_{1}+m_{2},-k=k_{1}+k_{2}$, \cite{Edmonds1957}, (4.3.4). Double coupled
form%
\[
\sum_{m_{1},m_{2},k_{1},k_{2}}%
\begin{pmatrix}
\ell_{1:2} & \ell\\
m_{1:2} & m
\end{pmatrix}
D_{m_{1},k_{1}}^{\left(  \ell_{1}\right)  }D_{m_{2},k_{2}}^{\left(  \ell
_{2}\right)  }%
\begin{pmatrix}
\ell_{1:2} & j\\
k_{1:2} & k
\end{pmatrix}
=\frac{\delta_{\ell,j}}{2\ell+1}D_{m,k}^{\left(  \ell\right)  \ast}%
\]
\cite{Edmonds1957}, (4.3.3)%
\begin{equation}
\sum_{m_{1},m_{2},m_{3}}D_{m_{1},k_{1}}^{\left(  \ell_{1}\right)  }%
D_{m_{2},k_{2}}^{\left(  \ell_{2}\right)  }D_{m_{3},k_{3}}^{\left(  \ell
_{3}\right)  }%
\begin{pmatrix}
\ell_{1:3}\\
m_{1:3}%
\end{pmatrix}
=%
\begin{pmatrix}
\ell_{1:3}\\
k_{1:3}%
\end{pmatrix}
, \label{D_3Prod}%
\end{equation}
The integrals
\begin{align}
\int_{SO\left(  3\right)  }D_{m,k}^{\left(  \ell\right)  }dg  &  =\delta
_{\ell,0}\delta_{m,0}\delta_{k,0}\label{Int_D}\\
\mathcal{G}_{k_{1},k_{2};m_{1},m_{2},}^{\ell_{1},\ell_{2}}  &  =\int%
_{SO\left(  3\right)  }D_{m_{1},k_{1}}^{\left(  \ell_{1}\right)  \ast}%
D_{m_{2},k_{2}}^{\left(  \ell_{2}\right)  }dg=\delta_{\ell_{1},\ell_{2}}%
\delta_{m_{1},m_{2}}\delta_{k_{1},k_{2}}\frac{1}{2\ell_{1}+1}\label{Int_D_D}\\
\mathcal{G}_{k_{1},k_{2},k_{3};m_{1},m_{2},m_{3}}^{\ell_{1},\ell_{2},\ell
_{3}}  &  =\int_{SO\left(  3\right)  }D_{m_{1},k_{1}}^{\left(  \ell
_{1}\right)  }D_{m_{2},k_{2}}^{\left(  \ell_{2}\right)  }D_{m_{3},k_{3}%
}^{\left(  \ell_{3}\right)  }dg=%
\begin{pmatrix}
\ell_{1:3}\\
m_{1:3}%
\end{pmatrix}%
\begin{pmatrix}
\ell_{1:3}\\
k_{1:3}%
\end{pmatrix}
\label{Int_D_D_D}%
\end{align}%
\[
\int_{SO\left(  3\right)  }D_{m_{1},k_{1}}^{\left(  \ell_{1}\right)  \ast
}D_{m_{2},k_{2}}^{\left(  \ell_{2}\right)  }D_{m_{3},k_{3}}^{\left(  \ell
_{3}\right)  }dg=\frac{1}{2\ell_{3}+1}C_{\ell_{1},m_{1};\ell_{2},m_{2}}%
^{\ell_{3},m_{3}}C_{\ell_{1},k_{1};\ell_{2},k_{2}}^{\ell_{3},k_{3}}%
\]
where the Haar measure is $dg=\sin\vartheta d\vartheta d\varphi d\gamma
/8\pi^{2}$, \cite{Varshalovich1988}, 4.11.1, \cite{Edmonds1957}, (4.6.2)

\item \textbf{Addition formula} \cite{GradshteynRyzhik}, 8.814,
\cite{Erdelyi11a} vol. 2, 11.4(8) pp. 236, \cite{Mueller1966}.
\begin{equation}
\sum_{m=-\ell}^{\ell}Y_{\ell}^{m\ast}\left(  L_{1}\right)  Y_{\ell}^{m}\left(
L_{2}\right)  =\frac{2\ell+1}{4\pi}P_{\ell}\left(  \cos\vartheta\right)
\label{FormulaAddition}%
\end{equation}
where $\cos\vartheta=L_{1}\cdot L_{2}$. For general dimension $d$
\cite{Yadrenko1983} p.72 provides an addition formula in terms of Gegenbauer
polynomials $C_{\ell}^{\nu}$.
\end{enumerate}

\section{Proofs, Bispectrum\label{Append_Bisp_Trisp}}

We use the following notations $Z_{\ell_{1:3}}^{m_{1:3}}=\left(  Z_{\ell_{1}%
}^{m_{1}},Z_{\ell_{2}}^{m_{2}},Z_{\ell_{3}}^{m_{3}}\right)  $, $%
\begin{pmatrix}
\ell_{1:3}\\
m_{1:3}%
\end{pmatrix}
=%
\begin{pmatrix}
\ell_{1} & \ell_{2} & \ell_{3}\\
m_{1} & m_{2} & m_{3}%
\end{pmatrix}
,$ $B_{3}\left(  \ell_{1:3}\right)  =B_{3}\left(  \ell_{1},\ell_{2},\ell
_{3}\right)  .$

\begin{proof}
[Proof of Lemma]\textbf{\label{Proof_L_Iso3}}\ref{Lemm_Iso3} Let
\[
\operatorname*{Cum}\nolimits_{3}\left(  Z_{\ell_{1:3}}^{m_{1:3}}\right)  =%
\begin{pmatrix}
\ell_{1:3}\\
m_{1:3}%
\end{pmatrix}
B_{3}\left(  \ell_{1:3}\right)  ,
\]
then
\begin{align}
\operatorname*{Cum}\nolimits_{3}\left(  \mathcal{Z}_{\ell_{1:3}}^{k_{1:3}%
}\right)   &  =\sum_{m_{1},m_{2},m_{3}}D_{k_{1},m_{1}}^{\left(  \ell
_{1}\right)  }D_{k_{2},m_{2}}^{\left(  \ell_{2}\right)  }D_{m_{3},k_{3}%
}^{\left(  \ell_{3}\right)  }\operatorname*{Cum}\nolimits_{3}\left(
Z_{\ell_{1:3}}^{m_{1:3}}\right) \label{Cum_3Z}\\
&  =\sum_{m_{1:3}}D_{k_{1},m_{1}}^{\left(  \ell_{1}\right)  }D_{k_{2},m_{2}%
}^{\left(  \ell_{2}\right)  }D_{m_{3},k_{3}}^{\left(  \ell_{3}\right)  }%
\begin{pmatrix}
\ell_{1:3}\\
m_{1:3}%
\end{pmatrix}
B_{3}\left(  \ell_{1:3}\right) \nonumber\\
&  =%
\begin{pmatrix}
\ell_{1:3}\\
k_{1:3}%
\end{pmatrix}
B_{3}\left(  \ell_{1:3}\right) \nonumber\\
&  =\operatorname*{Cum}\nolimits_{3}\left(  Z_{\ell_{1}}^{k_{1}},Z_{\ell_{2}%
}^{k_{2}},Z_{\ell_{3}}^{k_{3}}\right) \nonumber
\end{align}
where $\mathcal{Z}_{\ell_{1:3}}^{k_{1:3}}=\left(  \mathcal{Z}_{\ell_{1}%
}^{k_{1}},\mathcal{Z}_{\ell_{2}}^{k_{2}},\mathcal{Z}_{\ell_{3}}^{k_{3}%
}\right)  $ see (\ref{D_3Prod}). Hence the assumption of isotropy is
fulfilled.\newline If (\ref{BicovZ}) is not assumed then under the assumption
of isotropy we have $\operatorname*{Cum}\nolimits_{3}\left(  \mathcal{Z}%
_{\ell_{1:3}}^{k_{1:3}}\right)  =\operatorname*{Cum}\nolimits_{3}\left(
Z_{\ell_{1:3}}^{k_{1:3}}\right)  $, and integrate both sides of (\ref{Cum_3Z})
according to the Haar measure and obtain
\begin{align*}
\operatorname*{Cum}\nolimits_{3}\left(  Z_{\ell_{1:3}}^{k_{1:3}}\right)   &  =%
\begin{pmatrix}
\ell_{1:3}\\
k_{1:3}%
\end{pmatrix}
\sum_{m_{1:3}}%
\begin{pmatrix}
\ell_{1:3}\\
m_{1:3}%
\end{pmatrix}
\operatorname*{Cum}\nolimits_{3}\left(  Z_{\ell_{1:3}}^{m_{1:3}}\right) \\
&  =%
\begin{pmatrix}
\ell_{1:3}\\
k_{1:3}%
\end{pmatrix}
B_{3}\left(  \ell_{1:3}\right)
\end{align*}
see (\ref{Int_D_D_D}), where
\[
B_{3}\left(  \ell_{1:3}\right)  =\sum_{m_{1:3}}%
\begin{pmatrix}
\ell_{1:3}\\
m_{1:3}%
\end{pmatrix}
\operatorname*{Cum}\nolimits_{3}\left(  Z_{\ell_{1:3}}^{m_{1:3}}\right)  .
\]
Hence (\ref{BicovZ}) is a necessary and sufficient assumption for the third
order isotropy. In this case the bispectrum is a linear combination of the
cumulants fo the angular projections by the probability amplitude of coupling
three angular momenta $\ell_{1:3}$.
\end{proof}

\begin{proof}
[Proof of (\ref{Bisp_sym})]\label{Bisp_sym_degree}Consider the expression%
\[%
\begin{pmatrix}
\ell_{1:3}\\
k_{1:3}%
\end{pmatrix}
B_{3}\left(  \ell_{1:3}\right)  =%
\begin{pmatrix}
\ell_{1:3}\\
k_{1:3}%
\end{pmatrix}
\sum_{m_{1:3}}%
\begin{pmatrix}
\ell_{1:3}\\
m_{1:3}%
\end{pmatrix}
\operatorname*{Cum}\nolimits_{3}\left(  Z_{\ell_{1:3}}^{m_{1:3}}\right)  .
\]
Any permutation of $\ell_{1:3}$ implies the same sign of the Wigner
coefficients hence it does not change the value.
\end{proof}

\begin{proof}
[Proof: 3-product of\textbf{ }spherical harmonics is rotation invariant]%
\textbf{ }Indeed
\begin{align*}
I_{\ell_{1:3}}\left(  gL_{1:3}\right)   &  =\sum_{m_{1:3}}%
\begin{pmatrix}
\ell_{1:3}\\
m_{1:3}%
\end{pmatrix}
\sum_{k_{1:3}}D_{k_{1},m_{1}}^{\left(  \ell_{1}\right)  }D_{k_{2},m_{2}%
}^{\left(  \ell_{2}\right)  }D_{m_{3},k_{3}}^{\left(  \ell_{3}\right)
}Y_{\ell_{1}}^{k_{1}}\left(  L_{1}\right)  Y_{\ell_{2}}^{k_{2}}\left(
L_{2}\right)  Y_{\ell_{3}}^{k_{3}}\left(  L_{3}\right) \\
&  =\sum_{k_{1:3}}%
\begin{pmatrix}
\ell_{1:3}\\
k_{1:3}%
\end{pmatrix}
Y_{\ell_{1}}^{k_{1}}\left(  L_{1}\right)  Y_{\ell_{2}}^{k_{2}}\left(
L_{2}\right)  Y_{\ell_{3}}^{k_{3}}\left(  L_{3}\right)  .
\end{align*}
see (\ref{D_3Prod}).
\end{proof}

%

%

\section{Proofs, Trispectrum\label{Append_Trisp}}

Repeat the notations $\ell_{1:4}=\left(  \ell_{1},\ell_{2},\ell_{3},\ell
_{4}\right)  $, \newline$\operatorname*{Cum}\nolimits_{4}\left(  X\left(
L_{1}\right)  ,X\left(  L_{2}\right)  ,X\left(  L_{3}\right)  ,X\left(
L_{4}\right)  \right)  =\operatorname*{Cum}\nolimits_{4}\left(  X\left(
L_{1:4}\right)  \right)  ,$ $T_{4}\left(  \left.  \ell_{1},\ell_{2},\ell
_{3},\ell_{4}\right\vert \ell^{1}\right)  =T_{4}\left(  \left.  \ell
_{1:4}\right\vert \ell^{1}\right)  $.

\begin{lemma}
\label{Lemma_4D}%
\begin{multline*}
\sum_{m_{1:4}}D_{k_{1},m_{1}}^{\left(  \ell_{1}\right)  }D_{k_{2},m_{2}%
}^{\left(  \ell_{2}\right)  }D_{k_{3},m_{3}}^{\left(  \ell_{3}\right)
}D_{k_{4},m_{4}}^{\left(  \ell_{4}\right)  }\left(  -1\right)  ^{m}%
\begin{pmatrix}
\ell_{1:2} & \ell\\
m_{1:2} & -m
\end{pmatrix}%
\begin{pmatrix}
\ell & \ell_{3:4}\\
m & m_{3:4}%
\end{pmatrix}
\\
=\left(  -1\right)  ^{k}%
\begin{pmatrix}
\ell_{1:2} & \ell\\
k_{1:2} & -k
\end{pmatrix}%
\begin{pmatrix}
\ell_{3:4} & \ell\\
k_{3:4} & k
\end{pmatrix}
.
\end{multline*}

\end{lemma}

\begin{proof}
Indeed, from
\[
D_{k_{3},m_{3}}^{\left(  \ell_{3}\right)  }D_{k_{4},m_{4}}^{\left(  \ell
_{4}\right)  }=\sum_{\ell^{2},k^{2},m^{2}}\left(  2\ell^{2}+1\right)
\begin{pmatrix}
\ell_{3:4} & \ell^{2}\\
m_{3:4} & m^{2}%
\end{pmatrix}%
\begin{pmatrix}
\ell_{3:4} & \ell^{2}\\
k_{3:4} & k^{2}%
\end{pmatrix}
D_{k^{2},m^{2}}^{\left(  \ell^{2}\right)  \ast},
\]
and from $m_{1}+m_{2}+m_{3}+m_{4}=0$, it follows%
\begin{align*}
&  \sum_{m_{1:4}}D_{k_{1},m_{1}}^{\left(  \ell_{1}\right)  }D_{k_{2},m_{2}%
}^{\left(  \ell_{2}\right)  }D_{k_{3},m_{3}}^{\left(  \ell_{3}\right)
}D_{k_{4},m_{4}}^{\left(  \ell_{4}\right)  }\left(  -1\right)  ^{m}%
\begin{pmatrix}
\ell_{1:2} & \ell\\
m_{1:2} & -m
\end{pmatrix}%
\begin{pmatrix}
\ell & \ell_{3:4}\\
m & m_{3:4}%
\end{pmatrix}
\\
&  =\sum_{\ell^{2},k^{2},m^{2}}\sum_{m_{1:4}}D_{k_{1},m_{1}}^{\left(  \ell
_{1}\right)  }D_{k_{2},m_{2}}^{\left(  \ell_{2}\right)  }D_{k^{2},m^{2}%
}^{\left(  \ell^{2}\right)  \ast}\left(  2\ell^{2}+1\right) \\
&  \times\left(  -1\right)  ^{m}%
\begin{pmatrix}
\ell_{3:4} & \ell^{2}\\
m_{3:4} & m^{2}%
\end{pmatrix}%
\begin{pmatrix}
\ell_{3:4} & \ell^{2}\\
k_{3:4} & k^{2}%
\end{pmatrix}%
\begin{pmatrix}
\ell_{1:2} & \ell\\
m_{1:2} & -m
\end{pmatrix}%
\begin{pmatrix}
\ell & \ell_{3:4}\\
m & m_{3:4}%
\end{pmatrix}
\\
&  =\sum_{\ell^{2},k^{2},m^{2}}\sum_{m_{3:4}}%
\begin{pmatrix}
\ell_{3:4} & \ell^{2}\\
m_{3:4} & -m^{2}%
\end{pmatrix}%
\begin{pmatrix}
\ell & \ell_{3:4}\\
m & m_{3:4}%
\end{pmatrix}
\left(  2\ell^{2}+1\right) \\
&  \times\left(  -1\right)  ^{m}\sum_{m_{1:2}}D_{k_{1},m_{1}}^{\left(
\ell_{1}\right)  }D_{k_{2},m_{2}}^{\left(  \ell_{2}\right)  }D_{k^{2},m^{2}%
}^{\left(  \ell^{2}\right)  }\left(  -1\right)  ^{m^{2}-k^{2}}%
\begin{pmatrix}
\ell_{3:4} & \ell^{2}\\
k_{3:4} & -k^{2}%
\end{pmatrix}%
\begin{pmatrix}
\ell_{1:2} & \ell\\
m_{1:2} & -m
\end{pmatrix}
\\
&  =\delta_{m,-m^{2}}\sum_{\ell^{2},k^{2},m^{2}}\sum_{m_{1:2}}\delta
_{\ell,\ell^{2}}D_{k_{1},m_{1}}^{\left(  \ell_{1}\right)  }D_{k_{2},m_{2}%
}^{\left(  \ell_{2}\right)  }D_{-k^{2},m^{2}}^{\left(  \ell^{2}\right)  }%
\begin{pmatrix}
\ell_{1:2} & \ell\\
m_{1:2} & -m^{2}%
\end{pmatrix}
\left(  -1\right)  ^{-k^{2}}%
\begin{pmatrix}
\ell_{3:4} & \ell\\
k_{3:4} & -k^{2}%
\end{pmatrix}
\\
&  =%
\begin{pmatrix}
\ell_{1:2} & \ell\\
k_{1:2} & k^{2}%
\end{pmatrix}
\left(  -1\right)  ^{-k^{2}}%
\begin{pmatrix}
\ell_{3:4} & \ell\\
k_{3:4} & -k^{2}%
\end{pmatrix}
,
\end{align*}
\newline since $3j$-symbols are orthgonal
\[
\left(  2\ell+1\right)  \sum_{m_{1:2}}%
\begin{pmatrix}
\ell_{1:2} & \ell\\
m_{1:2} & -m
\end{pmatrix}%
\begin{pmatrix}
\ell_{1:2} & \ell^{1}\\
m_{1:2} & m^{1}%
\end{pmatrix}
=\delta_{\ell,\ell^{1}}\delta_{m,-m^{1}},
\]
see (\ref{W3j_Orth1}) and (\ref{D_3Prod}) is valid.
\end{proof}

\begin{lemma}
\label{Int_4D}%
\begin{align*}
&  \int_{SO\left(  3\right)  }D_{k_{1},m_{1}}^{\left(  \ell_{1}\right)
}D_{k_{2},m_{2}}^{\left(  \ell_{2}\right)  }D_{k_{3},m_{3}}^{\left(  \ell
_{3}\right)  }D_{k_{4},m_{4}}^{\left(  \ell_{4}\right)  }dg\\
&  =\sum_{\ell}\left(  2\ell+1\right)  \left(  -1\right)  ^{k-m}%
\begin{pmatrix}
\ell_{1:2} & \ell\\
k_{1:2} & -k
\end{pmatrix}%
\begin{pmatrix}
\ell_{3:4} & \ell\\
k_{3:4} & k
\end{pmatrix}%
\begin{pmatrix}
\ell_{1:2} & \ell\\
m_{1:2} & -m
\end{pmatrix}%
\begin{pmatrix}
\ell_{3:4} & \ell\\
m_{3:4} & m
\end{pmatrix}
\end{align*}

\end{lemma}

\begin{proof}
Here we apply the formula (\ref{CoupledD}) for the single coupled D-matrices
then the integral of triple product of D-matrices (\ref{Int_D_D_D}) and get
\begin{align*}
&  \int_{SO\left(  3\right)  }D_{k_{1},m_{1}}^{\left(  \ell_{1}\right)
}D_{k_{2},m_{2}}^{\left(  \ell_{2}\right)  }D_{k_{3},m_{3}}^{\left(  \ell
_{3}\right)  }D_{k_{4},m_{4}}^{\left(  \ell_{4}\right)  }dg\\
&  =\sum_{\ell^{1},k^{1},m^{1}}\left(  2\ell^{1}+1\right)  \left(  -1\right)
^{k^{1}-m^{1}}%
\begin{pmatrix}
\ell_{3:4} & \ell^{1}\\
k_{3:4} & k^{1}%
\end{pmatrix}%
\begin{pmatrix}
\ell_{3:4} & \ell^{1}\\
m_{3:4} & m^{1}%
\end{pmatrix}
\int_{SO\left(  3\right)  }D_{k_{1},m_{1}}^{\left(  \ell_{1}\right)  }%
D_{k_{2},m_{2}}^{\left(  \ell_{2}\right)  }D_{-k^{1},-m^{1}}^{\left(  \ell
^{1}\right)  }dg\\
&  =\sum_{\ell^{1},k^{1},m^{1}}%
\begin{pmatrix}
\ell_{1:2} & \ell^{1}\\
k_{1:2} & -k^{1}%
\end{pmatrix}%
\begin{pmatrix}
\ell_{3:4} & \ell^{1}\\
k_{3:4} & k^{1}%
\end{pmatrix}
\\
&  \times\left(  2\ell^{1}+1\right)  \left(  -1\right)  ^{k^{1}-m^{1}}%
\begin{pmatrix}
\ell_{1:2} & \ell^{1}\\
m_{1:2} & -m^{1}%
\end{pmatrix}%
\begin{pmatrix}
\ell_{3:4} & \ell^{1}\\
m_{3:4} & m^{1}%
\end{pmatrix}
.
\end{align*}

\end{proof}

\begin{lemma}
\label{Lemma_sym_L4}If $\ell_{1}+\ell_{2}+\ell_{3}+\ell_{4}$ is even then
\[
\sum_{\ell}\left(  2\ell+1\right)  \left(  -1\right)  ^{k-m}%
\begin{pmatrix}
\ell_{1:2} & \ell\\
k_{1:2} & -k
\end{pmatrix}%
\begin{pmatrix}
\ell & \ell_{3:4}\\
k & k_{3:4}%
\end{pmatrix}%
\begin{pmatrix}
\ell_{1:2} & \ell\\
m_{1:2} & -m
\end{pmatrix}%
\begin{pmatrix}
\ell & \ell_{3:4}\\
m & m_{3:4}%
\end{pmatrix}
\]
is symmetric in $\ell_{1:4}$.
\end{lemma}

\begin{proof}
Start with
\begin{multline*}%
\begin{pmatrix}
\ell_{1:2} & \ell\\
k_{1:2} & -k
\end{pmatrix}%
\begin{pmatrix}
\ell & \ell_{3:4}\\
k & k_{3:4}%
\end{pmatrix}
\\
=\sum_{\ell_{0}}\left(  2\ell_{0}+1\right)  \left(  -1\right)  ^{\Sigma
+\ell_{0}-\ell-k_{1}-k_{3}}%
\begin{pmatrix}
\ell_{2} & \ell_{3} & \ell_{0}\\
k_{2} & k_{3} & -k_{0}%
\end{pmatrix}%
\begin{pmatrix}
\ell_{1} & \ell_{0} & \ell_{4}\\
k_{1} & k_{0} & k_{4}%
\end{pmatrix}%
\begin{Bmatrix}
\ell_{1} & \ell_{2} & \ell\\
\ell_{3} & \ell_{4} & \ell_{0}%
\end{Bmatrix}
\end{multline*}
where $\Sigma=\ell_{1}+\ell_{2}+\ell_{3}+\ell_{4}$. Now we evaluate the
product
\begin{align*}
&  \sum_{\ell_{0}^{1}}\left(  2\ell_{0}^{1}+1\right)  \left(  -1\right)
^{\Sigma+\ell_{0}^{1}-\ell-k_{1}-k_{3}}%
\begin{pmatrix}
\ell_{2} & \ell_{3} & \ell_{0}^{1}\\
k_{2} & k_{3} & -k_{0}^{1}%
\end{pmatrix}%
\begin{pmatrix}
\ell_{1} & \ell_{0}^{1} & \ell_{4}\\
k_{1} & k_{0}^{1} & k_{4}%
\end{pmatrix}%
\begin{Bmatrix}
\ell_{1} & \ell_{2} & \ell\\
\ell_{3} & \ell_{4} & \ell_{0}^{1}%
\end{Bmatrix}
\\
&  \times\sum_{\ell_{0}^{2}}\left(  2\ell_{0}^{2}+1\right)  \left(  -1\right)
^{\Sigma+\ell_{0}^{2}-\ell-m_{1}-m_{3}}%
\begin{pmatrix}
\ell_{2} & \ell_{3} & \ell_{0}^{2}\\
m_{2} & m_{3} & -m_{0}^{2}%
\end{pmatrix}%
\begin{pmatrix}
\ell_{1} & \ell_{0}^{2} & \ell_{4}\\
m_{1} & m_{0}^{2} & m_{4}%
\end{pmatrix}%
\begin{Bmatrix}
\ell_{1} & \ell_{2} & \ell\\
\ell_{3} & \ell_{4} & \ell_{0}^{2}%
\end{Bmatrix}
.
\end{align*}
The values
\[
\sqrt{\left(  2\ell+1\right)  \left(  2\ell_{0}^{1}+1\right)  }%
\begin{Bmatrix}
\ell_{1} & \ell_{2} & \ell\\
\ell_{3} & \ell_{4} & \ell_{0}^{1}%
\end{Bmatrix}
\]
form a real orthogonal matrix, see \cite{Edmonds1957} (6.2.10), rows and
columns being labelled by $\ell$ and $\ell_{0}^{1}$. Hence if we sum it up by
$\ell$, the result is
\begin{align*}
&  \sum_{\ell_{0}^{1}}\left(  2\ell_{0}^{1}+1\right)  \left(  -1\right)
^{-k_{0}^{1}}%
\begin{pmatrix}
\ell_{2} & \ell_{3} & \ell_{0}^{1}\\
k_{2} & k_{3} & -k_{0}^{1}%
\end{pmatrix}%
\begin{pmatrix}
\ell_{1} & \ell_{0}^{1} & \ell_{4}\\
k_{1} & k_{0}^{1} & k_{4}%
\end{pmatrix}
\\
&  \times\left(  -1\right)  ^{-m_{0}^{1}}%
\begin{pmatrix}
\ell_{2} & \ell_{3} & \ell_{0}^{1}\\
m_{2} & m_{3} & -m_{0}^{1}%
\end{pmatrix}%
\begin{pmatrix}
\ell_{1} & \ell_{0}^{1} & \ell_{4}\\
m_{1} & m_{0}^{1} & m_{4}%
\end{pmatrix}
,
\end{align*}
q.e.d.
\end{proof}

\begin{proof}
[Proof of Lemma \ref{Lemm_Iso4}]We have
\[
\operatorname*{Cum}\nolimits_{4}\left(  \mathcal{Z}_{\ell_{1:4}}^{m_{1:4}%
}\right)  =\sum_{k_{1:4}}D_{m_{1},k_{1}}^{\left(  \ell_{1}\right)  }%
D_{m_{2},k_{2}}^{\left(  \ell_{2}\right)  }D_{m_{3},k_{3}}^{\left(  \ell
_{3}\right)  }D_{m_{4},k_{4}}^{\left(  \ell_{4}\right)  }\operatorname*{Cum}%
\nolimits_{4}\left(  Z_{\ell_{1:4}}^{k_{1:4}}\right)  .
\]
Under the assumption of isotropy $\operatorname*{Cum}\nolimits_{4}\left(
\mathcal{Z}_{\ell_{1:4}}^{m_{1:4}}\right)  =\operatorname*{Cum}\nolimits_{4}%
\left(  Z_{\ell_{1:4}}^{m_{1:4}}\right)  $, now integrate both sides by the
Haar measure, see Lemma \ref{Lemma_IntD_p} for $p=4$, and obtain
\begin{multline*}
\operatorname*{Cum}\nolimits_{4}\left(  Z_{\ell_{1:4}}^{m_{1:4}}\right)
=\sum_{k_{1:4}}\sum_{\ell^{1},k^{1},m^{1}}%
\begin{pmatrix}
\ell_{1:2} & \ell^{1}\\
m_{1:2} & -m^{1}%
\end{pmatrix}%
\begin{pmatrix}
\ell^{1} & \ell_{3:4}\\
m^{1} & m_{3:4}%
\end{pmatrix}
\left(  -1\right)  ^{m^{1}-k^{1}}\\
\times\left(  2\ell^{1}+1\right)
\begin{pmatrix}
\ell_{1:2} & \ell^{1}\\
k_{1:2} & -k^{1}%
\end{pmatrix}%
\begin{pmatrix}
\ell^{1} & \ell_{3:4}\\
k^{1} & k_{3:4}%
\end{pmatrix}
\operatorname*{Cum}\nolimits_{4}\left(  Z_{\ell_{1:4}}^{k_{1:4}}\right)
\end{multline*}
Note that each term according to summation $k_{1:4}$ is\textit{\ symmetric} in
$\ell_{1:4}$, see Lemma \ref{Lemma_sym_L4} above. Now, define%
\[
T_{4}\left(  \left.  \ell_{1:4}\right\vert \ell^{1}\right)  =\sum
_{k^{1},k_{1:4}}\left(  -1\right)  ^{k^{1}}\sqrt{2\ell^{1}+1}%
\begin{pmatrix}
\ell_{1:2} & \ell^{1}\\
k_{1:2} & -k^{1}%
\end{pmatrix}%
\begin{pmatrix}
\ell^{1} & \ell_{3:4}\\
k^{1} & k_{3:4}%
\end{pmatrix}
\operatorname*{Cum}\nolimits_{4}\left(  Z_{\ell_{1:4}}^{k_{1:4}}\right)  ,
\]
with this notation we have the cumulant in the form
\[
\operatorname*{Cum}\nolimits_{4}\left(  Z_{\ell_{1:4}}^{m_{1:4}}\right)
=\sum_{\ell^{1},m^{1}}%
\begin{pmatrix}
\ell_{1:2} & \ell^{1}\\
m_{1:2} & -m^{1}%
\end{pmatrix}%
\begin{pmatrix}
\ell^{1} & \ell_{3:4}\\
m^{1} & m_{3:4}%
\end{pmatrix}
\left(  -1\right)  ^{m^{1}}\sqrt{2\ell^{1}+1}T_{4}\left(  \left.  \ell
_{1:4}\right\vert \ell^{1}\right)  .
\]
If instead of isotropy (\ref{Cum4_Z}) is assumed then
\begin{align*}
\operatorname*{Cum}\nolimits_{4}\left(  \mathcal{Z}_{\ell_{1:4}}^{m_{1:4}%
}\right)   &  =\sum_{k_{1:4}}D_{m_{1},k_{1}}^{\left(  \ell_{1}\right)
}D_{m_{2},k_{2}}^{\left(  \ell_{2}\right)  }D_{m_{3},k_{3}}^{\left(  \ell
_{3}\right)  }D_{m_{4},k_{4}}^{\left(  \ell_{4}\right)  }\operatorname*{Cum}%
\nolimits_{4}\left(  Z_{\ell_{1:4}}^{k_{1:4}}\right) \\
&  =\sum_{k_{1:4}}D_{m_{1},k_{1}}^{\left(  \ell_{1}\right)  }D_{m_{2},k_{2}%
}^{\left(  \ell_{2}\right)  }D_{m_{3},k_{3}}^{\left(  \ell_{3}\right)
}D_{m_{4},k_{4}}^{\left(  \ell_{4}\right)  }\\
&  \times\sum_{\ell}%
\begin{pmatrix}
\ell_{1:2} & \ell\\
k_{1:2} & -k
\end{pmatrix}%
\begin{pmatrix}
\ell & \ell_{3:4}\\
k & k_{3:4}%
\end{pmatrix}
\left(  -1\right)  ^{k}\sqrt{2\ell+1}T_{4}\left(  \left.  \ell_{1:4}%
\right\vert \ell\right) \\
&  =\sum_{\ell}T_{4}\left(  \left.  \ell_{1:4}\right\vert \ell\right)
\sqrt{2\ell+1}\\
&  \times\sum_{k_{1:4}}D_{m_{1},k_{1}}^{\left(  \ell_{1}\right)  }%
D_{m_{2},k_{2}}^{\left(  \ell_{2}\right)  }D_{m_{3},k_{3}}^{\left(  \ell
_{3}\right)  }D_{m_{4},k_{4}}^{\left(  \ell_{4}\right)  }\left(  -1\right)
^{k}%
\begin{pmatrix}
\ell_{1:2} & \ell\\
k_{1:2} & -k
\end{pmatrix}%
\begin{pmatrix}
\ell & \ell_{3:4}\\
k & k_{3:4}%
\end{pmatrix}
,
\end{align*}
the Lemma \ref{Lemma_4D} can be applied and we get%
\begin{align*}
\operatorname*{Cum}\nolimits_{4}\left(  \mathcal{Z}_{\ell_{1:4}}^{m_{1:4}%
}\right)   &  =\sum_{m,\ell}%
\begin{pmatrix}
\ell_{1:2} & \ell\\
m_{1:2} & -m
\end{pmatrix}%
\begin{pmatrix}
\ell_{3:4} & \ell\\
m_{3:4} & m
\end{pmatrix}
\left(  -1\right)  ^{m}\sqrt{2\ell+1}T_{4}\left(  \left.  \ell_{1:4}%
\right\vert \ell\right) \\
&  =\operatorname*{Cum}\nolimits_{4}\left(  Z_{\ell_{1:4}}^{m_{1:4}}\right)  .
\end{align*}

\end{proof}

Particular cases
\begin{align*}
\operatorname*{Cum}\nolimits_{4}\left(  u_{\ell_{1}}\left(  L\right)
,u_{\ell_{2}}\left(  L\right)  ,u_{\ell_{3}}\left(  L\right)  ,u_{\ell_{4}%
}\left(  L\right)  \right)   &  =\sum_{m,\ell}T_{4}\left(  \left.  \ell
_{1},\ell_{2},\ell_{3},\ell_{4}\right\vert \ell\right)  \sqrt{2\ell+1}%
\sum_{m_{1:4}}\prod\limits_{j=1}^{4}Y_{\ell_{j}}^{m_{j}}\left(  L_{j}\right)
\\
&  \times%
\begin{pmatrix}
\ell_{1} & \ell_{2} & \ell\\
m_{1} & m_{2} & -m
\end{pmatrix}%
\begin{pmatrix}
\ell & \ell_{3} & \ell_{4}\\
m & m_{3} & 0
\end{pmatrix}
,
\end{align*}%
\begin{multline*}
\operatorname*{Cum}\nolimits_{4}\left(  u_{\ell_{1}}\left(  L_{1}\right)
,u_{\ell_{2}}\left(  L_{2}\right)  ,u_{\ell_{3}}\left(  L\right)  ,u_{\ell
_{4}}\left(  L\right)  \right) \\
=\operatorname*{Cum}\nolimits_{4}\left(  u_{\ell_{1}}\left(  g_{LL_{2}}%
L_{1}\right)  ,u_{\ell_{2}}\left(  g_{LL_{2}}L_{2}\right)  ,u_{\ell_{3}%
}\left(  N\right)  ,u_{\ell_{4}}\left(  N\right)  \right) \\
=\sqrt{\prod_{j=3:4}\frac{2\ell_{j}+1}{4\pi}}\sum_{\ell}T_{4}\left(  \left.
\ell_{1},\ell_{2},\ell_{3},\ell_{4}\right\vert \ell\right)
\begin{pmatrix}
\ell_{3} & \ell_{4} & \ell\\
0 & 0 & 0
\end{pmatrix}
\sqrt{2\ell+1}\\
\times\sum_{m_{1},m_{2}}Y_{\ell_{1}}^{m_{1}}\left(  g_{LL_{2}}L_{1}\right)
Y_{\ell_{2}}^{m_{2}}\left(  g_{LL_{2}}L_{2}\right)
\begin{pmatrix}
\ell_{1:2} & \ell\\
m_{1:2} & 0
\end{pmatrix}
\\
=\frac{\sqrt{2}}{4\pi}\sqrt{\prod_{j=3:4}\frac{2\ell_{j}+1}{4\pi}}\sum_{\ell}%
\begin{pmatrix}
\ell_{3} & \ell_{4} & \ell\\
0 & 0 & 0
\end{pmatrix}
T_{4}\left(  \left.  \ell_{1},\ell_{2},\ell_{3},\ell_{4}\right\vert
\ell\right)  \mathcal{I}_{\ell_{1},\ell_{2},\ell}\left(  \vartheta_{2}%
,\varphi_{2},\vartheta_{1}\right)  ,
\end{multline*}%
\begin{multline*}
\operatorname*{Cum}\nolimits_{4}\left(  u_{\ell_{1}}\left(  L_{1}\right)
,u_{\ell_{2}}\left(  L\right)  ,u_{\ell_{3}}\left(  L\right)  ,u_{\ell_{4}%
}\left(  L\right)  \right)  =\operatorname*{Cum}\nolimits_{4}\left(
u_{\ell_{1}}\left(  g_{LL_{1}}L_{1}\right)  ,u_{\ell_{2}}\left(  N\right)
,u_{\ell_{3}}\left(  N\right)  ,u_{\ell_{4}}\left(  N\right)  \right) \\
=\sqrt{\prod_{j=2:4}\frac{2\ell_{j}+1}{4\pi}}\sum_{\ell}T_{4}\left(  \left.
\ell_{1},\ell_{2},\ell_{3},\ell_{4}\right\vert \ell\right)  \sqrt{2\ell
+1}Y_{\ell_{1}}^{0}\left(  g_{LL_{1}}L_{1}\right)
\begin{pmatrix}
\ell_{1} & \ell_{2} & \ell\\
0 & 0 & 0
\end{pmatrix}%
\begin{pmatrix}
\ell & \ell_{3} & \ell_{4}\\
0 & 0 & 0
\end{pmatrix}
\\
=\sqrt{\prod_{j=1:4}\frac{2\ell_{j}+1}{4\pi}}\sum_{\ell}T_{4}\left(  \left.
\ell_{1},\ell_{2},\ell_{3},\ell_{4}\right\vert \ell\right)  \sqrt{2\ell
+1}P_{\ell_{1}}\left(  L\cdot L_{1}\right)
\begin{pmatrix}
\ell_{1} & \ell_{2} & \ell\\
0 & 0 & 0
\end{pmatrix}%
\begin{pmatrix}
\ell & \ell_{3} & \ell_{4}\\
0 & 0 & 0
\end{pmatrix}
,
\end{multline*}%
\begin{multline*}
\operatorname*{Cum}\nolimits_{4}\left(  u_{\ell_{1}}\left(  L\right)
,u_{\ell_{2}}\left(  L\right)  ,u_{\ell_{3}}\left(  L\right)  ,u_{\ell_{4}%
}\left(  L\right)  \right)  =\operatorname*{Cum}\nolimits_{4}\left(
u_{\ell_{1}}\left(  N\right)  ,u_{\ell_{2}}\left(  N\right)  ,u_{\ell_{3}%
}\left(  N\right)  ,u_{\ell_{4}}\left(  N\right)  \right) \\
=\sqrt{\prod_{j=1:4}\frac{2\ell_{i}+1}{4\pi}}\sum_{\ell}T_{4}\left(  \left.
\ell_{1},\ell_{2},\ell_{3},\ell_{4}\right\vert \ell\right)  \sqrt{2\ell+1}%
\begin{pmatrix}
\ell_{1} & \ell_{2} & \ell\\
0 & 0 & 0
\end{pmatrix}%
\begin{pmatrix}
\ell & \ell_{3} & \ell_{4}\\
0 & 0 & 0
\end{pmatrix}
.
\end{multline*}

\section{Proofs, Polyspectrum\label{Append_PolySp}}

Notation: $\mathcal{G}_{k_{1:p},m_{1:p}}^{\ell_{1:p}}=\mathcal{G}_{k_{1}%
,k_{2},\ldots,k_{p},m_{1},m_{2},\ldots,m_{p}}^{\ell_{1},\ell_{2},\ldots
,\ell_{3p}}$.

\begin{lemma}
\label{Lemma_IntD_p}Define $\ell^{0}=\ell_{1}$, $\ell^{p-2}=\ell_{p}$,
$k^{0}=-k_{1}$, $k^{p-2}=k_{p}$, $m^{0}=-m_{1}$, $m^{p-2}=m_{p}$, then for
$p>3$,
\begin{multline*}
\mathcal{G}_{k_{1:p},m_{1:p}}^{\ell_{1:p}}=\int_{SO\left(  3\right)  }%
\prod\limits_{a=1}^{p}D_{k_{a},m_{a}}^{\left(  \ell_{a}\right)  }dg=\sum
_{\ell^{a},m^{a},k^{a}}\left(  -1\right)  ^{\Sigma\left(  m^{1:p-3}%
-k^{1:p-3}\right)  }\prod\limits_{a=0}^{p-3}%
\begin{pmatrix}
\ell^{a} & \ell_{a+2} & \ell^{a+1}\\
-k^{a} & k_{a+2} & k^{a+1}%
\end{pmatrix}
\\
\times\prod\limits_{a=0}^{p-3}%
\begin{pmatrix}
\ell^{a} & \ell_{a+2} & \ell^{a+1}\\
-m^{a} & m_{a+2} & m^{a+1}%
\end{pmatrix}
\prod\limits_{a=1}^{p-3}\left(  2\ell^{a}+1\right)  ,
\end{multline*}

\end{lemma}

\begin{proof}
The fourth order key formula is emphasized here for understanding the
induction
\begin{align*}
&  \mathcal{G}_{k_{1:5};m_{1:5}}^{\ell_{1:5}}=\int_{SO\left(  3\right)
}D_{k_{1},m_{1}}^{\left(  \ell_{1}\right)  }D_{k_{2},m_{2}}^{\left(  \ell
_{2}\right)  }D_{k_{3},m_{3}}^{\left(  \ell_{3}\right)  }D_{k_{4},m_{4}%
}^{\left(  \ell_{4}\right)  }D_{k_{5},m_{5}}^{\left(  \ell_{5}\right)  }dg\\
&  =\sum_{\ell^{2},k^{2},m^{2}}\int_{SO\left(  3\right)  }D_{k_{1},m_{1}%
}^{\left(  \ell_{1}\right)  }D_{k_{2},m_{2}}^{\left(  \ell_{2}\right)
}D_{k_{3},m_{3}}^{\left(  \ell_{3}\right)  }D_{-k^{2},-m^{2}}^{\left(
\ell^{2}\right)  }dg\\
&  \times\left(  -1\right)  ^{m^{2}-k^{2}}\left(  2\ell^{2}+1\right)
\begin{pmatrix}
\ell_{4:5} & \ell^{2}\\
k_{4:5} & k^{2}%
\end{pmatrix}%
\begin{pmatrix}
\ell_{4:5} & \ell^{2}\\
m_{4:5} & m^{2}%
\end{pmatrix}
\\
&  =\sum_{\ell^{1:2},k^{1:2},m^{1:2}}%
\begin{pmatrix}
\ell_{1:2} & \ell^{1}\\
k_{1:2} & k^{1}%
\end{pmatrix}%
\begin{pmatrix}
\ell_{3} & \ell^{2} & \ell^{1}\\
k_{3} & k^{2} & -k^{1}%
\end{pmatrix}%
\begin{pmatrix}
\ell_{4:5} & \ell^{2}\\
k_{4:5} & -k^{2}%
\end{pmatrix}
\left(  -1\right)  ^{m^{1}-k^{1}}\\
&  \times\left(  2\ell^{1}+1\right)
\begin{pmatrix}
\ell_{1:2} & \ell^{1}\\
m_{1:2} & m^{1}%
\end{pmatrix}%
\begin{pmatrix}
\ell_{3} & \ell^{2} & \ell^{1}\\
m_{3} & m^{2} & -m^{1}%
\end{pmatrix}%
\begin{pmatrix}
\ell_{4:5} & \ell^{2}\\
m_{4:5} & -m^{2}%
\end{pmatrix}
\\
&  \times\left(  -1\right)  ^{m^{2}-k^{2}}\left(  2\ell^{2}+1\right) \\
&  =\sum_{\ell^{1:2},k^{1:2},m^{1:2}}%
\begin{pmatrix}
\ell_{1:2} & \ell^{1}\\
k_{1:2} & k^{1}%
\end{pmatrix}%
\begin{pmatrix}
\ell^{1} & \ell_{3} & \ell^{2}\\
-k^{1} & k_{3} & k^{2}%
\end{pmatrix}%
\begin{pmatrix}
\ell^{2} & \ell_{4:5}\\
-k^{2} & k_{4:5}%
\end{pmatrix}
\left(  -1\right)  ^{m^{1}-k^{1}}\left(  2\ell^{1}+1\right) \\
&  \times%
\begin{pmatrix}
\ell_{1:2} & \ell^{1}\\
m_{1:2} & m^{1}%
\end{pmatrix}%
\begin{pmatrix}
\ell^{1} & \ell_{3} & \ell^{2}\\
-m^{1} & m_{3} & m^{2}%
\end{pmatrix}%
\begin{pmatrix}
\ell^{2} & \ell_{4:5}\\
-m^{2} & m_{4:5}%
\end{pmatrix}
\times\left(  -1\right)  ^{m^{2}-k^{2}}\left(  2\ell^{2}+1\right)  .
\end{align*}
Define $\ell^{p-2}=\ell_{p}$, $k^{p-2}=k_{p}$, $m^{p-2}=m_{p}$
\[
\int_{SO\left(  3\right)  }\prod\limits_{a=1}^{p}D_{k_{a},m_{a}}^{\left(
\ell_{a}\right)  }dg=\left(  -1\right)  ^{m_{p}}\sum_{\ell^{a},m^{a},k^{a}%
}\prod\limits_{a=0}^{p-3}C_{\ell^{a},k^{a};\ell_{a+2},k_{a+2}}^{\ell
^{a+1},k^{a+1}}C_{\ell^{a},m^{a};\ell_{a+2},m_{a+2}}^{\ell^{a+1},m^{a+1}}%
\]%
\begin{multline*}
\int_{SO\left(  3\right)  }\prod\limits_{a=1}^{p}D_{k_{a},m_{a}}^{\left(
\ell_{a}\right)  }dg=\sum_{\ell^{a},m^{a},k^{a}}\left(  -1\right)
^{\Sigma\left(  m^{1:p-3}-k^{1:p-3}\right)  }%
\begin{pmatrix}
\ell_{1:2} & \ell^{1}\\
k_{1:2} & k^{1}%
\end{pmatrix}
\prod\limits_{a=1}^{p-3}%
\begin{pmatrix}
\ell^{a} & \ell_{a+2} & \ell^{a+1}\\
-k^{a} & k_{a+2} & k^{a+1}%
\end{pmatrix}
\\
\times%
\begin{pmatrix}
\ell_{1:2} & \ell^{1}\\
m_{1:2} & m^{1}%
\end{pmatrix}
\prod\limits_{a=1}^{p-3}%
\begin{pmatrix}
\ell^{a} & \ell_{a+2} & \ell^{a+1}\\
-m^{a} & m_{a+2} & m^{a+1}%
\end{pmatrix}
\left(  2\ell^{a}+1\right)
\end{multline*}

\end{proof}

\begin{lemma}
\label{Lemma_SumD_p} Define $\ell^{0}=\ell_{1}$, $\ell^{p-2}=\ell_{p}$,
$k^{0}=-k_{1}$, $k^{p-2}=k_{p}$, $m^{0}=-m_{1}$, $m^{p-2}=m_{p}$, then for
$p<3$ and for any $\ell^{1:p-3}=\left(  \ell^{1},\ell^{2},\ldots,\ell
^{p-3}\right)  $ we have
\begin{multline*}
\mathcal{H}_{k_{1:p},m_{1:p}}^{\ell_{1:p}}\left(  \ell^{1:p-3}\right)
=\sum_{m_{1:p}}\prod\limits_{a=1}^{p}D_{k_{a},m_{a}}^{\left(  \ell_{a}\right)
}\left(  -1\right)  ^{\Sigma m^{1:p-3}}\prod\limits_{a=0}^{p-3}%
\begin{pmatrix}
\ell^{a} & \ell_{a+2} & \ell^{a+1}\\
-m^{a} & m_{a+2} & m^{a+1}%
\end{pmatrix}
\\
=\left(  -1\right)  ^{\Sigma k^{1:p-3}}\prod\limits_{a=0}^{p-3}%
\begin{pmatrix}
\ell^{a} & \ell_{a+2} & \ell^{a+1}\\
-k^{a} & k_{a+2} & k^{a+1}%
\end{pmatrix}
.
\end{multline*}

\end{lemma}

\begin{proof}
Induction from $p=4$ to $p=5$ will show the general tratment. We prove that
$\mathcal{H}_{k_{1:5},m_{1:5}}^{\ell_{1:5}}\left(  \ell^{1:2}\right)  $
\begin{multline*}
\mathcal{H}_{k_{1:5},m_{1:5}}^{\ell_{1:5}}\left(  \ell^{1:2}\right) \\
=\sum_{m_{1:5}}D_{k_{1},m_{1}}^{\left(  \ell_{1}\right)  }D_{k_{2},m_{2}%
}^{\left(  \ell_{2}\right)  }D_{k_{3},m_{3}}^{\left(  \ell_{3}\right)
}D_{k_{4},m_{4}}^{\left(  \ell_{4}\right)  }D_{k_{5},m_{5}}^{\left(  \ell
_{5}\right)  }\left(  -1\right)  ^{m^{1}+m^{2}}%
\begin{pmatrix}
\ell_{1:2} & \ell^{1}\\
m_{1:2} & m^{1}%
\end{pmatrix}%
\begin{pmatrix}
\ell^{1} & \ell_{3} & \ell^{2}\\
-m^{1} & m_{3} & m^{2}%
\end{pmatrix}%
\begin{pmatrix}
\ell^{2} & \ell_{4:5}\\
-m^{2} & m_{4:5}%
\end{pmatrix}
\\
=\left(  -1\right)  ^{k^{1}+k^{2}}%
\begin{pmatrix}
\ell_{1:2} & \ell^{1}\\
k_{1:2} & k^{1}%
\end{pmatrix}%
\begin{pmatrix}
\ell^{1} & \ell_{3} & \ell^{2}\\
-k^{1} & k_{3} & k^{2}%
\end{pmatrix}%
\begin{pmatrix}
\ell^{2} & \ell_{4:5}\\
-k^{2} & k_{4:5}%
\end{pmatrix}
.
\end{multline*}
Indeed, we have
\[
D_{k_{4},m_{4}}^{\left(  \ell_{4}\right)  }D_{k_{5},m_{5}}^{\left(  \ell
_{5}\right)  }=\sum_{\ell^{2},k^{2},m^{2}}\left(  2\ell^{2}+1\right)
\begin{pmatrix}
\ell_{4:5} & \ell^{2}\\
m_{4:5} & m^{2}%
\end{pmatrix}%
\begin{pmatrix}
\ell_{4:5} & \ell^{2}\\
k_{4:5} & k^{2}%
\end{pmatrix}
D_{k^{2},m^{2}}^{\left(  \ell^{2}\right)  \ast}%
\]
$m_{1}+m_{2}+m_{3}+m_{4}=0$%
\begin{align*}
&  \mathcal{H}_{k_{1:5},m_{1:5}}^{\ell_{1:5}}\left(  \ell^{1:2}\right)
=\sum_{\ell^{2},k^{2},m^{2}}\sum_{m_{1:5}}D_{k_{1},m_{1}}^{\left(  \ell
_{1}\right)  }D_{k_{2},m_{2}}^{\left(  \ell_{2}\right)  }D_{k_{3},m_{3}%
}^{\left(  \ell_{3}\right)  }D_{k^{2},m^{2}}^{\left(  \ell^{2}\right)  \ast
}\left(  -1\right)  ^{m^{1}+m^{2}}\left(  2\ell^{2}+1\right) \\
&  \times%
\begin{pmatrix}
\ell_{4:5} & \ell^{2}\\
m_{4:5} & m^{2}%
\end{pmatrix}%
\begin{pmatrix}
\ell_{4:5} & \ell^{2}\\
k_{4:5} & k^{2}%
\end{pmatrix}%
\begin{pmatrix}
\ell_{1:2} & \ell^{1}\\
m_{1:2} & m^{1}%
\end{pmatrix}%
\begin{pmatrix}
\ell^{1} & \ell_{3} & \ell^{2}\\
-m^{1} & m_{3} & m^{2}%
\end{pmatrix}%
\begin{pmatrix}
\ell^{2} & \ell_{4:5}\\
-m^{2} & m_{4:5}%
\end{pmatrix}
,
\end{align*}
now by (\ref{D_conjugate})%
\begin{align*}
\mathcal{H}_{k_{1:5},m_{1:5}}^{\ell_{1:5}}\left(  \ell^{1:2}\right)   &
=\sum_{\ell^{2},k^{2},m^{2}}%
\begin{pmatrix}
\ell_{4:5} & \ell^{2}\\
k_{4:5} & k^{2}%
\end{pmatrix}
\left(  -1\right)  ^{k^{2}}\left(  2\ell^{2}+1\right)  \sum_{m_{4:5}}%
\begin{pmatrix}
\ell^{2} & \ell_{4:5}\\
-m^{2} & m_{4:5}%
\end{pmatrix}%
\begin{pmatrix}
\ell_{4:5} & \ell^{2}\\
m_{4:5} & m^{2}%
\end{pmatrix}
\\
&  \times\sum_{m_{1:3}}D_{k_{1},m_{1}}^{\left(  \ell_{1}\right)  }%
D_{k_{2},m_{2}}^{\left(  \ell_{2}\right)  }D_{k_{3},m_{3}}^{\left(  \ell
_{3}\right)  }D_{-k,-m}^{\left(  \ell\right)  }\left(  -1\right)  ^{m^{1}}%
\begin{pmatrix}
\ell_{1:2} & \ell^{1}\\
m_{1:2} & m^{1}%
\end{pmatrix}%
\begin{pmatrix}
\ell^{1} & \ell_{3} & \ell^{2}\\
-m^{1} & m_{3} & m^{2}%
\end{pmatrix}
.
\end{align*}
$3j$-symbols are orthogonal, see (\ref{W3j_Orth2}), \textbf{ }using
(\ref{D_3Prod}) we obtain%
\begin{align*}
\mathcal{H}_{k_{1:5},m_{1:5}}^{\ell_{1:5}}\left(  \ell^{1:2}\right)   &  =%
\begin{pmatrix}
\ell_{4:5} & \ell^{2}\\
k_{4:5} & k^{2}%
\end{pmatrix}
\left(  -1\right)  ^{k^{2}}\\
&  \times\sum_{m_{1:3}}D_{k_{1},m_{1}}^{\left(  \ell_{1}\right)  }%
D_{k_{2},m_{2}}^{\left(  \ell_{2}\right)  }D_{k_{3},m_{3}}^{\left(  \ell
_{3}\right)  }D_{-k^{2},-m^{2}}^{\left(  \ell\right)  }\left(  -1\right)
^{m^{1}}%
\begin{pmatrix}
\ell_{1:2} & \ell^{1}\\
m_{1:2} & m^{1}%
\end{pmatrix}%
\begin{pmatrix}
\ell^{1} & \ell_{3} & \ell^{2}\\
-m^{1} & m_{3} & -m^{2}%
\end{pmatrix}
\\
&  =\left(  -1\right)  ^{k^{1}+k^{2}}%
\begin{pmatrix}
\ell_{1:2} & \ell^{1}\\
k_{1:2} & k^{1}%
\end{pmatrix}%
\begin{pmatrix}
\ell^{1} & \ell_{3} & \ell^{2}\\
-k^{1} & k_{3} & k^{2}%
\end{pmatrix}%
\begin{pmatrix}
\ell^{2} & \ell_{3:4}\\
-k^{2} & k_{3:4}%
\end{pmatrix}
,
\end{align*}
obsereve that $k^{1}$ and $k^{2}$ are given by $k_{1:2}$ and $k_{3:4}$
respectively hence the signe of them can be chosen freely.
\end{proof}

\begin{proof}
[Proof of Lemma]\ref{Lemm_Iso_p} We have
\[
\operatorname*{Cum}\nolimits_{p}\left(  \mathcal{Z}_{\ell_{1:p}}^{m_{1:p}%
}\right)  =\sum_{k_{1:p}}\prod\limits_{a=1}^{p}D_{m_{a},k_{a}}^{\left(
\ell_{a}\right)  }\operatorname*{Cum}\nolimits_{4}\left(  Z_{\ell_{1:p}%
}^{k_{1:p}}\right)  .
\]
Under isotropy assumption
\begin{align*}
\operatorname*{Cum}\nolimits_{p}\left(  Z_{\ell_{1:p}}^{m_{1:p}}\right)   &
=\sum_{k_{1:p}}\int_{SO\left(  3\right)  }\prod\limits_{a=1}^{p}D_{m_{a}%
,k_{a}}^{\left(  \ell_{a}\right)  }dg\operatorname*{Cum}\nolimits_{p}\left(
Z_{\ell_{1:p}}^{k_{1:p}}\right) \\
&  =\sum_{\ell^{1:p-3},k^{1:p-3}}\left(  -1\right)  ^{\Sigma m^{1:p-3}}%
\prod\limits_{a=0}^{p-3}%
\begin{pmatrix}
\ell^{a} & \ell_{a+2} & \ell^{a+1}\\
-m^{a} & m_{a+2} & m^{a+1}%
\end{pmatrix}
\\
&  \times\sum_{k_{1:p}}\left(  -1\right)  ^{\Sigma k^{1:p-3}}\prod
\limits_{a=0}^{p-3}%
\begin{pmatrix}
\ell^{a} & \ell_{a+2} & \ell^{a+1}\\
-k^{a} & k_{a+2} & k^{a+1}%
\end{pmatrix}
\left(  2\ell^{a+1}+1\right)  \operatorname*{Cum}\nolimits_{p}\left(
Z_{\ell_{1:p}}^{k_{1:p}}\right) \\
&  =\sum_{\ell^{1:p-3},k^{1:p-3}}\left(  -1\right)  ^{\Sigma m^{1:p-3}}%
\prod\limits_{a=0}^{p-3}\sqrt{2\ell^{a}+1}%
\begin{pmatrix}
\ell^{a} & \ell_{a+2} & \ell^{a+1}\\
-m^{a} & m_{a+2} & m^{a+1}%
\end{pmatrix}
\widetilde{S}_{p}\left(  \left.  \ell_{1:p}\right\vert \ell^{1:p-3}\right)  .
\end{align*}
If isotropy is not assumed then from
\begin{align*}
\operatorname*{Cum}\nolimits_{p}\left(  Z_{\ell_{1:p}}^{k_{1:p}}\right)   &
=\sum_{\ell^{1:p-3},k^{1:p-3}}\left(  -1\right)  ^{\Sigma k^{1:p-3}}%
\prod\limits_{a=0}^{p-3}%
\begin{pmatrix}
\ell^{a} & \ell_{a+2} & \ell^{a+1}\\
-k^{a} & k_{a+2} & k^{a+1}%
\end{pmatrix}
\\
&  \times\prod\limits_{a=1}^{p-3}\sqrt{2\ell^{a}+1}\widetilde{S}_{p}\left(
\left.  \ell_{1:p}\right\vert \ell^{1:p-3}\right)  ,
\end{align*}
we obtain%
\begin{align*}
\operatorname*{Cum}\nolimits_{p}\left(  \mathcal{Z}_{\ell_{1:p}}^{m_{1:p}%
}\right)   &  =\sum_{k_{1:p}}\prod\limits_{a=1}^{p}D_{m_{a},k_{a}}^{\left(
\ell_{a}\right)  }\left(  -1\right)  ^{\Sigma k^{1:p-3}}\prod\limits_{a=0}%
^{p-3}%
\begin{pmatrix}
\ell^{a} & \ell_{a+2} & \ell^{a+1}\\
-k^{a} & k_{a+2} & k^{a+1}%
\end{pmatrix}
\\
&  \times\prod\limits_{a=1}^{p-3}\sqrt{2\ell^{a}+1}\widetilde{S}_{p}\left(
\left.  \ell_{1:p}\right\vert \ell^{1:p-3}\right) \\
&  =\left(  -1\right)  ^{\Sigma m^{1:p-3}}\prod\limits_{a=0}^{p-3}%
\begin{pmatrix}
\ell^{a} & \ell_{a+2} & \ell^{a+1}\\
-m^{a} & m_{a+2} & m^{a+1}%
\end{pmatrix}
\\
&  \times\prod\limits_{a=1}^{p-3}\sqrt{2\ell^{a}+1}\widetilde{S}_{p}\left(
\left.  \ell_{1:p}\right\vert \ell^{1:p-3}\right)  ,
\end{align*}
see Lemma \ref{Lemma_SumD_p} .
\end{proof}

\begin{acknowledgement}
The publication was supported by the T\'{A}MOP-4.2.2.C-11/1/KONV-2012-0001
project. The project has been supported by the European Union, co-financed by
the European Social Fund.
\end{acknowledgement}

\bibliographystyle{alpha}
\bibliography{00BiblMM13}

\end{document}